\documentclass[aps, pre, twocolumn, floatfix, superscriptaddress, longbibliography, nofootinbib]{revtex4-2}
\usepackage{amsmath, amsfonts, amssymb, amsthm}
\usepackage[usenames,dvipsnames,svgnames,table]{xcolor}
\usepackage{hyperref}
\hypersetup{%
    colorlinks=true,
    linkcolor=blue,
    filecolor=magenta,      
    urlcolor=red,
    citecolor=blue,
}

\usepackage[
  color=orange!40,
  textsize=singlespacetiny,
]{todonotes}

\definecolor{green}{HTML}{4daf4a}
\definecolor{purple}{HTML}{984ea3}

\usepackage{graphicx}
\usepackage{tikz}
\usetikzlibrary{shapes, shapes.geometric, arrows}
\usepackage{pgfplots}
\usepgfplotslibrary{dateplot}
\usepackage{pgffor}
\pgfplotsset{compat=1.10}
\usepgfplotslibrary{fillbetween}

\usepackage{booktabs}
\usepackage{makecell}
\usepackage{multirow}
\usepackage{siunitx}

\newtheorem{defn}{Definition}

\newtheorem{prop}[defn]{Proposition}
\newtheorem{lemma}[defn]{Lemma}

\usepackage{enumitem}

\newcommand{\lmat}{\begin{bmatrix}}
\newcommand{\rmat}{\end{bmatrix}}
\newcommand{\RR}{\mathbb{R}}
\newcommand{\XX}{\mathbb{X}}
\newcommand{\ZZ}{\mathbb{Z}}
\newcommand{\calD}{\mathcal{D}}

\newcommand{\calF}{\mathcal{F}}
\newcommand{\calH}{\mathcal{H}}
\newcommand{\calL}{\mathcal{L}}
\newcommand{\calO}{\mathcal{O}}
\newcommand{\calS}{\mathcal{S}}
\newcommand{\calX}{\mathcal{X}}
\newcommand{\scb}{\mathrm{SCB}}
\newcommand{\checkb}[1]{\check{#1}_{\mathrm{B}}}
\newcommand{\checkw}[1]{\check{#1}_{\mathrm{W}}}
\newcommand{\eps}{\varepsilon}
\newcommand{\degC}{\operatorname{deg}}
\newcommand{\degR}{\operatorname{deg}_{\RR}}
\newcommand{\IFS}{\operatorname{I}_{\mathrm{FS}}}

\begin{document}

\title{A group-equivariant autoencoder for identifying spontaneously broken symmetries}

\author{Devanshu Agrawal}
\affiliation{Department of Industrial and Systems Engineering, University of Tennessee, Knoxville, TN 37996, USA\looseness=-1}
\author{Adrian Del Maestro}
\affiliation{Department of Physics and Astronomy, University of Tennessee, Knoxville, TN 37996, USA}
\affiliation{Min H. Kao Department of Electrical Engineering and Computer Science,
University of Tennessee, Knoxville, TN 37996, USA\looseness=-1}
\affiliation{Institute for Advanced Materials and Manufacturing, University of Tennessee, Knoxville, Tennessee 37996, USA\looseness=-1}
\author{Steven Johnston}
\affiliation{Department of Physics and Astronomy, University of Tennessee, Knoxville, TN 37996, USA}
\affiliation{Institute for Advanced Materials and Manufacturing, University of Tennessee, Knoxville, Tennessee 37996, USA\looseness=-1}
\author{James Ostrowski}
\affiliation{Department of Industrial and Systems Engineering, University of Tennessee, Knoxville, TN 37996, USA\looseness=-1}

\begin{abstract}
We introduce the group-equivariant autoencoder (GE-autoencoder) -- a deep neural network (DNN) method that locates phase boundaries by determining which symmetries of the Hamiltonian have spontaneously broken at each temperature.  
We use group theory to deduce which symmetries of the system remain intact in all phases, and then use this information to constrain the parameters of the GE-autoencoder such that the encoder learns an order parameter invariant to these ``never-broken'' symmetries.  
This procedure produces a dramatic reduction in the number of free parameters such that the GE-autoencoder size is independent of the system size.  
We include symmetry regularization terms in the loss function of the GE-autoencoder so that the learned order parameter is also equivariant to the remaining symmetries of the system.  
By examining the group representation by which the learned order parameter transforms, we are then able to extract information about the associated spontaneous symmetry breaking.  
We test the GE-autoencoder on the 2D classical ferromagnetic and antiferromagnetic Ising models, finding that the GE-autoencoder 
(1) accurately determines which symmetries have spontaneously broken at each temperature; 
(2) estimates the critical temperature in the thermodynamic limit with greater accuracy, robustness, and time-efficiency than a symmetry-agnostic baseline autoencoder; and 
(3) detects the presence of an external symmetry-breaking magnetic field with greater sensitivity than the baseline method. 
Finally, we describe various key implementation details, including a new method for extracting the critical temperature estimate from trained autoencoders and calculations of the DNN initialization and learning rate settings required for fair model comparisons.
\end{abstract}

\maketitle

\section{Introduction}

Mapping a material's phase diagram is an important endeavor in condensed-matter physics and materials science~\cite{ashcroft1976solid, friedli2017statistical}. 
This is a crucial step towards practical applications as such diagrams can act as a road map to manipulating a material's  functionality. 
From a theoretical perspective, predicting the locations of phase boundaries for a given material can provide important insights into the microscopic physics that govern its behavior and can provide crucial validation of proposed low-energy models. 

In some cases (e.g., structural transitions), phase transitions are governed by high-energy processes and can be reasonably predicted with efficient computational methods \cite{Gomez2019Review} like density functional theory \cite{DFT1, DFT2, DFT3, Shahi:2018dft, Maurer:2019dft} or molecular dynamics simulations \cite{haile1992molecular, Sasaki:2020md, Shanavas2009review}. 
However, there are also numerous examples of phase transitions between novel states of matter that are completely governed by low-energy properties, which can be much harder to predict. 
Notable examples include the Mott and other metal-to-insulator transitions~\cite{Imada1998mit}, unconventional superconductivity~\cite{Johnston2010Review, Stewart2017Review, Keimer2015Review}, and quantum magnetism and spin liquid behavior~\cite{Zhou2017Review, Savary2016spinliquid}. 
For example, predicting the low-energy properties of strongly correlated systems has proven to be extremely challenging, even with the advent of state-of-the-art computational algorithms and the widespread availability of high-performance computing. 
A case in point is the single band Hubbard model, whose doping-temperature phase diagram is rich with many competing and intertwined orders~\cite{Keimer2015Review, Fradkin2015Review} and may or may not contain a superconducting ground state~\cite{Maier2005superconductivity, Zheng2017absence, Jiang2019stripes}.

In the Landau paradigm, a phase transition is characterized by an order parameter -- 
a measurable quantity encoding some macroscopic property of the system that undergoes a discontinuous change at a critical point (e.g., a critical temperature, critical pressure, etc.). 
The change in the order parameter is tied to an associated symmetry breaking, where the order parameter is zero in the high-symmetry disordered phase and nonzero in the low-symmetry ordered phase. 
Identifying when a phase transition occurs thus requires knowledge of an appropriate order parameter or the corresponding symmetry.

Physical intuition or experimental input can provide insight towards the identification of the correct order parameter and its relevant symmetry. 
However, there are well known examples of order parameters that are nonlocal or exist in a more abstract space. 
Notable examples include the Haldane transition in spin-$1$ antiferromagnetic chains~\cite{Haldane1983, Kennedy1992}, the breaking of gauge symmetry across the superconducting transition~\cite{Anderson1963}, or the emergence of topological order in the quantum Hall states~\cite{Wen:1990ws}. 
In cases like these, there is no general method for identifying order parameters and their associated symmetries. 
There are also materials where cross-over behavior is observed that may or may not be associated with a true phase transition. 
Perhaps the most famous example of this is the pseudogap ``phase'' of the high-T$_c$ cuprates~\cite{Kivelson2019}. 
Thus, it would be very advantageous to have a general method of identifying an order parameter, detecting sudden changes in its value, and determining the corresponding broken symmetry across the transition. 
Our goal here is to introduce such a method.

In recent years, techniques from machine learning---in particular, deep neural networks (DNNs)~\cite{goodfellow2016deep}---have been used to successfully identify phase transitions in both classical and quantum many-body lattice systems in a purely data-driven manner~\cite{carrasquilla2017machine, broecker2017machine, ch2017machine, wetzel2017machine,Morningstar:2018dl,Efthymiou:2019uy,Walker:2020hi,JohnstonReview}. 
DNNs are complex parametric models consisting of an alternating composition of linear and nonlinear transformations; 
such models now constitute the state-of-the-art for a variety of problems in domains such as computer vision and natural language understanding~\cite{krizhevsky2012imagenet, antipov2015learned, liang2017text}. 
To date, most applications using DNNs to detect phase transitions have focused on Monte Carlo (MC) simulations of lattice models, which is natural given that large volumes of training and validation data can be easily generated~\cite{JohnstonReview}.


More recently, methods from unsupervised learning have been applied to the problem of identifying phase transitions~\cite{Wang2016unsupervised, wetzel2017unsupervised, ch2018unsupervised, alexandrou2020critical, yevick2021variational}. 
Unsupervised learning is the paradigm used to find structure in unlabeled data, such as its intrinsic dimensionality. 
Perhaps the most well-known methods for dimensionality reduction are principal components analysis (PCA) and the autoencoder, 
where the latter is a DNN with an encoder-decoder architecture that may be thought of as a nonlinear generalization of PCA~\cite{hinton2006reducing, kingma2014auto}.

The main contribution of this paper is a new DNN method for identifying phase transitions, which we call the group-equivariant autoencoder (GE-autoencoder). 
In contrast to previous methods cited above, the GE-autoencoder is specifically designed to identify which symmetries of a given system are broken at each point in a region of its phase diagram; 
the identification of the corresponding phase transition is thus a corollary. 
In this way, the GE-autoencoder not only locates phase transitions but gives insight into its mechanism via the associated spontaneous symmetry breaking (SSB). 
Our method only assumes that (1) we have knowledge of the symmetry group $G$ of the system Hamiltonian and that (2) we have selected a latent dimensionality for the GE-autoencoder; 
the key steps of the GE-autoencoder method are then the following:
\begin{enumerate}
\item Use group theory to deduce the subgroup $G_{\mathrm{NB}}$ of ``never-broken symmetries''-- i.e., the symmetries in $G$ that remain in tact in all phases of the system. 
\item Constrain the GE-autoencoder such that it learns a $G_{\mathrm{NB}}$-invariant order parameter. 
\item Train the GE-autoencoder using ``symmetry regularization'' such that it learns a $G$-equivariant order parameter.
\end{enumerate}
During training, the GE-autoencoder learns the representation of $G$ by which the order parameter transforms, and from this we can extract information about the associated SSB.

The advantage of the GE-autoencoder over previous symmetry-agnostic DNN methods is three-fold. 
First, the GE-autoencoder exploits knowledge about the symmetries of the system that would otherwise be wasted. 
The point of using ML for identifying phase transitions is that it does not require us to have knowledge of the relevant order parameter, but this does not mean we should forget the knowledge we may have-- such as symmetries of the high-energy microscopic Hamiltonian. 
Second, thanks to the never-broken symmetries constraining the GE-autoencoder as well as training with symmetry regularization, we expect the GE-autoencoder to locate phase transitions with greater accuracy, efficiency, and robustness than symmetry-agnostic methods. 
Third and finally, as already mentioned, the GE-autoencoder not only identifies phase transitions but provides information about the associated SSB, thereby elucidating its mechanism. 
Having access to details on broken symmetries provides information on how to couple to the order parameter via a conjugate field-- a requirement for probing associated phase transitions in the laboratory.

Since the GE-autoencoder is a new method, we focus on the details of the methodology in this paper and test it as a proof-of-principle on the 2D ferromagnetic and antiferromagnetic Ising models. 
Moreover, we discuss numerous implementation details throughout the paper that were essential for obtaining conclusive results. 
The paper is organized as follows: 
In Sec.~\ref{sec:background}, we review as background the mathematical notion of SSB, the Ising model, and autoencoders. 
In Sec.~\ref{sec:methods}, we describe the GE-autoencoder method in detail, focusing on the case of a 1D order observable for clarity; 
we also describe the experimental setup, including calculations of the DNN initialization and learning rate settings required for fair experimental comparisons. 
In Sec.~\ref{sec:results}, we present our experimental results; 
we find  that the GE-autoencoder accurately identifies which symmetries are broken at each temperature and estimates the critical temperature with greater accuracy, time-efficiency, and robustness than a baseline autoencoder. 
We concurrently give additional details of the data analysis, including a new method for extracting stable critical temperature estimates from statistics of trained autoencoder models. 
In Sec.~\ref{sec:vector}, we extend the GE-autoencoder to support arbitrary finite symmetry groups and vector-valued order observables, paving the way for future applications. 
Finally, in Sec.~\ref{sec:discussion}, we conclude the paper with a discussion of its key findings, implications, and directions for future work.%
\footnote{Code to reproduce all results in this paper can be found at \url{https://github.com/dagrawa2/ssb_detection_ising}. Permanent link: \url{https://doi.org/10.5281/zenodo.6055507}.}

\section{Background}
\label{sec:background}

\subsection{Spontaneous symmetry breaking}
\label{sec:ssb}

In this section, we review the concept of spontaneous symmetry breaking (SSB) from a mathematical perspective, which will help us formulate the method. This  
discussion is based on the one given in Ref.~[\citenum{georgii2011gibbs}]. 

Consider a classical many-body system on a lattice whose size is parameterized by $L$ (e.g., for a hypercubic lattice, $L$ is the size of one dimension). Let $\XX$ denote the space of all lattice configurations, 
and suppose the system Hamiltonian is invariant under the action of a group $G$ on $\XX$. An \emph{equilibrium state} is then a distribution of lattice configurations that maximizes the entropy subject to a fixed expected internal energy. In other words, the equilibrium state solves a constrained convex optimization problem. 
For finite systems, the equilibrium state is uniquely the well-known Boltzmann distribution over lattice configurations. 
In the thermodynamic limit ($L\rightarrow\infty$), on the other hand, uniqueness is no longer necessary, and in general we have a polyhedral solution set $\calS$ of equilibrium states. An abrupt change in the structure (e.g., dimensionality) of this set $\calS$ with respect to temperature or any other tuning parameter in the Hamiltonian is called a \emph{phase transition}. 
Here we restrict ourselves to disorder-order phase transitions, where the equilibrium state changes from being unique (disordered) to not unique (ordered).

When viewed in this framework, every symmetry (element) in $G$ sends an equilibrium state to an equilibrium state. Thus, in the disordered phase, the unique equilibrium state is itself $G$-invariant. However, in the ordered phase, the equilibrium states are no longer necessarily $G$-invariant, as they may permute under the action of an element of $G$. This phenomenon is often called \emph{spontaneous symmetry breaking} (SSB).

The abstract polyhedral set $\calS$ of equilibrium states is made concrete by way of an \emph{order parameter} -- a linear embedding of $\calS$ 
into a Euclidean space of dimension $d = \operatorname{dim}(\calS)$, with the centroid of the polyhedron mapped to the origin. 
By the Riesz-Markov-Kakutani Representation Theorem, there exists an \emph{order observable} $\calO:\XX\mapsto\RR^d$ such that the order parameter sends each equilibrium state $\nu$ to the expectation
\[ \langle\calO\rangle_{\nu} = \int_{\XX}\calO(x)\,\mathrm{d}\nu(x). \]
Without loss of generality, $\calO$ can be chosen such that (1) $\langle\calO\rangle_{\nu}=0$ in the disordered phase and (2) it is  $G$-equivariant; 
by $G$-equivariant, we mean $\calO(gx) = \psi_g\calO(x)$ for all $g\in G,x\in\XX$, where $\psi$ is a nontrivial real-orthogonal representation of $G$. 
Every symmetry $g\in G$ for which $\psi_g= 1$ is then said to be \emph{never-broken} and is otherwise \emph{broken} in the ordered phase.

Markov Chain Monte Carlo (MCMC) simulations of lattice systems are designed to converge to the average equilibrium state $\overline{\nu}$ (the centroid of the set $\calS$), which is $G$-invariant and satisfies $\langle\calO\rangle_{\overline{\nu}}=0$ in both the disordered and ordered phases. 
Hence, the order parameter as defined above cannot be used to distinguish the two phases. 
It is possible, however, to define a general expression in terms of the observable $\calO$ that can. 
For simplicity, for most of this work we focus our exposition on scalar order parameters ($d=1$); 
for the case of higher-dimensional order parameters ($d\geq 2$), see Sec.~\ref{sec:vector}. 
In the scalar case, the expected absolute value $\langle |\calO|\rangle_{\overline{\nu}}$ is sufficient to distinguish the two phases, 
taking a value of zero in the disordered phase (in the thermodynamic limit) and a nonzero value in the ordered phase. 
Although a misnomer, we will refer to $\langle |\calO|\rangle_{\overline{\nu}}$ throughout this paper as the order parameter and will drop the subscript $\overline{\nu}$. 
In the limited case of $d=1$, the representation $\psi$ takes values in $\{-1, 1\}$ with $\psi_g=-1$ for at least one $g\in G$ since the representation must be nontrivial.

\subsection{The Ising model}
\label{sec:ising}

One of the simplest and most well-studied lattice systems is the classical Ising model in 2D; it is both rich enough to exhibit a second-order phase transition while also admitting an exact solution~\cite{onsager1944crystal}. 
We consider the Ising model on a square $L\times L$ lattice with $L$ even and periodic boundary conditions. 
A lattice configuration is obtained by assigning to each lattice site a classical spin $x_{\bf i}=\pm 1$, where ${\bf i} = (i_x,i_y)$ are the spatial indices of the site. The space of all lattice configurations is $\XX = \{-1, 1\}^{L\times L}$. 
The Ising Hamiltonian is
\begin{equation}\label{Eq:HIsing}
\calH(\mathbf{x}) = -J\sum_{\langle {\bf i},{\bf j}\rangle} x_{\bf i} x_{\bf j},
\end{equation}
where $\mathbf{x}$ is an $L\times L$ matrix with entries $x_{\bf i}\in\{-1, 1\}$, $J$ is the coupling constant, and the sum is taken over all pairs of neighboring lattice sites. We set $J = \pm 1$, where $J=1$ (resp.~$J=-1$) corresponds to a ferromagnetic (resp.~antiferromagnetic) magnetic interaction. 

Equation \eqref{Eq:HIsing} has both spatial (translations, reflections, and orthogonal rotations) and spin-flip ($\mathbf{x}\rightarrow -\mathbf{x}$) internal symmetries. 
The Ising symmetry group $G$ admits a presentation with independent generators $\alpha$, $\rho$, $\tau$, and $\sigma$~(see Fig.~\ref{fig:generators}). Picturing the Ising lattice as a matrix, $\alpha$ can be interpreted as a downward (cyclic) translation; 
$\rho$ is a $90^\circ$-counterclockwise rotation about the origin; 
$\tau$ is a reflection about the vertical line of symmetry; 
and $\sigma$ is the spin-flip internal symmetry (not shown in Fig.~\ref{fig:generators}). 
Every symmetry operation of the Ising model can be expressed in terms of the four generators of $G$.  For example, a rightward translation can be expressed as $\rho\alpha\rho^{-1}$ while a reflection about the diagonal can be expressed as 
$\tau\rho$, as shown in Fig.~\ref{fig:generators}.  
Algebraicly, this presentation of $G$ is defined by the following relations:
\begin{align*}
\alpha^L &= \rho^4 = \tau^2 = \sigma^2 = 1 \\
\rho\tau &= \tau\rho^3 \\
\alpha\rho^2 &= \rho^2\alpha^{-1} \\
\alpha\tau &= \tau\alpha \\
g\sigma &= \sigma g\ \forall g\in G.
\end{align*}

\begin{figure}
\centering
\begin{tikzpicture}
\scriptsize
\begin{axis}[
axis lines=center,
axis line style={->},
axis equal image, 
xmin=-1.2, xmax=1, 
ymin=-1, ymax=1.2, 
xtick={}, ytick={},
xticklabels={}, yticklabels={},
disabledatascaling,
hide x axis, hide y axis]
\foreach \x in {-1, -0.75, -0.5, -0.25, 0, 0.25, 0.5, 0.75, 1} {%
\addplot[black] coordinates {(\x, -1.0) (\x, 1.0)};
\addplot[black] coordinates {(-1.0, \x) (1.0, \x)};
}
\addplot[green, ->] coordinates {(-1.1, 1.0) (-1.1, 0.75)}
node[anchor=north, pos=1, text=green] {\Large $\alpha$};
\draw[red, ->] (axis cs:{cos(110)}, {sin(110)}) arc[radius=1, start angle=110, end angle=170];
\draw[red, ->] (axis cs:{cos(-80)}, {sin(-80)}) arc[radius=1, start angle=-80, end angle=-10]
node[anchor=west, pos=0.5, text=red] {\Large $\rho$};
\addplot[blue, dashed, <->] coordinates {(-0.02, -1.0) (-0.02, 1.0)};
\draw[blue, <->] (axis cs:-0.1, -0.3) arc[radius={sqrt(0.1)}, start angle={-atan(0.1/0.3)-90}, end angle={atan(0.1/0.3)-90}]
node[anchor=north, pos=0.25, text=blue] {\Large $\tau$};
\end{axis}
\node[above, font=\small\bfseries] at (current bounding box.north) {Generators};
\end{tikzpicture}%
\qquad %
\begin{tikzpicture}
\scriptsize
\begin{axis}[
axis lines=center,
axis line style={->},
axis equal image, 
xmin=-1.2, xmax=1, 
ymin=-1, ymax=1.2, 
xtick={}, ytick={},
xticklabels={}, yticklabels={},
disabledatascaling,
hide x axis, hide y axis]
\foreach \x in {-1, -0.75, -0.5, -0.25, 0, 0.25, 0.5, 0.75, 1} {%
\addplot[black] coordinates {(\x, -1.0) (\x, 1.0)};
\addplot[black] coordinates {(-1.0, \x) (1.0, \x)};
}

\addplot[white, ->] coordinates {(-1.1, 1.0) (-1.1, 0.75)}
node[anchor=north, pos=1, text=white] {\Large $\alpha$};

\addplot[brown, ->] coordinates {(-1.0, 1.1) (-0.75, 1.1)}
node[anchor=west, pos=1] {\Large $\textcolor{red}{\rho}\textcolor{green}{\alpha}\textcolor{red}{\rho^{-1}}$};
\addplot[magenta, dashed, <->] coordinates {(-1.0, 0.02) (1.0, 0.02)};
\draw[magenta, <->] (axis cs:-0.3, 0.1) arc[radius={sqrt(0.1)}, start angle={180-atan(0.1/0.3)}, end angle={180+atan(0.1/0.3)}]
node[anchor=east, pos=-0.1] {\Large $\textcolor{red}{\rho}\textcolor{blue}{\tau}\textcolor{red}{\rho^{-1}}$};
\addplot[violet, dashed, <->] coordinates {(-1.0, -1.0) (1.0, 1.0)};
\draw[violet, <->] (axis cs:0.282843, 0.141421) arc[radius={sqrt(0.1)}, start angle={45-atan(0.1/0.3)}, end angle={45+atan(0.1/0.3)}]
node[anchor=west, pos=0.25] {\Large $\textcolor{blue}{\tau}\textcolor{red}{\rho}$};
\end{axis}
\node[above, font=\small\bfseries] at (current bounding box.north) {Compounds};
\end{tikzpicture}
\caption{\label{fig:generators} %
Generators (top) of the spatial Ising symmetries acting on a square lattice, 
and some example compound spatial symmetries (bottom) that can be formed from the generators.
}
\end{figure}
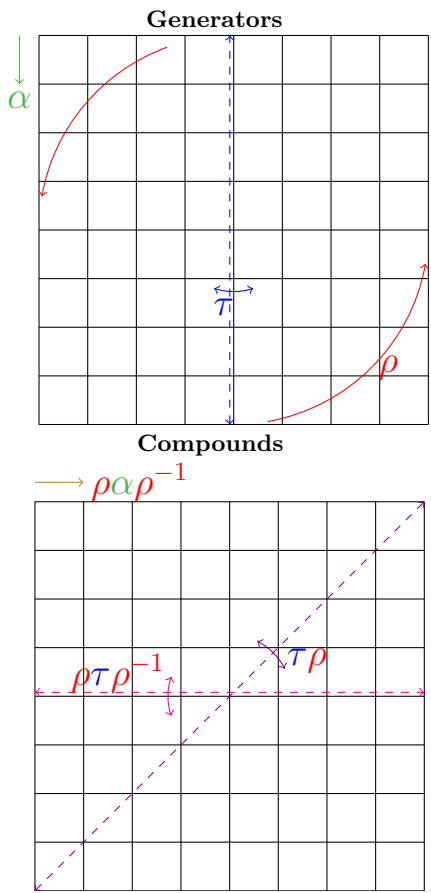

In both the ferromagnetic and antiferromagnetic cases, the 2D Ising model exhibits a second-order phase transition at the critical temperature $T_{\mathrm{c}}/J = \frac{\sqrt{2}}{\log(1+\sqrt{2})}=2.269\ldots$~\cite{onsager1944crystal}. 
The phase transition is associated with the spontaneous breaking of the $\mathbf{x}\mapsto -\mathbf{x}$ internal symmetry. 
The associated order parameter is the expected absolute value $\langle |M|\rangle$ (resp.~$\langle |M_{\mathrm{stag}}|\rangle$) of the magnetization $M$ (resp.~staggered magnetization $M_{\mathrm{stag}}$) in the ferromagnetic (resp.~antiferromagnetic) case, 
where
\begin{align}
M(\mathbf{x}) &= \frac{1}{L^2}\sum_{\bf i} x_{\bf i} \label{eq:magnetization},~\mathrm{and} \\
M_{\mathrm{stag}}(\mathbf{x}) &= \frac{1}{L^2}\left(\sum_{i_x+i_y \mbox{ even}}x_{\bf i} - \sum_{i_x+i_y \mbox{ odd}}x_{\bf i}\right). \label{eq:staggered_magnetization}
\end{align}
Note that both $M$ and $M_{\mathrm{stag}}$ are equivariant functions with respect to the spatial and internal Ising symmetries.

\subsection{Autoencoders}
\label{sec:autoencoders}

The core of our method is the autoencoder, a DNN architecture used for various unsupervised learning tasks~\cite{hinton2006reducing, kingma2014auto}, which we use for dimensionality reduction or ``compression.'' Given a dataset $\{\mathbf{x}_n\in\RR^m\}_{n=1}^N$, it is a common assumption in the traditional domains of computer vision and natural language understanding that the data points lie on a low-dimensional manifold embedded in $\RR^m$. The autoencoder is a means to discovering this intrinsic manifold structure. 
An autoencoder consists of a pair of DNNs --- an encoder $\calO:\RR^m\mapsto\RR^d$ (which will ultimately represent an observable in our application) and a decoder $\calD:\RR^d\mapsto\RR^m$, 
where $d < m$ is the assumed dimensionality of the intrinsic data manifold~(Fig.~\ref{fig:autoencoder}). 
The encoder thus maps its input to a low-dimensional ``latent'' or ``compressed'' representation in terms of intrinsic coordinates on the manifold, 
and the decoder attempts to reconstruct the original input given the latent representation by learning the embedding of the manifold into $\RR^m$. 
The autoencoder is trained by minimizing the reconstruction loss
\[ \calL(\calO, \calD) = \frac{1}{N}\sum_{n=1}^N L_{\mathrm{metric}}(\calD(\calO(\mathbf{x}_n)), \mathbf{x}_n), \]
where $\calL(\calO, \calD)$ means that $\calL$ is a function of the network parameters of $\calO$ and $\calD$, and $L_{\mathrm{metric}}$ is some metric (such as mean square error or binary cross-entropy) that measures the difference between the reconstructed and original inputs. 
Once trained, the encoder can be used to obtain low-dimensional ``summaries'' of the data. Below we describe how this aspect can be used for identifying phase transitions.

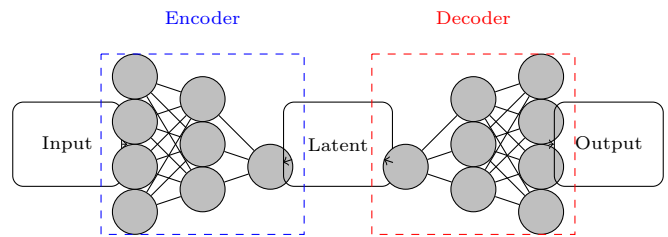
\begin{figure}
\centering
\def\layersep{1.5cm}
\begin{tikzpicture}[node distance=\layersep, auto, scale=0.6]
\scriptsize
\tikzstyle{neuron}=[circle, draw, fill=black!25, minimum size=17pt, inner sep=0pt];
\tikzstyle{block} = [rectangle, draw, 
text width=4.5em, text centered, rounded corners, minimum height=4em]
\tikzstyle{annot} = [text width=4em, text centered]

\node [block] (input) at (-\layersep, 2.5) {Input};
\node[draw=none] (input-right) at (0, 2.5) {};
\path[->] (input) edge (input-right);

\foreach \y in {1,...,4}
\node[neuron] (EI-\y) at (0, \y) {};
\foreach \y in {1,...,3}
\node[neuron] (EH-\y) at (\layersep, \y.5) {};
\foreach \y in {1,...,1}
\node[neuron] (EO-\y) at (2*\layersep, \y+1) {};
\foreach \source in {1,...,4}
\foreach \dest in {1,...,3}
\path (EI-\source) edge (EH-\dest);
\foreach \source in {1,...,3}
\foreach \dest in {1,...,1}
\path (EH-\source) edge (EO-\dest);

\node [block] (latent) at (3*\layersep, 2.5) {Latent};
\path[->] (EO-1) edge (latent);

\foreach \y in {1,...,1}
\node[neuron] (DI-\y) at (4*\layersep, \y+1) {};
\foreach \y in {1,...,3}
\node[neuron] (DH-\y) at (5*\layersep, \y.5) {};
\foreach \y in {1,...,4}
\node[neuron] (DO-\y) at (6*\layersep, \y) {};
\foreach \source in {1,...,1}
\foreach \dest in {1,...,3}
\path (DI-\source) edge (DH-\dest);
\foreach \source in {1,...,3}
\foreach \dest in {1,...,4}
\path (DH-\source) edge (DO-\dest);

\path[->] (latent) edge (DI-1);

\node [block] (output) at (7*\layersep, 2.5) {Output};
\node[draw=none] (output-left) at (6*\layersep, 2.5) {};
\path[->] (output-left) edge (output);

\draw[blue, dashed] (-0.5*\layersep, 4.5)--(2.5*\layersep, 4.5)--(2.5*\layersep, 0.5)--(-0.5*\layersep, 0.5)--cycle;
\node[annot, anchor=south, text=blue] at (\layersep, 5.0) {Encoder};

\draw[red, dashed] (3.5*\layersep, 4.5)--(6.5*\layersep, 4.5)--(6.5*\layersep, 0.5)--(3.5*\layersep, 0.5)--cycle;
\node[annot, anchor=south, text=red] at (5*\layersep, 5.0) {Decoder};

\end{tikzpicture}
\caption{\label{fig:autoencoder} %
Schematic illustration of the autoencoder architecture.
}
\end{figure}

\section{Methods}
\label{sec:methods}

\subsection{Detecting phase transitions with autoencoders}
\label{sec:autoencoders_for_phase}
In our autoencoder, the encoder is exactly the sought-after order observable $\calO:\XX\mapsto\RR^d$ as introduced in Sec.~\ref{sec:ssb}, except that we do not require it to be $G$-equivariant at present. 
The decoder $\calD:\RR^d\times\RR\mapsto\XX$ then represents the conditional Boltzmann distribution of lattice configurations given a value of the observable $\calO$ and a temperature $T$. 
Note that in contrast to traditional autoencoders, this decoder accepts a second argument --- the temperature --- as a direct input; we do this because we know the Boltzmann distribution of lattice configurations depends on temperature.  This explicit temperature dependence is the first novelty of our method, as previous works on autoencoders for identifying phase transitions assumed a temperature-independent architecture.  Once the autoencoder is trained, we may interpret any abrupt change in the distribution of the learned observable $\calO$ with respect to temperature as indicative of a phase transition.

Even in the absence of any knowledge about the symmetries of the Hamiltonian, previous works found that autoencoders could identify phase transitions with some accuracy~\cite{alexandrou2020critical}. 
However, the reason for the autoencoder's efficacy remains unclear at this time. 
Rather than relying only on the empirical success of autoencoders, we motivate their use with the following intuition: Using general information theory, we can show that training an autoencoder is equivalent to maximizing the entropy of the observable $\calO$ learned by the encoder, 
where we regard $\calO$ as a function of the random lattice configuration across different temperatures. 
Thus, training the autoencoder moves the distribution of $\calO$ closer to a uniform distribution with as large of a support as possible, 
and this in turn means that $\calO$ learns to aggregate low-probability states together. 
This property seems to mimic the type of coarse-graining performed in Landau theory, 
where the competition between high-probability states and aggregations of low-probability states drives a phase transition.

For the example application of the Ising model, we define the encoder and decoder to have shallow neural network architectures, each with one hidden layer of nonlinear activation units:
\begin{align}
\calO(\mathbf{x}) &= \mathbf{c} + \sum_{k=1}^h \mathbf{a}_k\phi(\langle \mathbf{w}_k, \mathbf{x}\rangle_F + b_k) \label{eq:encoder} \\
\calD(\mathbf{z}, T) &= \tanh\left[\mathbf{c}^{\prime} + \sum_{k=1}^{h^{\prime}}\mathbf{a}^{\prime}_k\phi(\langle \mathbf{w}^{\prime}_k, \mathbf{z}\rangle_F+b^{\prime}_k+b^{\prime\prime}T)\right], \label{eq:decoder}
\end{align}
where $b_k\in\RR$, $\mathbf{a}_k,\mathbf{c}\in\RR^d$, $\mathbf{w}_k\in\RR^{L\times L}$; 
$b^{\prime}_k,b^{\prime\prime}\in\RR$, $\mathbf{a}^{\prime}_k,\mathbf{c}^{\prime}\in\RR^{L\times L}$, $\mathbf{w}^{\prime}_k\in\RR^d$; 
$h$ and $h^\prime$ are the number of hidden neurons in the encoder and decoder, respectively; 
$\langle\cdot, \cdot\rangle_F$ denotes the Frobenius inner product (Hadamard product of matrices followed by a sum over all entries); 
$\tanh()$ is applied elementwise; 
and $\phi:\RR\mapsto\RR$ is the elementwise leaky rectified linear unit (ReLU) activation function defined as
\begin{equation}
\phi(y) = 
\begin{cases}
0.01y, & \mbox{ if } y < 0 \\
y, & \mbox{ otherwise.}
\end{cases}
\end{equation}
The $\tanh$ function is used in the decoder to guarantee each output component lies in the interval $(-1, 1)$. 
We set $h=4$ and $h^\prime=64$. We also set $d=1$, as stated in Sec.~\ref{sec:ssb}. 
Since the magnetization and staggered magnetization are linear functions, a linear encoder and linear decoder would have been sufficient. However, for the purpose of demonstrating the efficacy of our method, we assume no knowledge of the system except a dataset of MC-sampled lattice configurations over a range of temperatures and the group $G$ of Ising model symmetries. We therefore consider an architecture deliberately more complex than a linear autoencoder, and one that would be a reasonable initial choice given no additional information about the system.

Now, given a dataset $\{(\mathbf{x}_n, T_n)\in \XX\times [0, \infty)\}_{n=1}^N$ of lattice configurations $\mathbf{x}_n$ at temperatures $T_n$, we train the autoencoder by minimizing the loss
\begin{equation} \label{eq:loss}
\calL(\calO, \calD) = \frac{1}{N}\sum_{n=1}^N L_{\mathrm{BCE}}(\calD(\calO(\mathbf{x}_n), T_n), \mathbf{x}_n),
\end{equation}
where $L_{\mathrm{BCE}}:(-1, 1)^{L\times L}\times\{-1, 1\}^{L\times L}\mapsto (0, \infty)$ is the binary cross-entropy loss function defined as
\begin{multline} \label{eq:bce}
L_{\mathrm{BCE}}(\mathbf{\hat{x}}, \mathbf{x}) =
-\sum_{\bf i}\left[\left(\frac{1+x_{\bf i}}{2}\right)\log\left(\frac{1+\hat{x}_{\bf i}}{2}\right)
    \right. \\ 
\left. + \left(\frac{1-x_{\bf i}}{2}\right)\log\left(\frac{1-\hat{x}_{\bf i}}{2}\right)\right],
\end{multline}
where $\mathbf{\hat{x}}$ is the output of the autoencoder.

\subsection{The group-equivariant autoencoder}

We now extend the baseline autoencoder introduced in Sec.~\ref{sec:autoencoders_for_phase} to a \emph{group-equivariant autoencoder} (GE-autoencoder) by incorporating our prior knowledge about the symmetries of the system's Hamiltonian into the network architecture. 
Once trained, we will then be able to interpret the GE-autoencoder to infer which symmetries are spontaneously broken at any temperature. 

\subsubsection{The subgroup of never-broken symmetries}

The first step is to see if we can identify a subgroup of ``never-broken symmetries'' -- symmetries that do not spontaneously break at any temperature. Identifying these will reduce the number of symmetries that we will ultimately have to check. Our approach for this step is entirely group-theoretic.

Recall from Sec.~\ref{sec:ssb} that a symmetry $g\in G$ will remain unbroken as a function of model parameters and temperature if $\psi_g=1$. In the absence of any knowledge about the true representation $\psi$ associated with the order parameter, 
we can deduce a subgroup of never-broken symmetries by finding all symmetries $g\in G$ such that $\psi_g=1$ for all representations $\psi:G\mapsto\{-1, 1\}$. 
We establish such a subgroup for the Ising symmetry group in Prop.~\ref{prop:never_broken} (see Appendix~\ref{appendix:propositions:never}); 
we denote the subgroup  as $\scb(L) = \langle\alpha^2,\rho^2,(\alpha\rho)^2\rangle$ and refer to it as the \emph{special checkerboard group}, 
as it represents the group of all proper (i.e., no reflections) symmetries of an $L\times L$ checkerboard that map black (resp.~white) squares onto black (resp.~white) squares%
\footnote{The set of black squares (resp.~white squares) is also referred to as sublattice A (resp.~sublattice B) in the literature.}. 
Thus, all even-parity translational symmetries and the $180^\circ$-rotational symmetry are never spontaneously broken in the Ising model. 

Having established a subgroup of never-broken symmetries, it can be shown (see Appendix~\ref{appendix:propositions:never}) that the only symmetries we have to check for SSB are $\alpha^{m_1}\rho^{m_2}\tau^{m_3}\sigma^{m_4}$ for $m_i\in \{0, 1\}$. 
This represents a reduction from $16L^2$ to $16$ symmetries to check, so that the complexity of detecting SSB is now independent of lattice size; an important advance of the proposed method.

\subsubsection{Incorporating symmetries into the encoder}

The next step of our method is to incorporate the deduced subgroup of never-broken symmetries into our autoencoder. 
Recall from Sec.~\ref{sec:ssb} that the observable $\calO:\XX\mapsto\RR^d$, which is modeled by the encoder network of the autoencoder, must be $G$-equivariant; 
we start by first constraining the parameters of the encoder $\calO$ [Eq.~\eqref{eq:encoder}] such that it is invariant to the subgroup of never-broken symmetries $\scb(L)$. 
However, it turns out that there are many inequivalent ways to do this, and it is unclear which set of constraints is optimal. 
A complete classification of all ways this information can be incorporated, as well as the development of a metric by which to determine which way is best, is beyond the scope of this paper and is left for future work. Here, we enforce invariance in a simple way and find that it yields good results. 
The general idea is illustrated in the top panel of Fig.~\ref{fig:encoder}.

We start with Eq.~\eqref{eq:encoder} for the observable encoder $\calO:\XX\mapsto\RR$. 
We constrain the elements $w_{k,{\bf i}}$ of each matrix $\mathbf{w}_k$ to be
\begin{equation} \label{eq:uivi}
w_{k,{\bf i}} = \frac{2}{L^2}
\begin{cases}
u_k, & \mbox{ if } i_x+i_y \mbox{ is even} \\
v_k, & \mbox{ otherwise.}
\end{cases}
\end{equation}
Each $\mathbf{w}_k$ is thus constrained to have a ``checkerboard'' pattern and is invariant under the action of $\scb(L)$. 
Since it can be shown that $G$ acts orthogonally on all of $\RR^{L\times L}$, then the invariance of $\calO$ under $\scb(L)$ immediately follows.

\begin{figure}
\centering
\def\layersep{1.5cm}
\begin{tikzpicture}[node distance=\layersep, auto, scale=0.75]
\scriptsize
\tikzstyle{neuron}=[circle, draw, fill=black!25, minimum size=17pt, inner sep=0pt];
\tikzstyle{annot} = [text width=4em, text centered]

\foreach \y/\shade in {1/25, 2/5, 3/25, 4/5}
\node[rectangle, draw, fill=black!\shade] (I-\y) at (-\layersep, \y) {};
\draw[black, dashed] (-1.4*\layersep, 4.5)--(-0.6*\layersep, 4.5)--(-0.6*\layersep, 0.5)--(-1.4*\layersep, 0.5)--cycle;
\node[annot, anchor=south, text=black] at (-\layersep, 4.6) {Input};

\foreach \y in {1,...,4}
\node[neuron] (EI-\y) at (0, \y) {};
\foreach \y in {1,...,3}
\node[neuron] (EH-\y) at (\layersep, \y.5) {};
\foreach \y in {1,...,1}
\node[neuron] (EO-\y) at (2*\layersep, \y+1) {};
\foreach \source/\shade in {1/100, 2/50, 3/100, 4/50}
\foreach \dest/\c in {1/red, 2/green, 3/blue}
\path[\c!\shade] (EI-\source) edge (EH-\dest);
\foreach \source/\shade in {1/100, 2/75, 3/50}
\foreach \dest in {1,...,1}
\path[purple!\shade] (EH-\source) edge (EO-\dest);
\draw[black, dashed] (-0.4*\layersep, 4.5)--(2.4*\layersep, 4.5)--(2.4*\layersep, 0.5)--(-0.4*\layersep, 0.5)--cycle;
\node[annot, anchor=south, text=black] at (\layersep, 4.6) {Encoder};

\foreach \y in {1,...,4}
\path[->] (I-\y) edge (EI-\y);

\node[above, font=\small\bfseries] at (current bounding box.north) {Constrained encoder};

\end{tikzpicture}%
\qquad %
\def\layersep{1.5cm}
\begin{tikzpicture}[node distance=\layersep, auto, scale=0.75]
\scriptsize
\tikzstyle{neuron}=[circle, draw, fill=black!25, minimum size=17pt, inner sep=0pt];
\tikzstyle{annot} = [text width=4em, text centered]

\foreach \y/\shade in {1/25, 2/5, 3/25, 4/5}
\node[rectangle, draw, fill=black!\shade] (I-\y) at (-\layersep, \y) {};
\draw[black, dashed] (-1.4*\layersep, 4.5)--(-0.6*\layersep, 4.5)--(-0.6*\layersep, 0.5)--(-1.4*\layersep, 0.5)--cycle;
\node[annot, anchor=south, text=black] at (-\layersep, 4.6) {Input};

\foreach \y in {1,...,2}
\node[neuron] (EI-\y) at (0, \y+1) {};
\foreach \y in {1,...,3}
\node[neuron] (EH-\y) at (\layersep, \y.5) {};
\foreach \y in {1,...,1}
\node[neuron] (EO-\y) at (2*\layersep, \y+1) {};
\foreach \source/\shade in {1/100, 2/50}
\foreach \dest/\c in {1/red, 2/green, 3/blue}
\path[\c!\shade] (EI-\source) edge (EH-\dest);
\foreach \source/\shade in {1/100, 2/75, 3/50}
\foreach \dest in {1,...,1}
\path[purple!\shade] (EH-\source) edge (EO-\dest);
\draw[black, dashed] (-0.4*\layersep, 4.5)--(2.4*\layersep, 4.5)--(2.4*\layersep, 0.5)--(-0.4*\layersep, 0.5)--cycle;
\node[annot, anchor=south, text=black] at (\layersep, 4.6) {Encoder};

\foreach \source/\dest in {1/1, 2/2, 3/1, 4/2}
\path[->] (I-\source) edge (EI-\dest);

\node[above, font=\small\bfseries] at (current bounding box.north) {Reduced encoder};

\end{tikzpicture}
\caption{\label{fig:encoder} %
Top: Schematic illustration of an encoder constrained to be invariant to a group of symmetries (even-unit cyclic translations in the illustration). 
Network edges of same color and shade are constrained to have equal weight. 
Bottom: Reduced representation of the same encoder. 
Rather than replicating weights in the first layer, lattice sites that feed into the same input neuron are first averaged.
}
\end{figure}
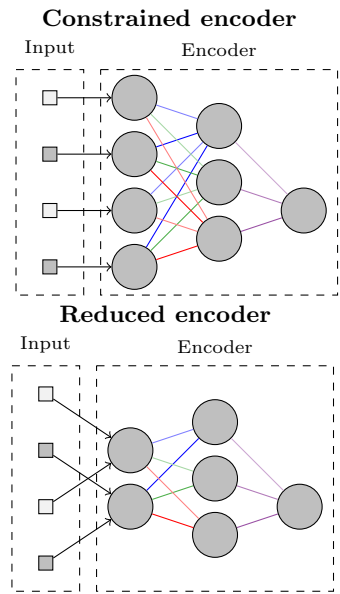

The constraints imposed on the $\mathbf{w}_k$ allow for a significant simplification of the expression for the encoder; 
this is illustrated in Fig.~\ref{fig:encoder}~(bottom). 
For an $L\times L$ lattice configuration $\mathbf{x}$, define
\begin{align*}
\checkb{x} &= \frac{2}{L^2}\sum_{i_x+i_y \mbox{ even}} x_{\bf i} \\
\checkw{x} &= \frac{2}{L^2}\sum_{i_x+i_y \mbox{ odd}} x_{\bf i}.
\end{align*}
We refer to $(\checkb{x}, \checkw{x})$ as the ``checkerboard average'' of the lattice configuration $\mathbf{x}$, i.e., the average value over all ``black squares'' and the average value over all the ``white squares''. 
Then we have
\[ \langle \mathbf{w}_k, \mathbf{x}\rangle_F = u_k\checkb{x} + v_k\checkw{x}, \]
which allows for a more efficient implementation of the encoder. 
We define the reduced encoder $\check{\calO}:[-1, 1]^2\mapsto\RR$ by
\begin{equation} \label{eq:encoder_reduced}
\check{\calO}((\checkb{x}, \checkw{x})) = c + \sum_{k=1}^h a_k\phi(u_k\checkb{x}+v_k\checkw{x}+b_k).
\end{equation}
This result allows us to evaluate the encoder in Eq.~\eqref{eq:encoder} in two separate steps: 
First, we compute the checkerboard average $(\checkb{x}, \checkw{x})$ of the input lattice configuration $\mathbf{x}$. 
This task is a one-time calculation and can be done across the entire available dataset of lattice configurations as a preprocessing step. 
Second, we evaluate the reduced encoder [Eq.~\eqref{eq:encoder_reduced}] on these checkerboard averages. This task is now independent of the lattice size $L$. 
Note that the checkerboard average of $\mathbf{x}$ is manifestly invariant under the action of $\scb(L)$, and thus so is the reduced encoder $\check{\calO}$. 
Moreover, the form of the reduced encoder [Eq.~\eqref{eq:encoder_reduced}] places additional constraints on the spatial symmetries and implies $\psi_{\alpha} = \psi_{\rho} = \psi_{\tau}$ (see Prop.~\ref{prop:art} in Appendix~\ref{appendix:propositions:art} for details). 
The upshot is that we now need only estimate $\psi_{\sigma}$ and one of $\psi_{\alpha}$, $\psi_{\rho}$, and $\psi_{\tau}$ from the data; 
we choose $\psi_{\tau}$ without loss of generality. 
If $\psi_{\tau}=1$, then all spatial symmetries in $\langle\alpha,\rho,\tau\rangle$ never break spontaneously. 
Estimating $\psi_{\sigma}$ and $\psi_{\tau}$ is discussed in Sec.~\ref{sec:regularization}.

\subsubsection{Incorporating symmetries into the decoder}
\label{sec:decoder}

As with the encoder, we now incorporate the subgroup of never-broken symmetries into the decoder of our autoencoder. 
The starting point is Prop.~\ref{prop:invariant_min} (see Appendix~\ref{appendix:propositions:invariant}), which states that under suitable conditions, if an unsupervised model is fit to a dataset containing symmetries, then the fit model will be invariant to those symmetries at least when restricted to the dataset. 
Proposition~\ref{prop:invariant_min} provides a strong motivation to assume that our autoencoder $f = \calD\circ\calO$ is $G$-invariant;  
for $g\in G$,
\begin{align*}
g\calD(\calO(\mathbf{x}), T) &= \calD(\calO(g\mathbf{x}), T) \\
g\calD(\calO(\mathbf{x}), T) &= \calD(\psi_g\calO(\mathbf{x}), T).
\end{align*}
We demand this hold for all $\calO(\mathbf{x})$ and make the stronger assumption
\[ g\calD(z, T) = \calD(\psi_g z, T),~\forall g\in G,z\in\RR. \]
We have already deduced that $\psi_g=1$ for all $g\in\scb(L)$, giving us the constraint
\[ g\calD = \calD, \forall g\in\scb(L). \]
This constraint necessitates the output of the decoder to have a checkerboard pattern as in Eq.~\eqref{eq:uivi}, 
and thus it suffices to have the decoder return only two values ---  
one representing the value on the black squares and the other for the white squares. 
Therefore, we define the reduced decoder $\check{\calD}:\RR\times (0, \infty)\mapsto (-1, 1)^2$ as
\begin{equation} \label{eq:decoder_reduced}
\check{\calD}(z, T) = \tanh\left[\mathbf{c}^\prime + \sum_{k=1}^{h^\prime} \mathbf{a}^\prime_k\phi(w^\prime_k z + b^\prime_k + b^{\prime\prime} T)\right],
\end{equation}
where $\mathbf{a}^\prime_k,\mathbf{c}^\prime\in\RR^2$ and $w^\prime_k,b^\prime_k,b^{\prime\prime}\in\RR$; 
$h^\prime$ is the number of hidden neurons (we set $h^\prime=64$); 
and the functions $\tanh$ and $\phi$ are applied elementwise. 
We interpret the output as the checkerboard average of the output of the unreduced decoder $\calD$. 
Note that like the reduced encoder, the reduced decoder is now independent of the lattice size $L$. 
The reduced encoder [Eq.~\eqref{eq:encoder_reduced}] and reduced decoder [Eq.~\eqref{eq:decoder_reduced}] together comprise a reduced autoencoder $\check{\calD}\circ\check{\calO}$, 
which can now be trained directly on the preprocessed and reduced dataset of checkerboard-averaged lattice configurations.

\subsubsection{Symmetry regularization}
\label{sec:regularization}

The final step is to ensure that the (reduced) encoder is not only invariant to the subgroup $\scb(L)$ of never-broken symmetries but is in fact $G$-equivariant; 
we want $\check{\calO}(g\check{x}) = \psi_g\check{\calO}(\check{x})$ for all $g\in \{\alpha,\rho,\tau,\sigma\}$. 
By Prop.~\ref{prop:art}, it is sufficient to consider only $\tau$ and $\sigma$, 
and since $\psi_g=\pm 1$, then we want $\check{\calO}(g\check{x}) = \pm\check{\calO}(x)$ for all $g\in\{\tau,\sigma\}$. 
We impose this as a soft constraint by including regularization terms in the loss function used to train the reduced autoencoder:
\begin{align}
\check{\calL}(\check{\calO}, \check{\calD}) 
&= \frac{1}{N}\sum_{n=1}^N L_{\mathrm{BCE}}(\check{\calD}(\check{\calO}(\check{x}_n), T_n), \check{x}_n) \nonumber \\
&+ \lambda \sum_{g\in\{\tau,\sigma\}} \left(1 - \frac{\Vert\check{\calO}\circ g\Vert}{\Vert\check{\calO}\Vert}\right)^2 \nonumber \\
&+ \lambda \sum_{g\in\{\tau,\sigma\}} [1 - L_{\mathrm{cos}}(\check{\calO}, \check{\calO}\circ g)^2] \nonumber \\
&+ \lambda \min_{g\in\{\tau,\sigma\}} [1 + L_{\mathrm{cos}}(\check{\calO}, \check{\calO}\circ g)], \label{eq:loss_reduced}
\end{align}
where $\lambda \geq 0$ is a regularization coefficient and $L_{\mathrm{cos}}$ is the cosine similarity between observables defined as
\begin{equation*}
L_{\mathrm{cos}}(\calO_1, \calO_2) = \frac{\langle\calO_1, \calO_2\rangle}{\Vert\calO_1\Vert\Vert\calO_2\Vert},
\end{equation*}
where
\begin{equation*}
\langle\calO_1, \calO_2\rangle = \frac{1}{N}\sum_{n=1}^N \calO_1(\mathbf{x}_n)^\top\calO_2(\mathbf{x}_n).
\end{equation*}
The first regularization term [i.e., the second term in Eq.~\eqref{eq:loss_reduced}]
enforces the soft constraint $\Vert\check{\calO}\circ g\Vert \approx \Vert\check{\calO}\Vert$. The second regularization term drives the cosine similarity to one of its extreme values $\pm 1$. Together, these two terms encode the constraint $\check{\calO}\circ g \approx \pm 1\check{\calO}$ as desired. 
To explain the final regularization term, recall that we require $\psi_g=-1$ for some $g\in G$ to avoid a trivial representation. The last term drives at least one of $\psi_{\tau}$ and $\psi_{\sigma}$ to $-1$ to satisfy this requirement. 

Once trained, we estimate $\psi_g$ for $g\in\{\tau,\sigma\}$ with the final cosine similarity:
\begin{equation} \label{eq:psi}
\psi_g \approx L_{\mathrm{cos}}(\check{\calO}, \check{\calO}\circ g).
\end{equation}

\subsection{Experimental setup}

\subsubsection{Datasets}

We generate datasets of lattice configurations by MC-sampling the 2D ferromagnetic and antiferromagnetic Ising models. 
We impose periodic boundary conditions on an $L\times L$ lattice and consider $L=16$, $32$, $64$, and $128$. For each lattice size, we consider 100 temperatures with 25 values in $[1.04, 2]$ in increments of $0.04$, 50 values in $[2.01, 2.5]$ in increments of $0.01$, and 25 values in $[2.54, 3.5]$ in increments of $0.04$. This distribution of temperatures is evenly distributed about the theoretical critical temperature $T_{\mathrm{c}}=\frac{2}{\log(1+\sqrt{2})}=2.269\ldots$ and denser near $T_{\mathrm{c}}$. 
Although not uniform, the temperature samples are constant across all order observable models and thus do not effect the model comparison; 
we sample more temperatures near $T_{\mathrm{c}}$ only to ensure we achieve results sufficiently stable to draw meaningful conclusions. 
For each lattice size and temperature, we run the Wolff algorithm first for $10,000$ iterations to allow for thermal equilibration and then for an additional $50,000$ iterations during which we record the lattice configuration every 10 iterations. 
We thus obtain $5,000$ samples for each lattice size and temperature and for each of the ferromagnetic and antiferromagnetic cases (although in practice we only use $4,096$ samples). We also preprocess copies of these datasets by checkerboard-averaging the lattices, which will be used to train the GE-autoencoder. We evenly split each set of $4,096$ lattice configurations into a training-validation set and a test set, 
and we further partition the $2,048$ training-validation samples into eight ``folds'' each of size $256$ that will be used to measure sampling variance.

\subsubsection{Order observables}

Given only the MCMC datasets and the Ising symmetry group $G$, our  objective is to detect when a phase transition occurs by (1) identifying the associated spontaneous symmetry breaking and (2) estimating the temperature where it occurs (\emph{i.e.}~the critical temperature). Importantly, we assume no prior knowledge about the Ising model beyond the given datasets and the symmetry group. 

We test three ``order observables'' from which we hope to derive order parameters:
\begin{enumerate}
\item Magnetization~[Eqs.~\eqref{eq:magnetization} \& \eqref{eq:staggered_magnetization}]: %
(In the antiferromagnetic case, ``magnetization'' will be understood to mean the 
staggered magnetization.) This observable is the standard order parameter used for the Ising model. Here, we use it to provide a ground-truth estimate for the critical temperature in comparison to the exact value obtained from the Onsager solution. 

\item Baseline-autoencoder~[Eqs.~\eqref{eq:encoder}, \eqref{eq:decoder}, \& \eqref{eq:loss}]: %
This autoencoder does not exploit the symmetry group $G$ and is used as a 
machine learning baseline. We will refer to its encoder and decoder as baseline-encoder and baseline-decoder, respectively. Once trained, we interpret the output of the encoder as an order observable. Note that $G$-equivariance is not guaranteed a priori. 

\item GE-autoencoder~[Eqs.~\eqref{eq:encoder_reduced}, \eqref{eq:decoder_reduced}, \& \eqref{eq:loss_reduced}]: %
This autoencoder takes advantage of the symmetry group $G$, and we thus expect it to be more accurate and more efficient than the baseline-autoencoder. 
We will refer to its encoder and decoder as GE-encoder and GE-decoder, respectively. 
Once trained, we interpret the output of the encoder as a $G$-equivariant order observable. We also interpret its representation of $G$ to identify which symmetries spontaneously break. 
Finally, as the GE-autoencoder acts on checkerboard-averaged lattice configurations, the same network architecture can be applied to different sizes of lattices. 
Therefore, we will also consider the case of a ``multiscale GE-autoencoder'', which is trained simultaneously on all four lattice sizes in our dataset while using only  one-quarter of the MCMC data for each lattice size.
\end{enumerate}
We evaluate magnetization and the trained baseline-encoder and GE-encoder observables on all lattice configurations in our datasets to obtain measurement distributions and subsequently order parameters. Further details on using these order parameters to estimate the critical temperature are given in Sec.~\ref{sec:results}.

\subsubsection{Training details}

Independent of the lattice size $L$ and training-validation fold $j$, we train and validate the baseline- and GE-autoencoders on a dataset of $100N$ lattice configurations, where $50\%$ of the data set is randomly selected for training and the remaining $50\%$ is used for validation. This dataset consists of the last $N$ MC-sampled lattice configurations (out of the total $256$ configurations in the $j^\text{th}$ fold) at each of the 100 temperatures considered. 
We test various values of $N$ (ranging from $8$ to $256$ in powers of $2$) to measure the data efficiency of the GE-autoencoder vs.\ baseline-autoencoder, as well as the dependence of estimated $T_{\mathrm{c}}$ on the training-validation sample size. 
We train all autoencoders using the Adam optimizer with learning rate of $0.001$ (or equivalent; see below) and minibatch size $N$ for $64$ epochs; 
this choice guarantees $50$ iterations in each epoch and ensures that the scale of the noise generated from stochastic gradient descent is the same for all experiments~\cite{smith2018dont}. 
While training the GE-autoencoder, we include the symmetry regularization terms in the loss function [Eq.~\eqref{eq:loss_reduced}] only for the second half of training epochs; 
this practice prevents the randomly initialized GE-autoencoder from getting trapped in the nearest local minimum of the loss landscape that possibly corresponds to an incorrect group representation $\psi$ (i.e., the breaking of incorrect symmetries). 
We note that switching on regularization halfway through training is, in a sense, the simplest schedule from weak to strong regularization, 
and our choice of regularization schedule can be further validated in the same way as all other optimizer hyperparameter settings -- namely, in terms of the validation loss. 
We test three different seeds to randomly initialize the autoencoder network parameters. Together with eight training data folds, we thus have $24$ trials of each autoencoder experiment.

\begin{figure*}[t]
\centering
\includegraphics[width=2.0\columnwidth]{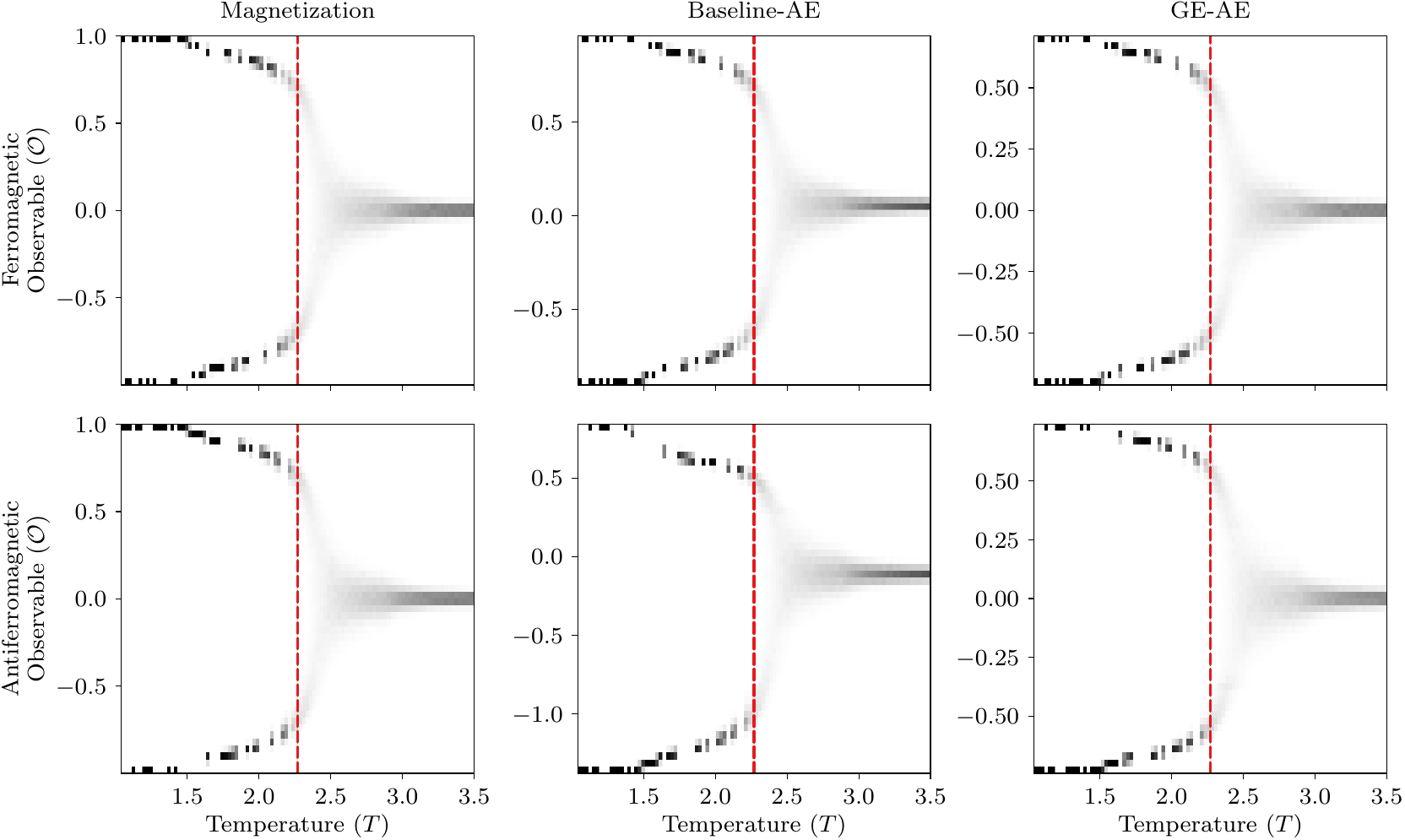}
\caption{\label{fig:distributions} %
Distributions of the order observable values at each temperature $T$ on a lattice of size $L=128$ for one example training data fold and initialization seed and using $N=256$ training-validation samples per temperature. 
Note that the observable scales ($y$-axes) of the baseline-encoder and GE-encoder are arbitrary, as any nonzero rescaling of the encoder can be compensated by the inverse scaling in the decoder; 
thus, only the relative shapes of these distributions are meaningful. 
All three order observables suggest a phase transition in both the ferromagnetic and antiferromagnetic cases near the theoretical critical temperature (dashed red).
}
\end{figure*}

Nontrivial parameter initialization and learning rate settings were needed to obtain reasonable comparisons (see Appendix~\ref{appendix:lr} for details). 
We initialize the GE-autoencoder as usual and set the learning rate to $0.001$ based on validation learning curves. 
However, to ensure a fair comparison between the baseline- and GE-autoencoders, and to avoid artifactual and noisy results due to hand-tuned hyperparameter settings, 
we initialize the baseline-autoencoder such that it is functionally equivalent to the initial GE-autoencoder; 
i.e., the baseline-autoencoder satisfies the same symmetry constraints as the GE-autoencoder at initialization time. 
We then set a separate learning rate for each layer of the baseline-autoencoder such that, it would remain equivalent to the GE-autoencoder throughout training if we maintained the symmetry constraints on the baseline-autoencoder. This requires setting smaller learning rates for larger layers to prevent large sums of parameter updates flowing through the network. 
As a result of these settings, the baseline-autoencoder and GE-autoencoder are identical in terms of their initial values and their learning dynamics and differ only in the symmetry constraints and symmetry regularization imposed on the GE-autoencoder.

\section{Results}
\label{sec:results}

\begin{figure*}[t]
\centering
\includegraphics[width=2.0\columnwidth]{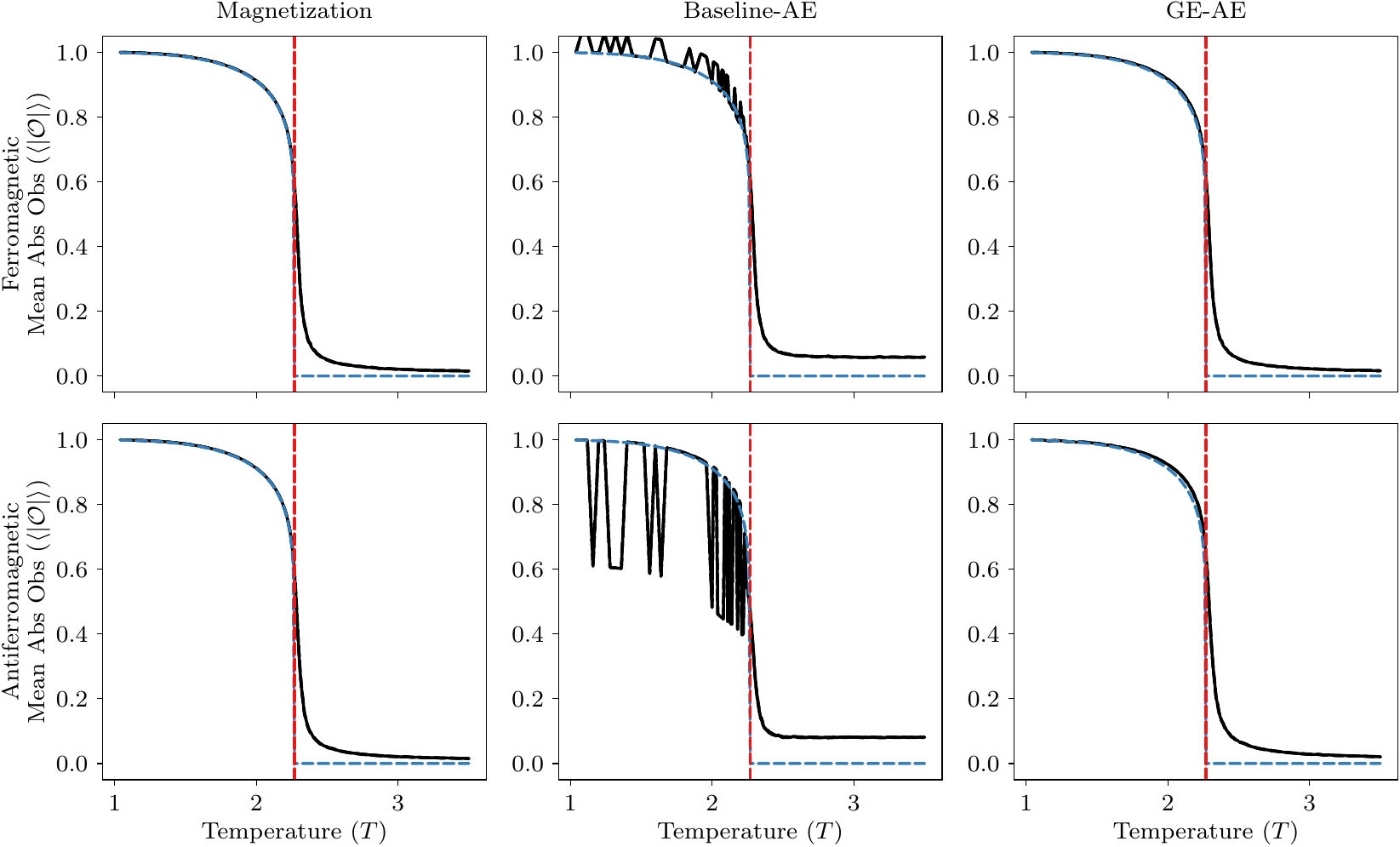}
\caption{\label{fig:orders} %
Order parameters (mean absolute value of the distributions in Fig.~\ref{fig:distributions}) vs.\ temperature $T$ on a lattice of size $L=128$ for one example training data fold and initialization seed and using $N=256$ training-validation samples per temperature. 
Note that the curves have been normalized to the same scale as Onsager's solution for comparison. 
The standard deviation curves (dashed) were obtained using the jackknife resampling method (these may be difficult to see as the standard deviations are small). 
All three order parameters suggest a phase transition in both the ferromagnetic and antiferromagnetic cases near the theoretical critical temperature (dashed red). 
However, the order parameters derived from magnetization and the GE-encoder are smoother and give better approximations to Onsager's solution (dashed blue).
}
\end{figure*}

\subsection{Identifying phase transitions}

Figure~\ref{fig:distributions} is representative of the observed distributions of the magnetization,  
baseline-encoder output, and GE-encoder output over all lattice configurations in our dataset at each temperature. For brevity, we present the distributions only for the largest lattice size $L=128$ and the largest number of training-validation samples per temperature $N=256$. 
In contrast to magnetization and the GE-encoder, the distribution of the baseline-encoder is not symmetric about zero in the antiferromagnetic case. 
This asymmetry is a consequence of a redundancy in the autoencoder network: 
The (baseline) encoder may be freely transformed by any invertible affine function since the first layer of the decoder can always undo it. 
The center of the baseline-encoder distribution is therefore arbitrary. Although previous works~\cite{alexandrou2020critical} have reported approximately symmetric encoder distributions for the Ising model, our results show that this is not guaranteed unless some form of explicit symmetry regularization is used, as in the GE-encoder. 
Similarly, the scale of the encoder is arbitrary as well (even in the GE-encoder) although the scale is not relevant for SSB. 
Nevertheless, all distributions exhibit an abrupt qualitative change near the theoretically known critical temperature $T_{\mathrm{c}}=\frac{2}{\log(1+\sqrt{2})}=2.269\ldots$, 
and hence all three observables are able to identify the phase transition in the Ising model to some degree.

We derive an order parameter from each of the three observables by calculating the mean absolute value under each distribution at each temperature~(Fig.~\ref{fig:orders}). 
Note that while this procedure is justified for magnetization and the GE-encoder as these are $G$-equivariant observables, it is not justified a priori for the asymmetric baseline-encoder observable in the antiferromagnetic case. 
Nevertheless, we do it anyways to provide a baseline case where symmetries were not taken into consideration. 
Thanks to symmetry constraints and regularization, the GE-encoder learns a smoother order parameter that is almost identical to the magnetization order parameter up to a scale factor (this is made quantitative in Appendix~\ref{appendix:onsager}). 
We also compare each order parameter to Onsager's exact solution for spontaneous magnetization in the thermodynamic limit~\cite{onsager1944crystal} 
\begin{equation} \label{eq:onsager}
M_{\mathrm{ONS}}(T) = 
\begin{cases}
\left[1-\sinh^{-4}\left(\frac{2}{T}\right)\right]^{\frac{1}{8}}, & \mbox{ if } T < T_{\mathrm{c}} \\
0, & \mbox{ otherwise,}
\end{cases}
\end{equation}
which is plotted as the dashed blue line in Fig.~\ref{fig:orders}. 
The smooth order parameter of the GE-encoder provides a better approximation to Onsager's solution compared to the baseline-encoder~(Fig.~\ref{fig:orders}). Moreover,  the GE-encoder in the thermodynamic limit converges to Onsager's solution with less error than the baseline-encoder (see Appendix~\ref{appendix:onsager} for details). As an immediate consequence, the GE-encoder is able to identify the phase transition as being second-order, which previous works~\cite{alexandrou2020critical} could not do.

\subsection{Identifying spontaneously broken symmetries}
\label{sec:learned_rep}

Can we identify at each temperature which symmetries of the system have spontaneously broken? We have seen that the GE-encoder order parameter becomes nonzero below some critical temperature (Fig.~\ref{fig:orders}), and thus breaks the $\ZZ_2 = \psi_G$ symmetry. 
It follows that every Ising symmetry $g\in G$ such that $\psi_g=-1$ breaks below this critical temperature, while $\psi_g=1$ implies that $g$ remains unbroken. 
Using Eq.~\eqref{eq:psi}, we estimate $\psi_g$ for each generator of $G$~(Table~\ref{table:generators}). 
In the ferromagnetic case, we find that $\psi_g\approx -1$ only for $g=\sigma$, and hence only the internal spin-flip symmetry breaks. 
In the antiferromagnetic case, we obtain $\psi_g\approx -1$ for every generator $g\in\{\alpha,\rho,\tau,\sigma\}$. In other words, the internal spin-flip symmetry as well as every spatial symmetry not in the special checkerboard subgroup $\scb(L)$ breaks. Our results are in agreement with the known SSB in the Ising model across the magnetic transition, and thus we conclude that our GE-autoencoder method can correctly and accurately detect SSB.

\begin{table}
\centering
\caption{\label{table:generators} %
Estimated latent representation $\psi_g$ of the spatial symmetry generators ($g=\alpha,\rho,\tau$) and the internal symmetry generator ($g=\sigma$). 
Estimates were averaged over all $24$ trials (eight training data folds and three initialization seeds) as well as over all lattice sizes $L$ and training-validation sample sizes $N$ as they showed little variation; 
the reported uncertainties are standard deviations. 
In the antiferromagnetic case, odd spatial symmetries and the internal symmetry spontaneously break at some temperature; 
in the ferromagnetic case, only the internal symmetry spontaneously breaks.
}
\begin{tabular}{cS[table-format=-1.5(2),table-align-uncertainty=true]S[table-format=-1.5(2),table-align-uncertainty=true]}
\toprule
\quad & {Spatial} & {Internal} \\
\midrule
Ferromagnetic & 0.99996\pm 0.00027 & -1.00001\pm 0.00029 \\
Antiferromagnetic & -0.99982\pm 0.00042 & -0.99993\pm 0.00040 \\
\bottomrule
\end{tabular}
\end{table}

For contrast, we also measure the degree to which the baseline-encoder is equivariant; 
we again use Eq.~\eqref{eq:psi} but replace the GE-encoder $\check{\calO}$ with the baseline-encoder $\calO$. 
As with the GE-encoder, we average estimates over all $24$ trials (eight training data folds and three initialization seeds) as well as over all lattice sizes $L$ and training-validation sample sizes $N$. 
In the ferromagnetic case, we find that $\psi_g\approx 0.9998(2)$ for spatial symmetry generators $g\in\{\alpha,\rho,\tau\}$ and $\psi_{\sigma}\approx -0.92(8)$ for the internal spin-flip symmetry generator $\sigma$. 
In the antiferromagnetic case, we obtain $\psi_g\approx -0.87(5)$ for both spatial and internal symmetry generators. 
The baseline-encoder is thus approximately equivariant and transforms by approximately the correct group representation. 
However, the GE-autoencoder learns the representation $\psi$ with significantly greater accuracy and sometimes with orders of magnitude more robustness than the baseline, particularly in the antiferromagnetic case.

\subsection{Estimating the critical temperature}
\label{sec:results:temperature}

\begin{figure*}
\centering
\includegraphics[width=2.0\columnwidth]{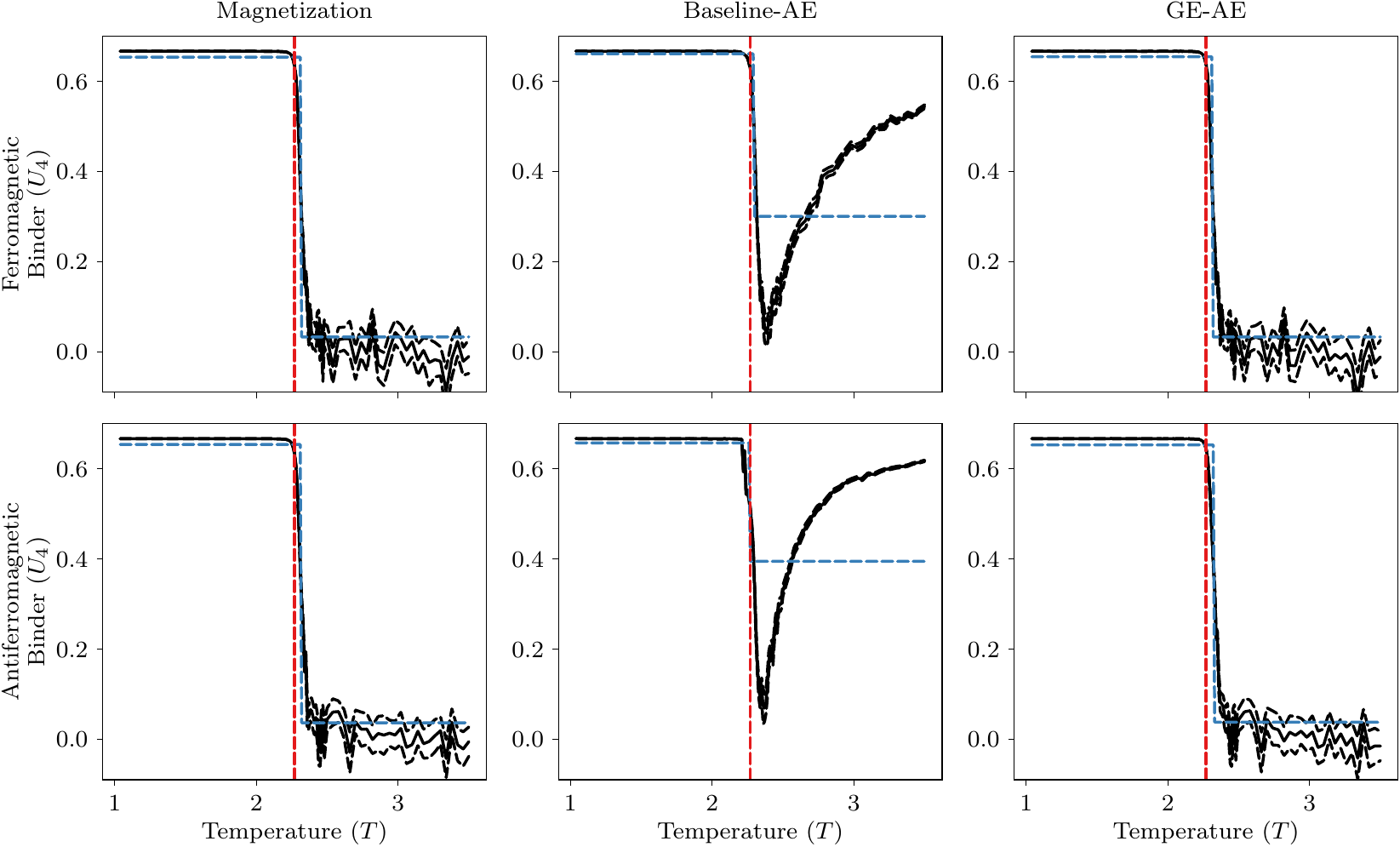}
\caption{\label{fig:u4s} %
Fourth-order Binder cumulants $U_4$ of the order observables at each temperature $T$ on a lattice of size $L=128$ for one example training data fold and initialization seed and using $N=256$ training-validation samples per temperature. 
The standard deviation curves (dashed) were obtained using the jackknife resampling method (these may be difficult to see since the standard deviations are small). 
In all cases, the Binder cumulant drops abruptly near the theoretical critical temperature (dashed red). 
Step functions (blue) are fit to the Binder cumulants using least-squares, and the locations of their jump discontinuities are taken as estimates of the critical temperature.
}
\end{figure*}

We now turn to our second question: Is an SSB-based approach to identifying phase transitions from data more accurate than a purely data-driven approach? In particular, does the GE-encoder give a more accurate estimate of $T_{\mathrm{c}}$ than the baseline-encoder?

We begin by estimating $T_{\mathrm{c}}$ independently for each lattice size based on the fourth Binder cumulant~\cite{binder1993monte}:
\[ U_4 = 1 - \frac{\langle\calO^4\rangle}{3\langle\calO^2\rangle^2}, \]
where $\calO$ is the order observable. 
We obtain Binder cumulant vs.\ temperature curves for each of the three order observables~(Fig.~\ref{fig:u4s}). (Once again, for brevity, we present the curves only for the largest lattice size $L=128$ and number of training-validation samples per temperature $N=256$.) We emphasize that we are simulating a scenario in which we only have access to a dataset of lattice configurations and the symmetries of the system and are not aware that the system is in fact the Ising model. 
Our choice to look at the Binder cumulant should therefore be interpreted only as a ``guess'', and its only justification is the aposteriori observation that the Binder cumulant curves all display an abrupt change near the theoretical critical temperature. 

We perform least-squares regression to fit a step function to each jackknife-sample Binder cumulant vs.\ temperature dataset. We then interpret the location of the jump discontinuity of the step function as a jackknife-sample estimate of $T_{\mathrm{c}}$. 
However, if our dataset includes temperatures $T_1,T_2,\ldots,T_{100}$ and if we find the jump discontinuity to lie in the open interval $T_i, T_{i+1})$ for some $i$, 
then moving the jump discontinuity to any other temperature in $(T_i, T_{i+1})$ would result in a fit that is just as good as the original step function. 
This approach therefore only allows us to obtain interval estimates of the critical temperature. 
To obtain point estimates, we set up and solve a convex optimization problem in which we seek to minimize the jackknife standard deviation in the estimate subject to the constraints defined by the jackknife-sample interval estimates (see Appendix~\ref{appendix:optimization} for details). 
In this way, for each lattice size $L$ and number of training-validation samples per temperature $N$, we obtain a critical temperature estimate as the jackknife%
\footnote{%
We remark on an important detail in the jackknife calculation: 
It is common to reduce the bias in the jackknife estimate by combining the jackknife mean with the estimate obtained without resampling. 
The argument for this, however, relies on a Taylor expansion of the underlying estimator~\cite{young2015everything}, 
and it turns out that our critical temperature estimator is not everywhere-differentiable; under a small perturbation of the Binder cumulant estimates, our critical temperature estimate either remains constant or changes abruptly if the jackknife-sample interval estimates change. 
Indeed, in our original $T_{\mathrm{c}}$ estimates, we found that the bias estimate was either zero or so extreme that it often pushed the critical temperature estimate outside the range of temperatures included in our dataset.  
In contrast, when we did not reduce the bias, we obtained more stable results. 
Overall, since the bias estimate is known to scale as $\frac{1}{n}$ while the jackknife standard deviation scales as $\frac{1}{\sqrt{n}}$ --- so that the bias is typically much smaller than the standard deviation for a sufficiently large sample size $n$ and can often be ignored~\cite{young2015everything} --- we were confident that our extreme bias estimates were spurious. Therefore, we have not adjusted for them in the results presented~(Figs.~\ref{fig:tc_ferro}-\ref{fig:tc_antiferro}).%
} %
 mean averaged over all $24$ trials (eight training data folds and three initialization seeds), along with a standard deviation~(Figs.~\ref{fig:tc_ferro}-\ref{fig:tc_antiferro}).

\begin{figure}
\centering
\includegraphics[width=1.0\columnwidth]{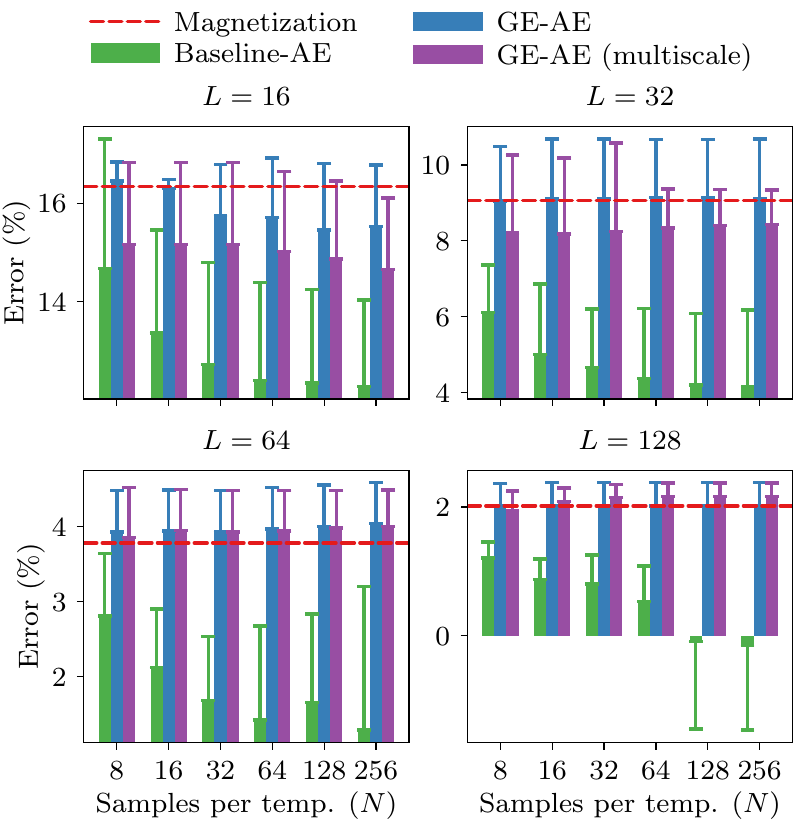}
\caption{\label{fig:tc_ferro} %
Ferromagnetic critical temperature estimates (expressed as percent errors relative to the exact theoretical critical temperature).
The error bars represent standard deviations that combine the standard deviation of means across $24$ trials (eight training data folds and three initialization seeds) and the jackknife standard deviations for each fold and seed, with most of the variance coming from the former. 
The standard deviation associated to magnetization is $0$. 
The baseline-encoder consistently achieves the lowest error and improves with more training data. 
The GE-encoders, on the other hand, are more similar to magnetization, with multiscale training reducing error slightly, and are also more stable in terms of their standard deviations.
}
\end{figure}

In both the ferromagnetic and antiferromagnetic cases, the baseline-encoder consistently achieves lower error in its critical temperature estimates than do the GE-encoders~(Figs.~\ref{fig:tc_ferro}-\ref{fig:tc_antiferro}). 
Moreover, in contrast to the GE-encoders, the baseline-encoder makes better use of more training data, as its error decreases with increasing training dataset size. 
The GE-encoders, on the other hand, achieve errors closer to that of magnetization, 
and their estimates are also more stable in terms of lower standard deviations. 
Multiscale training (i.e., all four lattice sizes in the training dataset) results in an additional reduction in error. 

We speculate that the baseline-encoder achieves the lowest error in $T_{\mathrm{c}}$ estimation because it is a more flexible network in comparison to the GE-encoder. As such, 
it may be able to express certain nonlinearities that the GE-encoder cannot. 
From the proximity of the GE-encoder to magnetization, we infer that its four hidden neurons have aligned such that the GE-encoder is approximately a linear function of its input. 
In contrast, if the baseline-encoder has learned a nonlinearity such that it squashes (resp.~inflates) the value assigned to lattice configurations with low (resp.~high) absolute magnetization, 
then the mean absolute value of the baseline-encoder vs.\ temperature curve will be more ``bowed'' compared to Onsager's solution. Such bowing is indeed what we see~(Fig.~\ref{fig:orders}). 
This observation is consistent with previous works reporting that deeper and more flexible autoencoders incorrectly classify the Ising phase transition as first-order~\cite{alexandrou2020critical}, 
and this could also explain why the baseline-encoder estimates the critical temperature with greater accuracy, even though it does worse when extrapolated to the thermodynamic limit (see Sec.~\ref{sec:results:extrapolate}).

\begin{figure}
\centering
\includegraphics[width=1.0\columnwidth]{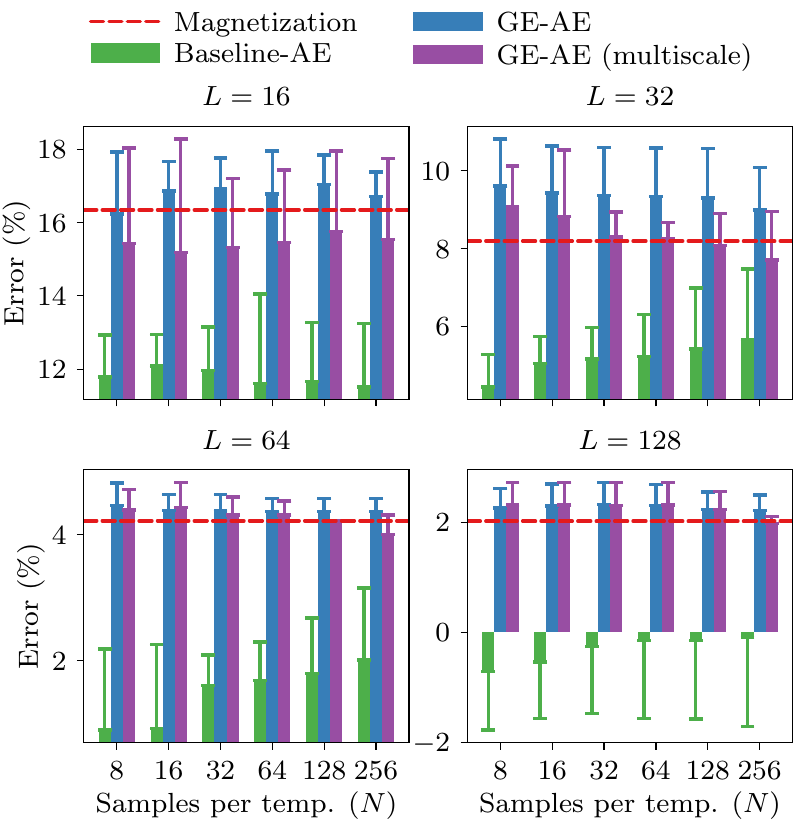}
\caption{\label{fig:tc_antiferro} %
Antiferromagnetic critical temperature estimates (expressed as percent errors relative to the exact theoretical critical temperature).
The error bars represent standard deviations that combine the standard deviation of means across $24$ trials (eight training data folds and three initialization seeds) and the jackknife standard deviations for each fold and seed, with most of the variance coming from the former. 
The standard deviation associated to magnetization is $0$. 
The baseline-encoder consistently achieves the lowest error and improves with more training data. 
The GE-encoders, on the other hand, are more similar to magnetization, with multiscale training reducing error slightly, and are also more stable in terms of their standard deviations.
}
\end{figure}

\subsection{Extrapolating the critical temperature estimates}
\label{sec:results:extrapolate}

Here we perform finite-size scaling analysis on the critical temperature estimates presented in the last section in order to obtain estimates at infinite lattice size -- i.e., in the thermodynamic limit. 
For each training data fold and initialization seed, order observable, training-validation sample size, and each of the ferromagnetic and antiferromagnetic cases, we perform least-squares linear regression on the critical temperature estimates against inverse lattice size ($r^2$ value $\approx 1$ for all fits). 
We plot an example of these fits for $N=256$ training-validation samples per temperature in Fig.~\ref{fig:tc_vs_L}. 
The $y$-intercepts of the linear fits are then taken to be the critical temperature estimates at $L^{-1}=0$ -- i.e., the thermodynamic limit. 
We find that while the baseline-encoder achieves the lowest error in its estimation of the critical temperature for individual finite lattice sizes~(Figs.~\ref{fig:tc_ferro}-\ref{fig:tc_antiferro}), 
the GE-encoders are significantly more accurate once their estimates are extrapolated to infinite lattice size~(Fig.~\ref{fig:tc_extrapolate}). 

\begin{figure}
\centering
\includegraphics[width=1.0\columnwidth]{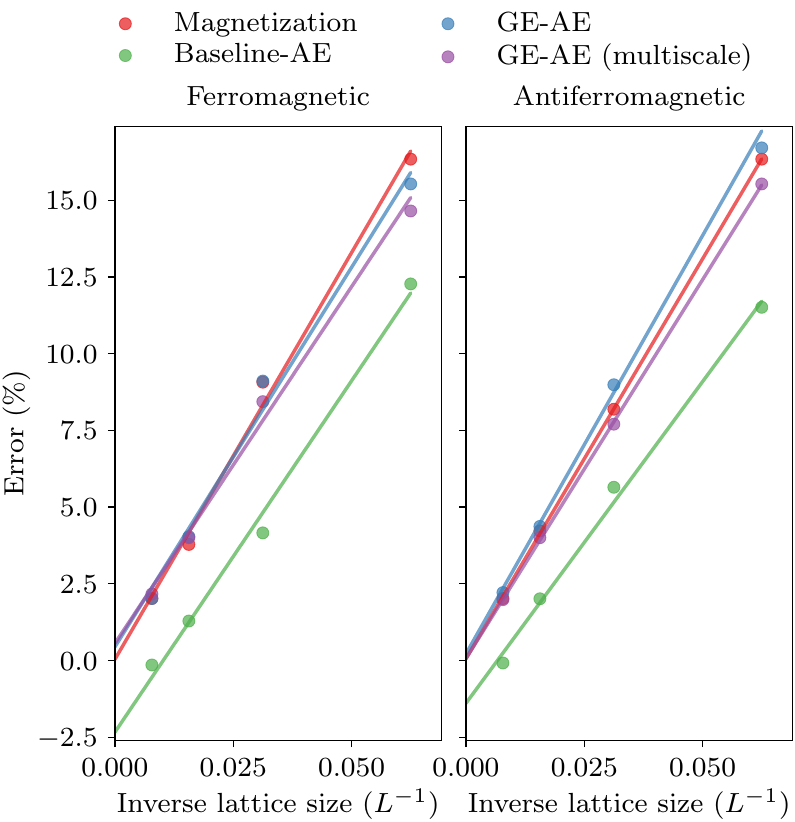}
\caption{\label{fig:tc_vs_L} %
Linear dependence of critical temperature estimates (expressed as percent errors relative to the exact theoretical critical temperature) on inverse lattice size using $N=256$ training-validation samples per temperature and averaged over $24$ trials (eight training data folds and three initialization seeds). 
Although the baseline-encoder achieves lower error for each lattice size individually, the linear extrapolation of its estimates to infinite lattice size ($y$-intercept) is a worse estimate compared to the GE-encoders.
}
\end{figure}

\begin{figure}
\centering
\includegraphics[width=1.0\columnwidth]{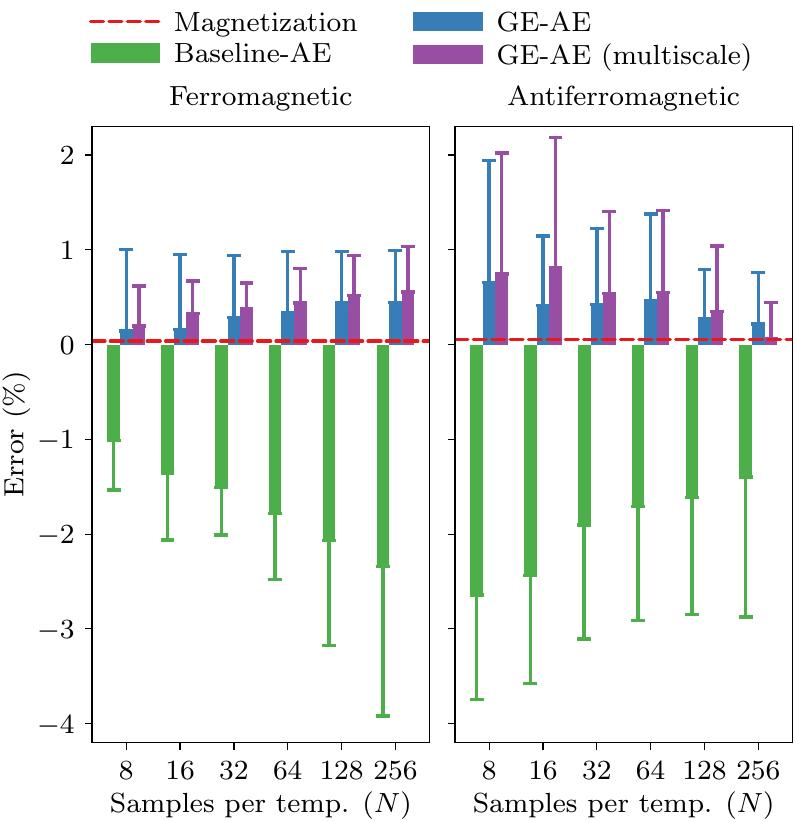}
\caption{\label{fig:tc_extrapolate} %
Critical temperature estimates (expressed as percent errors relative to the exact theoretical critical temperature) extrapolated to infinite lattice size. 
The error bars represent standard deviations that combine the standard deviation of means across $24$ trials (eight training data folds and three initialization seeds) and the jackknife standard deviations for each fold and seed, with most of the variance coming from the former. 
The standard deviation associated to magnetization is $0$. 
Although the baseline-encoder achieves lower error on finite lattices (Figs.~\ref{fig:tc_ferro}-\ref{fig:tc_antiferro}), the GE-encoder estimates extrapolate to have lower error in the thermodynamic limit.
}
\end{figure}

\subsection{Measuring the time efficiency}

We compare the time efficiencies of the GE-autoencoder and baseline-autoencoder methods. 
Since we were able to exploit never-broken symmetries to reduce the network size of the GE-autoencoder, we expect it to be significantly more efficient. 
Figure~\ref{fig:times} reports the total computation times for the GE-autoencoder and baseline-autoencoder methods for each lattice size. These values include the time to generate all the data (i.e., run the MC simulation), all preprocessing time such as checkerboard-averaging the lattice configurations, and the time needed to train and validate the autoencoder and to evaluate the trained encoder on the entire dataset of lattice configurations. 
Importantly, the training-validation-evaluation time is a \emph{sum} over all 24 trials (eight training data folds and three initialization seeds) to reflect the computation needed to obtain error bars on the critical temperature estimates. 
Note that whether we are looking at the ferromagnetic or antiferromagnetic case has no impact on execution time, as the autoencoder architectures and learning hyperparameter settings are identical in both cases. 
Moreover, as a consequence of allowing the minibatch size to proportionally vary with the training-validation sample size $N$, we found that the execution time depended very little on $N$. Thus, we report each execution time as an average over all sample sizes $N$ and over the ferromagnetic vs.\ antiferromagnetic cases. 
We find that the GE-autoencoder method is significantly faster than the baseline-autoencoder. 
Moreover, multiscale training gives an additional boost in efficiency in the computation time needed to extrapolate estimates to infinite lattice size, as the multiscale GE-autoencoder needs to be trained only once across the four finite lattice sizes. 
We therefore conclude that our GE-autoencoder method is indeed more efficient than the baseline-autoencoder method.

\begin{figure}
\centering
\includegraphics[width=1.0\columnwidth]{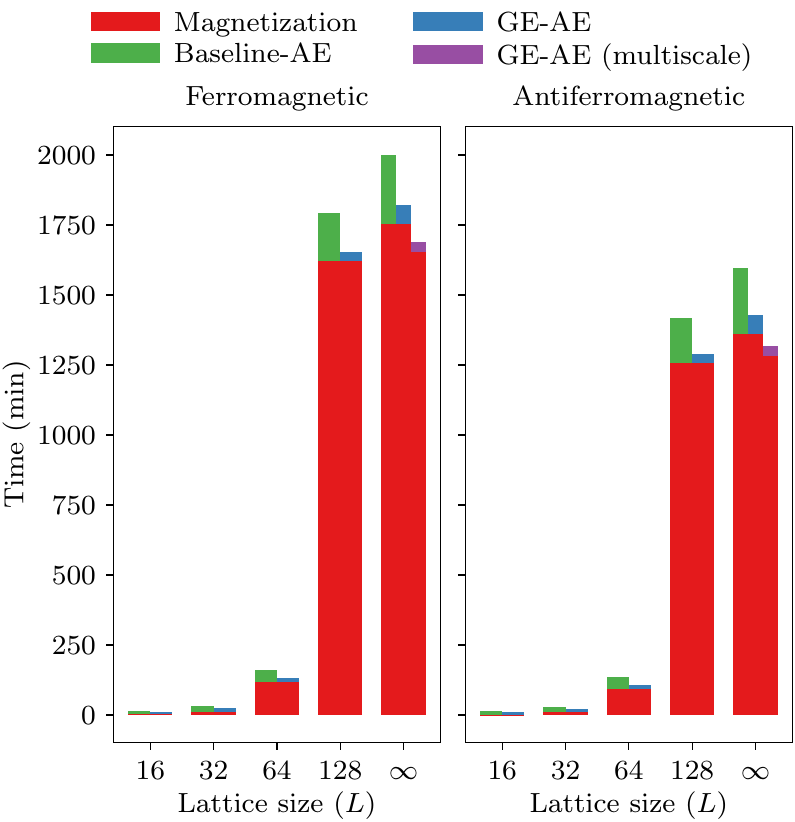}
\caption{\label{fig:times} %
Computation times to generate data (red) and train (not red) the autoencoders for each lattice size $L$. 
Note that the time to measure magnetization is reported as equivalent to data generation time as these measurements were taken on-the-fly during the MC simulation. 
Note also that each training time is a \emph{sum} over all eight training folds and RNG seeds. 
The training times were then averaged over all training-validation sample sizes $N$ and over both the ferromagnetic and antiferromagnetic cases, 
as these factors did not contribute to significant variation in execution time. 
The GE-autoencoder is significantly more time-efficient (lower time) than the baseline-autoencoder. 
The times reported for infinite lattice size are just sums of the times for the finite lattice sizes, 
although the multiscale GE-autoencoder exhibits greater efficiency as it only requires a single network to be trained across all finite lattice sizes.
}
\end{figure}

\subsection{Detecting an external magnetic field}
\label{sec:field}

Finally, we investigate if the baseline-encoder and GE-encoder order observables can be used to detect the presence of a weak external magnetic field $h$ in the ferromagnetic Ising model by adding a term to Eq.~\eqref{Eq:HIsing}:
\begin{equation}
    \calH(\mathbf{x}) \to \calH(\mathbf{x}) - h\sum_{\bf i} {x}_{\bf i}\, . 
\end{equation}
We assume the magnetic field is uniform and note that it breaks the internal symmetry of the Ising model.
We consider two temperatures---$2.0$ and $2.5$---slightly below and above the critical temperature, and we consider three field strengths $0.001J$, $0.01J$, and $0.1J$ (where $J=1$ is the coupling constant in the Ising Hamiltonian). 
For each case, as well as the case of no external field at all, we use the Wolff algorithm with a ``ghost site''~\cite{coniglio1989exact} to simulate the Ising model with lattice size $L=128$ in an external magnetic field; 
we generate $N=2,000$ sample lattice configurations for each temperature and field strength. 
If $x^0_1,\ldots,x^0_N$ (resp. $x_1,\ldots,x_N$) are the sample lattice configurations in the absence (resp. presence) of an external magnetic field, then for each encoder $\calO$, we compute the statistic
\begin{equation}
D = \left \lvert \frac{1}{N}\sum_{n=1}^N\calO(x_n) - \frac{1}{N}\sum_{n=1}^N\calO(x^0_n)\right\rvert. 
\end{equation}
Using the baseline-encoder and GE-encoder already fitted to data as described in previous sections, we obtain $24$ measurements of $D$ for each encoder (eight training data folds and three initialization seeds). 
We then define the ``confidence score''
\begin{equation} \label{eq:xi}
\xi = \frac{\langle D\rangle}{\sqrt{\langle D^2\rangle-\langle D\rangle^2}}.
\end{equation}
Intuitively, a value sufficiently far from zero indicates that the order parameter has shifted and hence there is an external symmetry-breaking field. 
We normalize by the standard deviation of $D$ so that $\xi$ is independent of the arbitrary scale learned by each encoder; 
it also boosts the score when the measurement of $D$ is robust across the $24$ trials. 
Figure~\ref{fig:field} shows the $\xi$ scores for the baseline-encoder and GE-encoder at each temperature and external field strength. 
In all cases, the GE-encoder attains a significantly higher score than does the baseline-encoder, meaning that it is more sensitive to and detects with greater confidence the presence of a weak external magnetic field. 
At temperature $2.5$, we note that the baseline-encoder does become more confident (increasing score) in its detection as the external field becomes stronger, which is to be expected; 
however, the GE-encoder remains confident even in its detection of the weakest field.

\begin{figure}
\centering
\includegraphics[width=1.0\columnwidth]{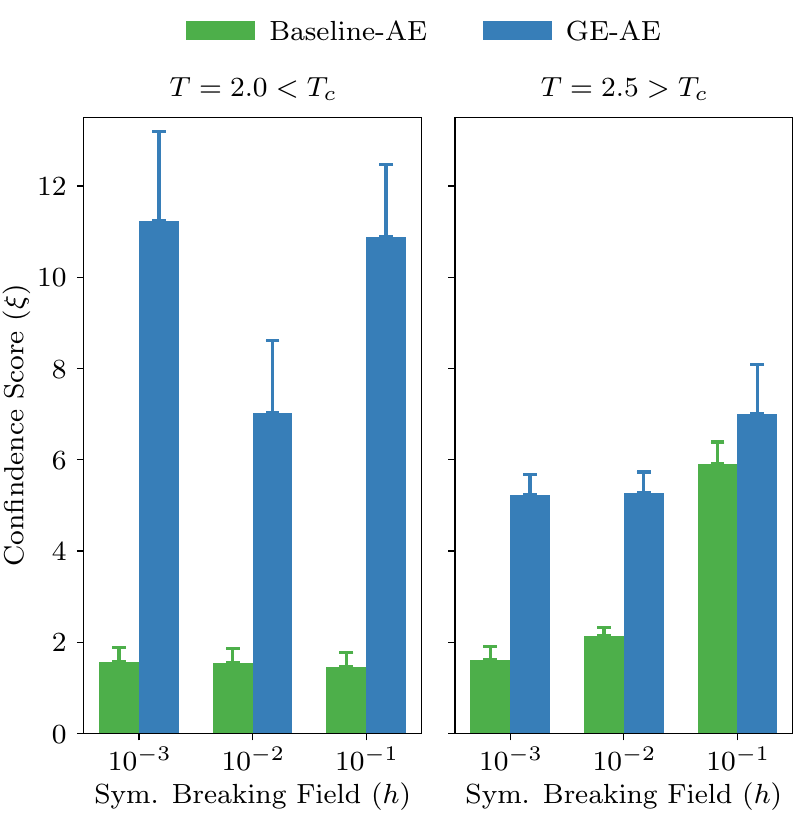}
\caption{\label{fig:field} %
``Confidence scores'' $\xi$ (see Eq.~\ref{eq:xi}) of the baseline-encoder (green) and GE-encoder (blue) at detecting the presence of an external magnetic field at two different temperatures and for three different field strengths. 
The error bars represent standard deviations estimated using $10,000$ bootstrap draws from the $24$ trials used to calculate $\xi$. 
The GE-encoder consistently detects the external field with greater confidence than does the baseline-encoder.
}
\end{figure}

\section{On vector order observables}
\label{sec:vector}

\subsection{A 2D order observable}
\label{sec:2D}

We have thus far assumed that the order observable is scalar-valued or equivalently that a sufficient choice for the latent dimension of the autoencoder networks is $1$. 
In this section, we justify this assumption by instead assuming a 2D vector order observable; 
if the real 2D orthogonal representation $\psi:G\mapsto\mathrm{O}(2, \RR)$ by which the order observable transforms---as learned by the GE-autoencoder---is equivalent to a 1D representation, then we may conclude that a 1D order observable is sufficient to describe the phase transition. 
The exposition in this section will also help to illustrate how our GE-autoencoder method can be extended to a somewhat more complex scenario.

We construct the 2D GE-autoencoder in direct analogy to the 1D case. 
First, we calculate the subgroup of never-broken symmetries and find it to be $H = \langle\alpha^2, \beta^2\rangle$, 
which is the subgroup generated by the even horizontal translations and even vertical translations (see Prop.~\ref{prop:never_broken_2D} in Appendix~\ref{appendix:2D}). 
This subgroup has ${L^2}/{4}$ symmetries for even lattice size $L$, 
and the corresponding quotient group $G/H$ of possibly-broken symmetries has $16L^2/(L^2/4) = 64$ elements. 
Compare this to the 1D case, in which there are $L^2$ never-broken symmetries and only $16$ possibly-broken symmetries.

Next, we constrain the GE-autoencoder so that its encoder is invariant to the subgroup $H$ of never-broken symmetries. 
Instead of constraining each weight matrix $\mathbf{w}_k$ (associated to the $k^\text{th}$ hidden neuron) to have a checkerboard pattern as in the 1D case, 
we constrain it to be a tiling of a $2\times 2$ submatrix; 
the resulting weight matrix is then invariant to the action of the subgroup $H$, and hence so is the encoder. 
Applying this constrained encoder to an input $L\times L$ lattice configuration $\mathbf{x}$ is then equivalent to the following procedure: 
First, we construct a reduced 4D representation $\check{x} = (\check{x}_1,\check{x}_2,\check{x}_3,\check{x}_4)$ of $\mathbf{x}$ as follows:
\begin{align*}
\check{x}_1 &= \frac{4}{L^2}\sum_{i_x \mbox{ even, } i_y \mbox{ even}} x_{\bf i} \\
\check{x}_2 &= \frac{4}{L^2}\sum_{i_x \mbox{ even, } i_y \mbox{ odd}} x_{\bf i} \\
\check{x}_3 &= \frac{4}{L^2}\sum_{i_x \mbox{ odd, } i_y \mbox{ even}} x_{\bf i} \\
\check{x}_4 &= \frac{4}{L^2}\sum_{i_x \mbox{ odd, } i_y \mbox{ odd}} x_{\bf i}.
\end{align*}
The representation $\check{x}$ is just the $2\times 2$ block-average of $\mathbf{x}$. 
Then, we feed $\check{x}$ into a reduced unconstrained encoder $\check{\calO}:[-1, 1]^4\mapsto\RR^2$. 
Finally, we use the 4D output of a reduced decoder $\check{\calD}$ to reconstruct an $L\times L$ lattice configuration by tiling a $2\times 2$ block in the horizontal and vertical directions.

Let $\psi:G\mapsto\mathrm{O}(2,\RR)$ be the 2D latent representation of $G$ by which the 2D order observable transforms. 
We already know $\psi_g=I\forall g\in H$, 
where $I$ is the $2\times 2$ identity matrix. 
Our particular choice of the GE-autoencoder architecture as described above places an additional constraint on $\psi$:
\[ \psi_{\alpha}\check{\calO} = \psi_{\rho^2\tau}\check{\calO}. \]
(see Prop.~\ref{prop:art_2D} in Appendix~\ref{appendix:2D}). 
Thus, we need only learn three of the four generators of $\psi$ while training the GE-autoencoder: $\psi_{\rho}$, $\psi_{\tau}$, and $\psi_{\sigma}$. 
This result is analogous to Prop.~\ref{prop:art} for the 1D case.

We train the 2D GE-autoencoder with the loss function [Eq.~\eqref{eq:loss_2D}] given below in Sec.~\ref{sec:general}; 
it is a generalization of the loss function [Eq.~\eqref{eq:loss_reduced}] and holds for any order observable dimension $d$. 
We train the 2D GE-autoencoder with the same datasets and optimizer settings as in the 1D case. 
After training, we obtain the learned latent representation $\psi$ of $G$ using the estimator [Eq.~\eqref{eq:psihat}] given below in Sec.~\ref{sec:general}. 
Observe that the GE-autoencoder is invariant under the transformation $\check{\calO}(\cdot)\rightarrow A\check{\calO}(\cdot)$ and $\check{\calD}(\cdot, \cdot)\rightarrow \check{\calD}(A^{-1}\cdot, \cdot)$ for any invertible $2\times 2$ matrix $A$. 
Thus, without loss of generality, we transform the learned representation $\psi$ into the eigenbasis of $\psi_{\sigma}$, 
so that $\psi_{\sigma}$ is diagonal with sorted diagonal elements. 
In both the ferromagnetic and antiferromagnetic cases, the elements $\psi_{\rho}$, $\psi_{\tau}$, and $\psi_{\sigma}$---and hence the entire representation $\psi$---are approximately diagonal (Table~\ref{table:generators_2D}). 
More precisely, we see that the learned representation $\psi$ approximately admits the decomposition $\psi = \psi_1\oplus\psi_2$ where $\psi_1$ is the representation learned in the 1D case and $\psi_2$ is the trivial representation. 
Thus, the learned 2D representation is equivalent to the 1D representation in Sec.~\ref{sec:learned_rep}, and hence a 1D order observable is sufficient.

\begin{table}
\centering
\caption{\label{table:generators_2D} %
Estimated 2D latent representation $\psi_g$ by which the learned 2D order observable (i.e., GE-encoder) transforms, for $g\in \{\rho, \tau, \sigma\}$. 
Estimates were averaged over all $24$ trials (eight training data folds and three initialization seeds) as well as over all lattice sizes $L$ and training-validation sample sizes $N$ as they showed little variation; 
the reported uncertainties are standard deviations. 
In both the ferromagnetic and antiferromagnetic cases, the learned representations are equivalent to the 1D representations presented in Sec.~\ref{sec:learned_rep}.
}
\begin{tabular}{ccc}
\toprule
\quad & Ferromagnetic & Antiferromagnetic \\
\midrule
$\psi_{\rho}$ & 
$\lmat
\hphantom{-}1.0000(2) & 0.000(0) \\
\hphantom{-}0.0000(9) & 1.000(0) \\
\rmat$ & $\lmat
-0.96(5) & 0.000(0) \\
\hphantom{-}0.00(0) & 0.994(9) \\
\rmat$ \\
\quad & \quad & \quad \\ 
$\psi_{\tau}$ & 
$\lmat
\hphantom{-}1.0000(2) & 0.000(0) \\
\hphantom{-}0.0000(8) & 1.000(0) \\
\rmat$ & $\lmat
-0.96(5) & 0.000(0) \\
\hphantom{-}0.00(0) & 0.994(9) \\
\rmat$ \\
\quad & \quad & \quad \\ 
$\psi_{\sigma}$ & 
$\lmat
-1.0000(8) & 0.000(0) \\
\hphantom{-}0.0000(0) & 1.000(2) \\
\rmat$ & $\lmat
-0.96(5) & 0.000(0) \\
\hphantom{-}0.00(0) & 0.994(9) \\
\rmat$ \\
\bottomrule
\end{tabular}
\end{table}

\subsection{A general procedure}
\label{sec:general}

We end this section with some remarks on a general procedure for calculating the subgroup of never-broken symmetries and applying the GE-autoencoder method. 
Consider an arbitrary statistical-mechanical system in thermal equilibrium with finite symmetry group $G$, 
and suppose we assume an order observable of the system to have dimension $d$ 
(we can regard $d$ as a hyperparameter that we select as part of the GE-autoencoder architecture). 
Then by Prop.~\ref{prop:chi} (see Appendix~\ref{appendix:chi}), the subgroup $H_d$ of never-broken symmetries can be calculated as
\begin{equation} \label{eq:Hd}
H_d = \bigcap_{\chi\in\hat{G}\mid \degR(\chi)\leq d} \operatorname{ker}(\chi),
\end{equation}
where $\hat{G}$ is the set of all (complex-)irreducible characters of $G$; 
$\operatorname{ker}(\chi)$ is the kernel of the character $\chi$ defined as the preimage set $\chi^{-1}(\degC(\chi))$ of the degree $\degC(\chi) = \chi(1)$; 
and we define $\degR(\chi)$ to be the degree of the smallest real character built out of $\chi$:
\[ \degR(\chi) = \operatorname{deg}(\chi) 
\begin{cases}
1, & \mbox{ if } \IFS(\chi) = 1 \\
2, & \mbox{ otherwise,}
\end{cases} \]
where $\IFS(\chi)$ is the Frobenius-Schur indicator of $\chi$. 
We can therefore quickly calculate the subgroup of never-broken symmetries once we have the character table of $G$.

An immediate corollary of Eq.~\eqref{eq:Hd} (and Prop.~\ref{prop:chi}) is that we have the normal series $H_1\trianglerighteq H_2\trianglerighteq \cdots \trianglerighteq \{1\}$; 
i.e., the subgroup of never-broken symmetries shrinks with increasing dimension of the order observable and eventually becomes trivial, 
so that beyond some finite value of $d$, the GE-autoencoder is no longer able to exploit any never-broken symmetries and is as large as the baseline-autoencoder.

We implemented Eq.~\eqref{eq:Hd} in GAP%
\footnote{GAP is a computer algebra system for computational discrete algebra with particular emphasis on computational group theory~\citep{anon2021gap}.} %
and used our implementation to calculate the subgroup of never-broken symmetries for the Ising symmetry group $G$ for various dimensions $d$ of the order observable (Table~\ref{table:Hd}). 
We note that a limitation of this computational approach is that it requires a numerical value for the lattice size $L$, 
and the computation time increases with $L$. 
We ran our code for lattice sizes $L\in\{4,8,16\}$ and observed that the presentations of the returned subgroups in terms of generators and relations were independent of lattice size, 
allowing us to conclude empirically that these presentations hold for arbitrarily large even lattice size $L$.

The subgroups of never-broken symmetries of the Ising symmetry group for $d\in\{1,2\}$ agree with the theoretical calculations in Props.~\ref{prop:never_broken},\ref{prop:never_broken_2D} (Table~\ref{table:Hd}). 
For $d=3$, we get the same subgroup as with $d=2$. 
For $d=4$, the subgroup of never-broken symmetries becomes trivial, so that the GE-autoencoder is just as large as the baseline-autoencoder; 
this remains the case for all $d\geq 4$ since the subgroups form a descending series with increasing $d$.

\begin{table}
\centering
\caption{\label{table:Hd} %
Subgroup $H_d$ of never-broken symmetries of the Ising symmetry group $G$ assuming various dimensions $d$ of the order observable. 
Here $\beta = \rho\alpha\rho^{-1}$ is the horizontal translation generator. 
The subgroups were computed using our GAP implementation of Eq.~\eqref{eq:Hd} for lattice sizes $L\in\{4,8,16\}$.
}
\begin{tabular}{S[table-format=1]ccc}
\toprule
{$d$} & $|H_d|$ & $|G/H_d|$ & $H_d$ \\
\midrule
1 & $L^2$ & 16 & $\langle\alpha^2, \rho^2, (\alpha\rho)^2\rangle$ \\
2 & $L^2/4$ & 64 & $\langle\alpha^2, \beta^2\rangle$ \\
3 & $L^2/4$ & 64 & $\langle\alpha^2, \beta^2\rangle$ \\
4 & 1 & $16L^2$ & $\{1\}$ \\
\bottomrule
\end{tabular}
\end{table}

Once we have determined the subgroup $H_d$ of never-broken symmetries, we constrain the GE-autoencoder so that the encoder is $H_d$-invariant and the decoder returns an $H_d$-invariant lattice configuration. 
If the encoder and decoder each have one hidden layer of neurons, then the starting point for determining the appropriate parameter constraints is a recent classification of all invariant shallow neural networks~\cite{agrawal2022classification}. 
Finally, once the GE-autoencoder architecture is selected, we train the network by minimizing a loss function with symmetry regularization that generalizes Eq.~\eqref{eq:loss_reduced} for 1D order observables to order observables of any dimension $d$. 
The loss function is given below.

Suppose we have a dataset of $N$ lattice configurations. 
For every $g\in G$, let $Z_g\in\RR^{N\times d}$ whose $n^\text{th}$ row is $\check{O}(g^{-1}\check{x}_n)$; 
i.e., $Z_g$ is the matrix of outputs of the GE-encoder evaluated on the dataset, where all lattice configurations were first transformed by $g^{-1}$ 
(we use $g^{-1}$ instead of $g$ because the GE-encoder outputs are stacked as row vectors in $Z_g$). 
For the identity element $g=1$, let $Z = Z_1$ for further brevity. 
Let $P$ be the orthogonal projection operators onto the null space of $Z$:
\[ P = I - Z^+Z, \]
where $Z^+$ is the Moore-Penrose pseudoinverse of $Z$ and $I$ the identity matrix. 
Finally, let $A$ be a $d\times d$ learnable matrix parameter, 
and define the following estimator of $\psi_g$:
\begin{equation} \label{eq:psihat}
\hat{\psi}_g = Z^+Z_g + P A.
\end{equation}
Then the general loss function for a GE-autoencoder with latent dimension $d$ has the form
\begin{align}
\check{\calL}(\check{\calO}, \check{\calD}, A) 
&= \frac{1}{N}\sum_{n=1}^N L_{\mathrm{BCE}}(\check{\calD}(\check{\calO}(\check{x}_n), T_n), \check{x}_n) \nonumber \\
&+ \lambda\sum_{g\in\Gamma} (R_g + S_g), \label{eq:loss_2D}
\end{align}
where $\Gamma\subset G$ is the minimal set of symmetry generators necessary to check (e.g., for Ising symmetries, $\Gamma = \{\tau,\sigma\}$ for $d=1$ and $\Gamma=\{\rho\tau\sigma\}$ for $d=2$) 
and
\begin{align}
R_g &= \Vert I - \hat{\psi}_g^{\top} \hat{\psi}_g\Vert_F^2 \label{eq:R} \\
S_g &= \frac{\Vert Z_g-Z\hat{\psi}_g \Vert_F^2}{\Vert Z_g\Vert_F^2}, \label{eq:S}
\end{align}
and where $\Vert\cdot\Vert_F^2$ is the squared Frobenius matrix norm (sum of squared matrix elements). 
These regularization terms are obtained by relaxing the hard equivariance constraint $Z_g = Z\psi_g$ to the minimization of $\Vert Z_g-Z\psi_g \Vert_F^2$. 
Specifically, $\hat{\psi}_g$ is the minimizer (i.e., linear least squares estimator); 
$R_g$ ensures $\hat{\psi}_g$ is (approximately) an orthogonal representation; 
and $S_g$ ensures $\Vert Z_g-Z\hat{\psi}_g\Vert_F^2$ is minimized. 
Note the denominator in Eq.~\eqref{eq:S} is included so the optimizer does not simply rescale the GE-encoder.

The above loss function is an approximate generalization of the loss function for the case of a 1D order observable [Eq.~\eqref{eq:loss_reduced}]. 
Specifically, if $d=1$, then $R_g$ and $S_g$ reduce to expressions equivalent to the first and second regularization terms in Eq.~\eqref{eq:loss_reduced}, 
with $S_g$ matching the second regularization term exactly. 
(see Appendix~\ref{appendix:reg_2D} for details). 
Note that Eq.~\eqref{eq:loss_reduced} includes a third regularization term whose original purpose was to ensure a nontrivial representation $\psi$ is learned; 
however, based on a small sample of numerical tests, we believe that this third regularization term is not strictly necessary, and we have thus omitted its generalization from Eq.~\eqref{eq:loss_2D}.

\section{Discussion}
\label{sec:discussion}

We introduced the group-equivariant autoencoder (GE-autoencoder) -- a deep neural network (DNN) architecture that can be used to locate phase transitions by detecting the associated spontaneous symmetry breaking (SSB). We demonstrated its efficacy for the 2D classical ferromagnetic and antiferromagnetic Ising models, finding that the GE-autoencoder (1) accurately determines which symmetries are broken at each temperature, and (2) estimates the critical temperature with greater accuracy and time-efficiency than an SSB-agnostic autoencoder.

We also found the GE-autoencoder to be more robust than the baseline-autoencoder in an interesting sense. 
Recall that we deliberately selected nonlinear architectures for the autoencoders, even though (staggered) magnetization---the ``true'' order observable---is a linear function of the lattice configuration. 
This models the likely scenario in real applications where the DNN being used is more expressive than the unknown order observable we are seeking. 
Ideally, if the DNN is robust, then it should be able to reduce its expressivity to fit the target order observable. 
As discussed in Sec.~\ref{sec:results:temperature}, we suspect that the baseline-autoencoder performs better than the GE-autoencoder on finite lattices but worse in the thermodynamic limit because it learned an inappropriate nonlinear order observable; 
that the GE-autoencoder could accurately learn magnetization without overfitting attests to its superior robustness. 
This robustness is also reflected in the GE-autoencoder's greater sensitivity to the presence of an external symmetry-breaking magnetic field (Sec.~\ref{sec:field}).

There are several implementation details in our method that are worth noting. 
First, in our proof-of-principle example, we found that the $T_{\mathrm{c}}$ estimation by fitting a step function to the fourth Binder cumulant vs.\ temperature curve works well, 
and it still allows for finite-size scaling analysis (see Sec.~\ref{sec:results:extrapolate}). 
In contrast, the intersection of second Binder cumulant curves across lattice sizes is not guaranteed for autoencoders, 
and the intersection points of fourth Binder cumulant curves across lattice sizes is not guaranteed to be unique. 
Second, proper initialization and learning rate settings (see Appendix~\ref{appendix:lr}) were critical for a fair comparison of the baseline-autoencoder and GE-autoencoder. 
Our approach to determining these settings could be useful for deep learning experiments in general, where one DNN is a constrained copy of another. 
Finally, in contrast to ML methods for phase detection that are entirely data-driven, our method includes group-theoretic considerations that allow us to exploit some of the structure available in the problem, namely symmetries; 
given the benefits we have identified here, we think this practice should always be followed whenever possible.

Our work has several physical implications as well. For example, recent progress has been made in transfer-learning from small to large lattice sizes implementing ideas of block decimation from the renormalization group as a way to more efficiently extrapolate to the thermodynamic limit \cite{Efthymiou:2019uy}. The GE-autoencoder, on the other hand, is scale-independent and could therefore be well-suited for this transfer-learning task. 
Indeed, we found that the multiscale GE-autoencoder is more time-efficient and slightly more accurate than the single-scale GE-autoencoder. 
In future work, we aim to improve the multiscale GE-autoencoder until it saves us from having to simulate the largest lattice size (e.g., $L=128$) entirely. 
For example, researchers have used different generative ML methods to learn the distribution of lattice configurations of a system, allowing them to simulate the system more efficiently than with traditional MC methods~\cite{PhysRevB.95.041101, ShenPRB2018, Li2019ann, AlbergoPRD2019, PhysRevB.101.115111, PhysRevB.98.041102}. Combining these methods with ideas introduced in this paper could allow us to generate large sample lattices based only on a dataset of small-to-moderate lattices, possibly offering a significant speedup over MC methods.

The superior accuracy of the GE-autoencoder could also translate into greater robustness against the sign problem. There is empirical evidence suggesting that DNNs could overcome the sign problem in the single-band Hubbard model to some extent~\cite{broecker2017machine}. Subsequent work, however, showed that the sign problem returns at more extreme temperatures and doping~\cite{ch2017machine}, although this work also demonstrates that transfer-learning from sign-problem-free regions to sign-problem-prone regions in the phase diagram could work for small amounts of doping. 
If the GE-autoencoder is indeed more robust than previous DNN methods, then perhaps it could allow us to access even more extreme temperatures and doping levels.

Finally, knowledge of the SSB associated to a phase transition has value beyond a means to phase detection alone. For example, identifying the relevant symmetries could help elucidate the mechanism driving subtle phase transitions and thus offer a means to control them in practice to realize real-world applications.

Lastly, we end with remarks on some potential improvements to our approach to be addressed in the future. 
First, we would like to formalize, further develop, and better understand the general procedure presented in Sec.~\ref{sec:general}. 
For example, the computational approach described in Sec.~\ref{sec:2D} for finding the subgroup of never-broken symmetries of an arbitrary finite symmetry group requires that we specify a value for the system size; 
understanding how the never-broken symmetries depend on system size would thus require us to test various system sizes and then search for a pattern, 
and it is unclear how feasible this would be for general systems. 
On the other hand, if we only consider a class of systems for which the dependence of the symmetry group on the system size takes a particular form, then it may be feasible to derive a specialized algorithm for calculating the never-broken symmetries for arbitrary system sizes.

Second, we would like to consider example systems with more complicated symmetry groups. 
This could include those with continuous internal symmetries and gauge symmetries; 
however, a first step would be to consider finite nonabelian internal symmetry groups that must act on vector-valued lattice configurations.

Third and finally, we would like to better understand how to choose the representation by which the never-broken symmetries are imposed on the GE-autoencoder. 
Relating to this is our finding that the baseline-autoencoder estimates the critical temperature with greater accuracy than the GE-autoencoder for individual  lattice sizes (i.e., without finite-size scaling analysis). This aspect was  discussed in Sec.~\ref{sec:results:temperature}. As speculated there, this observation could be due to the greater flexibility of the baseline-autoencoder. 
However, we suspect that this greater flexibility is also why the baseline-autoencoder exhibits a slightly sharper transition in its order parameter vs.\ temperature curve, while the GE-autoencoder learns an order parameter very close to that of magnetization. 
We thus suspect that the baseline-autoencoder is more prone to incorrectly classifying the Ising phase transition as first-order instead of second-order 
(as was observed in Ref.~[\citenum{alexandrou2020critical}]) 
when compared to the GE-autoencoder.  If greater flexibility is desired, we believe the better way is to use a more flexible GE-autoencoder; this could become important for cases where the nature of the phase transition (i.e. continuous or weakly first order) is under dispute.  The network complexity required for the GE-autoencoder to match the critical-temperature-estimation performance of the baseline-autoencoder is likely related to our choice of how we impose the never-broken symmetries. We intend to investigate this relationship in future work.

Much work remains to be done to establish the GE-autoencoder as a mature and trusted machine learning model; however, the proof-of-principle application presented here provides optimism for its utility in addressing open problems related to the search for spontaneous symmetry breaking in the pseudogap phase of the high temperature cuprates, or at topological phase transitions.

\section*{Acknowledgements}
Work by S.~J. and A.~D. was supported by the U.S. Department of Energy, Office of Science, Office of Basic Energy Sciences, under Award Number DE-SC0022311. Work by D.~A. and J.~O. was supported by the U.S. Department of Energy, Office of Science,  Advanced Scientific Computing Research, under award Number DE-SC0018175.

\appendix

\section{Propositions}
\label{appendix:propositions}

\subsection{Subgroup of never-broken Ising symmetries}
\label{appendix:propositions:never}

Recall from Sec.~\ref{sec:ssb} that a symmetry $g\in G$ is never-broken if $\psi_g=1$ (i.e., if $g$ is in the kernel of $\psi$, denoted $\operatorname{ker}(\psi)$). 
In the absence of any knowledge about the true representation $\psi$ associated with the order parameter, 
we can deduce a subgroup of never-broken symmetries by finding all symmetries $g\in G$ such that $\psi_g=1$ (i.e., $g\in \operatorname{ker}(\psi)$) for all representations $\psi:G\mapsto\{-1, 1\}$. 
The following proposition establishes such a subgroup for the Ising symmetry group.

\begin{prop} \label{prop:never_broken}
Let $\Psi$ be the set of all real scalar representations $\psi:G\mapsto\{-1, 1\}$ of the Ising symmetry group $G$. 
Then $\bigcap_{\psi\in\Psi} \operatorname{ker}(\psi) = \langle\alpha^2,\rho^2,(\alpha\rho)^2\rangle$.
\end{prop}
\begin{proof}
Let $g\in G$ such that $g=h^2$ for some $h\in G$. 
Then for every $\psi\in\Psi$, we have
\[ \psi_g = \psi_{h^2} = \psi_h^2 = (\pm 1)^2 = 1, \]
and hence $g\in\bigcap_{\psi\in\Psi}\operatorname{ker}(\psi)$; 
in particular, we have
\[ \langle\alpha^2, \rho^2, (\alpha\rho)^2\rangle \subseteq \bigcap_{\psi\in\Psi}\operatorname{ker}(\psi). \]
All that remains to show is the reverse inclusion. 
Let $g\in\bigcap_{\psi\in\Psi}\operatorname{ker}(\psi)$, 
so that $\psi_g=1$ for every $\psi\in\Psi$. 
By the defining relations of $G$, $g$ admits the expression $g = h\tau^{m_3}\sigma^{m_4}$, 
where $h$ is a product of $m_1$ copies of $\alpha$ and $m_2$ copies of $\rho$ in some order. 
Now since $\psi$ takes values in an Abelian group, then for all $g_1,g_2\in G$,
\[ \psi_{g_1g_2} = \psi_{g_1}\psi_{g_2} = \psi_{g_2}\psi_{g_1} = \psi_{g_2g_1}. \]
Using this fact, we have
\begin{align}
\psi_g 
&= \psi_h \psi_{\tau^{m_3}} \psi_{\sigma^{m_4}} \nonumber \\
&= \psi_{\alpha^{m_1}\rho^{m_2}} \psi_{\tau^{m_3}} \psi_{\sigma^{m_4}} \nonumber \\
&= \psi_{\alpha}^{m_1} \psi_{\rho}^{m_2} \psi_{\tau}^{m_3} \psi_{\sigma}^{m_4} = 1. \label{eq:mi} \\
\end{align}
This holds for every $\psi\in\Psi$. 
Now each $\psi\in\Psi$ is completely determined by its values on the generators $\alpha$, $\rho$, $\tau$, and $\sigma$. 
Moreover, it is easy to verify that if we apply $\psi$ to every defining relation of $G$, 
then the resulting equations are satisfied for all choices of $\psi_{\alpha}$, $\psi_{\rho}$, $\psi_{\tau}$, and $\psi_{\sigma}$. 
Thus, $\Psi$ is precisely the set of homomorphisms $\psi:G\mapsto\{-1, 1\}$ determined by every combination of values in $\{-1, 1\}$ on the four generators of $G$. 
Let $\psi_i\in\Psi$ such that $\psi_i$ is $-1$ on the $i$th generator and $1$ on the other three generators. 
Evaluating Eq.~\eqref{eq:mi} at $\psi=\psi_i$, we get $(-1)^{m_i} = 1$, 
implying that each $m_i$ is even. 
The element $g$ thus takes the form $g = h\tau^{2n_3}\sigma^{2n_4} = h$, 
where $h$ is a product of $2n_1$ copies of $\alpha$ and $2n_2$ copies of $\rho$ in some order.

On the other hand, let $\beta = \rho\alpha\rho^{-1}$, and note that $\alpha$ and $\beta$ commute. 
Then $G = (\langle\alpha,\beta\rangle\rtimes\langle\rho,\tau\rangle)\times\langle\sigma\rangle$, 
implying that each $g\in G$ admits the unique representation 
\[ g = \alpha^{q_1}\beta^{q_2}\rho^{q_3}\tau^{q_4}\sigma^{q_5}. \]
If $\psi_g=1$, then we have already deduced that $g$ must be a product of $\alpha$'s and $\rho$'s, 
and hence $g=\alpha^{q_1}\beta^{q_2}\rho^{q_3}$. 
Noting that $\beta^{q_2}=\rho\alpha^{q_2}\rho^{-1}$, 
this expression of $g$ has $q_1+q_2$ copies of $\alpha$ and $1-1+q_3=q_3$ copies of $\rho$. 
Equating these to the previously obtained exponents in the expression of $g$, we have $q_1+q_2=2n_1$ and $q_3=2n_2$, 
implying that $q_3$ is even and that $q_1$ and $q_2$ have the same parity. 
Now using the fact that $\alpha$ and $\beta$ commute, $g$ admits the form
\begin{align*}
g 
&= \alpha^{q_1-q_2}\alpha^{q_2}\beta^{q_2}\rho^{q_3} \\
&= \alpha^{q_1-q_2}(\alpha\beta)^{q_2}\rho^{q_3} \\
&= (\alpha^2)^{n_1-q_2}(\alpha\beta)^{q_2}(\rho^2)^{n_2}, 
\end{align*}
implying $g\in\langle\alpha^2,\rho^2,\alpha\beta\rangle$. 
Finally, note that
\begin{align*}
\alpha\beta 
&= \alpha\rho\alpha\rho^{-1} \\
&= \alpha\rho\alpha\rho^3 \\
&= \alpha\rho\alpha\rho\rho^2 \\
&= (\alpha\rho)^2\rho^2.
\end{align*}
Since $\rho^2\in\langle\alpha^2,\rho^2,\alpha\beta\rangle$, then $\langle\alpha^2,\rho^2,(\alpha\rho)^2\rangle$, 
and therefore $g\in\langle\alpha^2,\rho^2,(\alpha\rho)^2\rangle$, establishing
\[ \bigcap_{\psi\in\Psi}\operatorname{ker}(\psi) \subseteq \langle\alpha^2,\rho^2,(\alpha\rho)^2\rangle. \]
\end{proof}

Intuitively, even for a general observable dimension $d$, the observable $\calO:\XX\mapsto\RR^d$ reduces dimensionality from $L^2$ to $d$ where $L$ may be arbitrarily large in the thermodynamic limit. 
If $d$ is less than the minimum dimensionality required for $\psi$ to be a faithful representation, then some subgroup of $G$ will necessarily be modded out by $\psi$.

It can be shown that $\scb(L) \trianglelefteq G$, and hence we have the quotient group
\begin{align} \label{eq:quotient_group}
    G/H &= \langle\alpha H, \rho H, \tau H, \sigma H\rangle \nonumber \\
        &= \{\alpha^{m_1}\rho^{m_2}\tau^{m_3}\sigma^{m_4} H: m_i\in \{0, 1\}\},
        \nonumber \\
    H &= \scb(L).
\end{align}
This quotient group represents the reduction in the number of symmetries we will need to check for SSB; 
indeed, $|G| = 16L^2$ and $|H|=L^2$, so that $|G/H| = 16$.

\subsection{Constraint on the representations of spatial symmetries}
\label{appendix:propositions:art}

The form of the reduced encoder [Eq.~\eqref{eq:encoder_reduced}] places additional constraints on the spatial symmetries in the quotient group [Eq.~\eqref{eq:quotient_group}], as specified in the following proposition.

\begin{prop} \label{prop:art}
Suppose $\check{\calO}:[-1, 1]^2\mapsto\RR$ is nonzero for at least one lattice configuration. 
For every $\psi\in\Psi$ such that $\check{\calO}(g\check{x}) = \psi_g\check{\calO}(\check{x})$ for all Ising symmetries $g\in G$, 
we have $\psi_{\alpha} = \psi_{\rho} = \psi_{\tau}$.
\end{prop}
\begin{proof}
Clearly $\alpha$ (downward translation), $\rho$ ($90^\circ$-rotation), and $\tau$ (reflection) map black squares to white squares and vice versa on an $L\times L$ checkerboard where $L$ is even. 
Thus, for all $g\in\{\alpha, \rho, \tau\}$,
\begin{align*}
g(\checkb{x}, \checkw{x}) &= (\checkw{x}, \checkb{x}) \\
\check{\calO}(g(\checkb{x}, \checkw{x})) &= \check{\calO}((\checkw{x}, \checkb{x})) \\
\psi_g\check{\calO}((\checkb{x}, \checkw{x})) &= \check{\calO}((\checkw{x}, \checkb{x})),
\end{align*}
and hence,
\[ \psi_{\alpha}\check{\calO}((\checkb{x}, \checkw{x})) = \psi_{\rho}\check{\calO}((\checkb{x}, \checkw{x})) = \psi_{\tau}\check{\calO}((\checkb{x}, \checkw{x})). \]
Evaluating this on a lattice configuration $\mathbf{x}$ for which $\check{\calO}$ is nonzero, we obtain $\psi_{\alpha} = \psi_{\rho} = \psi_{\tau}$ as desired.
\end{proof}

The key idea is that $g(\checkb{x}, \checkw{x}) = (\checkw{x}, \checkb{x})$ for all $g\in\{\alpha,\rho,\tau\}$. 
The upshot is that we now need only estimate $\psi_{\sigma}$ and one of $\psi_{\alpha}$, $\psi_{\rho}$, and $\psi_{\tau}$ from the data.

\subsection{Unsupervised-learning of symmetries}
\label{appendix:propositions:invariant}

The following proposition states that under suitable conditions, if an unsupervised model is fit to a dataset containing symmetries, then the fit model will be invariant to those symmetries at least when restricted to the dataset.

\begin{prop} \label{prop:invariant_min}
Let $G$ be a finite group and $\calX = \{\mathbf{x}_n\in\RR^m\}_{n=1}^N$ a $G$-invariant dataset; 
i.e., $g\mathbf{x}_n\in\calX$ for every $g\in G$ and $\mathbf{x}_n\in\calX$. 
Let $\calF$ be a convex set of functions $f:\calX\mapsto\RR^m$ such that 
$gfg^{-1}\in\calF$ for every $f\in\calF$ and $g\in G$. 
Let $L_{\mathrm{metric}}:\RR^m\times\calX\mapsto\RR$ be a function such that
\begin{enumerate}
\item $L_{\mathrm{metric}}(g\mathbf{x}, g\mathbf{x}^{\prime}) = L_{\mathrm{metric}}(\mathbf{x}, \mathbf{x}^{\prime})$ for all $g\in G$ and $(\mathbf{x}, \mathbf{x}^{\prime})\in\RR^m\times\calX$, and
\item $L_{\mathrm{metric}}$ is strictly convex in its first argument.
\end{enumerate}
Then every global minimizer $f_*\in\calF$ of the loss function $\calL:\calF\mapsto\RR$ given by
\[ \calL(f) = \frac{1}{N}\sum_{n=1}^N L_{\mathrm{metric}}(f(\mathbf{x}_n), \mathbf{x}_n) \]
is $G$-equivariant on the dataset $\calX$.
\end{prop}
\begin{proof}
First we establish that the loss function $\calL$ is both strictly convex and $G$-invariant. 
For each $n=1,\ldots,N$, it is easy to show that the map from $\calF$ to $\RR$ given by $f\rightarrow L_{\mathrm{metric}}(f(\mathbf{x}_n), \mathbf{x}_n)$ is strictly convex. 
(Note that this relies on the hypothesis that the domain of the functions in $\calF$ is the dataset $\calX$, 
so that for $f_1,f_2\in\calF$, $f_1=f_2$ if and only if $f_1(\mathbf{x}_n)=f_2(\mathbf{x}_n)$ for all $n=1,\ldots,N$.) 
Moreover, since $\calL$ is a convex combination of such maps, then it is also strictly convex. 
Now for every $f\in\calF$ and $g\in G$, we have
\begin{align*}
\calL(gfg^{-1}) 
&= \frac{1}{N}\sum_{n=1}^N L_{\mathrm{metric}}(gf(g^{-1}\mathbf{x}_n), \mathbf{x}_n) \\
&= \frac{1}{N}\sum_{n=1}^N L_{\mathrm{metric}}(gf(g^{-1}\mathbf{x}_n), g(g^{-1}\mathbf{x}_n)) \\
&= \frac{1}{N}\sum_{n=1}^N L_{\mathrm{metric}}(gf(\mathbf{x}_n), g\mathbf{x}_n)  \tag{$g^{-1}\mathbf{x}_n\rightarrow \mathbf{x}_n$} \\
&= \frac{1}{N}\sum_{n=1}^N L_{\mathrm{metric}}(f(\mathbf{x}_n), \mathbf{x}_n) \\
&= \calL(f),
\end{align*}
where we used the fact that the dataset $\calX$ is $G$-invariant in the reindexing step. 
Thus, $\calL$ is $G$-invariant, where each $g\in G$ acts on $\calL$ by conjugation of its argument. 

Now let $f_*\in\calF$ be a global minimizer of $\calL$. 
By $G$-invariance, $gf_*g^{-1}$ is also a global minimizer for every $g\in G$, 
and by convexity, $\frac{f_*+gf_*g^{-1}}{2}$ is a global minimizer for each $g\in G$ as well. 
On the other hand, convexity and $G$-invariance together imply that
\begin{align*}
\calL\left(\frac{f_*+gf_*g^{-1}}{2}\right) 
&\leq \frac{\calL(f_*) + \calL(gf_*g^{-1})}{2} \\
&= \frac{\calL(f_*) + \calL(f_*)}{2} \\
&= \calL(f_*),
\end{align*}
but since $f_*$ and $\frac{f_*+gf_*g^{-1}}{2}$ are both global minimizers, then this inequality is in fact an equality:
\[ \calL\left(\frac{f_*+gf_*g^{-1}}{2}\right) = \frac{\calL(f_*) + \calL(gf_*g^{-1})}{2}. \]
Finally, since $\calL$ is strictly convex, then this necessitates $f_* = gf_*g^{-1}$ 
or equivalently $f(g\mathbf{x}_n) = gf(\mathbf{x}_n)$ for all $g\in G$ and $\mathbf{x}_n\in\calX$, 
thereby proving the claim.
\end{proof}

Note that the functions in $\calF$ are restricted to the dataset $\calX$; 
the minimizer $f_*$ is therefore not guaranteed to be $G$-equivariant off the dataset. 

For our particular case, the function $L_{\mbox{metric}}$ is the binary cross-entropy $L_{\mathrm{BCE}}$ [Eq.~\eqref{eq:bce}]. 
By its form in Eq.~\eqref{eq:bce}, $L_{\mathrm{BCE}}$ is clearly invariant under the spin-flip symmetry $(\mathbf{\hat{x}}, \mathbf{x})\rightarrow (-\mathbf{\hat{x}}, -\mathbf{x})$, 
and it is manifestly invariant under all spatial symmetries as it is a sum over the lattice sites; 
thus, $L_{\mathrm{BCE}}$ is $G$-invariant where $G$ is the Ising symmetry group. 
The strict convexity of $L_{\mathrm{BCE}}$ is evident; 
since $x_{\bf i}=\pm 1$, then each summand in Eq.~\eqref{eq:bce} is either $-\log\left(\frac{1+\hat{x}_{\bf i}}{2}\right)$ or $-\log\left(\frac{1-\hat{x}_{\bf i}}{2}\right)$, 
and in either case, its second derivative with respect to $\hat{x}_{\bf i}$ is strictly positive, so that each summand is strictly convex. 
The set $\calF$ is the set of functions (restricted to our dataset) expressible as autoencoders of fixed depth but arbitrary widths; 
we allow for arbitrary widths to ensure that any convex combination of autoencoders is again expressible as a single autoencoder in $\calF$. 
In practice, however, we make the assumption that the global minimizer $f_*\in\calF$ can be accessed by the single autoencoder architecture we proposed. 
There are additional caveats to Prop.~\ref{prop:invariant_min}: 
First, our dataset of lattice sets is probably only approximately $G$-invariant; 
second, actually finding $f_*$ is nontrivial since the map from its network parameters to the autoencoder $f$ is in general nonconvex. 
Nevertheless, we take Prop.~\ref{prop:invariant_min} as strong motivation to make the assumption that our autoencoder $f = \calD\circ\calO$ is $G$-invariant.

\section{Network initialization and learning rates}
\label{appendix:lr}

Here we provide details on the network parameter initialization and learning rate settings needed for a fair comparison of the baseline-autoencoder and GE-autoencoder. 
The GE-autoencoder is a small network that is equivalent to the larger baseline-autoencoder with the parameters of its first encoding layer and last decoding layer constrained to a checkerboard pattern (ignoring the symmetry regularization terms in the GE-autoencoder loss function). 
The idea is to initialize and set the learning rates of the baseline-autoencoder such that, if the checkerboard constraint were imposed and maintained on the baseline-autoencoder, then it would be and would remain functionally equivalent to the GE-autoencoder at initialization time and throughout training.

\paragraph*{Initialization} %
We initialize the weight matrix and bias vector of each layer of the GE-autoencoder with IID values sampled under a uniform distribution over $\left[-\frac{1}{\sqrt{h_{\mathrm{in}}}}, \frac{1}{\sqrt{h_{\mathrm{in}}}}\right]$, 
where $h_{\mathrm{in}}$ is the input dimension of the layer. 
If $(u_k, v_k)$ are the initial weights of the first encoding layer of the GE-autoencoder incident to the $k$th hidden neuron, then we initialize the weights of the first encoding layer of the baseline-autoencoder according to Eq.~\eqref{eq:uivi}. 
We initialize the weight matrix and bias vector of the last decoding layer of the baseline-autoencoder by simply tiling the initialized parameters of the last decoding layer of the GE-autoencoder according to a checkerboard pattern. 
All remaining parameters of the baseline-autoencoder are structurally equivalent to those of the GE-autoencoder, 
and we thus initialize them to be equal to the corresponding initialized parameters in the GE-autoencoder.

\paragraph*{Learning rate} %
We set the learning rate of all parameters of the GE-autoencoder to $\eta_{\mathrm{GE}} = 0.001$. 
We proceed to deduce the appropriate learning rate $\eta$ for the weights in the first encoding layer of the baseline-autoencoder; 
we do this for the weights on ``black squares'' ($i_x+i_y$ even; see Eq.~\eqref{eq:uivi}); 
the argument for ``white squares'' is analogous. 
Assuming full-batch gradient descent for simplicity and ignoring symmetry regularization, the update rules for the weights in the first encoding layers of the baseline-autoencoder and GE-autoencoder are
\[ \Delta w_{k,{\bf i}} = -\eta\frac{\partial\calL}{\partial w_{k,{\bf i}}} \mbox{ and }
\Delta u_k = -\eta_{\mathrm{GE}}\frac{\partial\calL}{\partial u_k}. \]
However, by Eq.~\eqref{eq:uivi}, we have $\Delta w_{k,{\bf i}} = \frac{2}{L^2}\Delta u_k$ and hence
\[ \eta\frac{\partial\calL}{\partial w_{k,{\bf i}}} = \frac{2}{L^2}\eta_{\mathrm{GE}}\frac{\partial\calL}{\partial u_k}. \]
Applying the chain rule to the right side, we have
\begin{align*}
\eta\frac{\partial\calL}{\partial w_{k,{\bf i}}} 
&= \frac{2}{L^2}\eta_{\mathrm{GE}}\sum_{j_x+j_y \mbox{ even}}\frac{\partial\calL}{\partial w_{k,{\bf j}}}\frac{\partial w_{k,{\bf j}}}{\partial u_k} \\
&= \frac{2}{L^2}\eta_{\mathrm{GE}}\left(\frac{2}{L^2}\sum_{j_x+j_y \mbox{ even}}\frac{\partial\calL}{\partial w_{k,{\bf j}}}\right).
\end{align*}
The quantity in parentheses is the average gradient over black squares. 
By averaging both sides over ${\bf j}$ with $j_x+j_y$ even, we can obtain such an average gradient on both sides; 
cancelling then leaves us with the desired learning rate:
\[ \eta = \frac{2}{L^2}\eta_{\mathrm{GE}}. \]
As an example, for lattice size $L=128$, the learning rate of the weights in the first encoding layer of the baseline-autoencoder should be set to approximately $1.2\times 10^{-7}$-- 
a value unlikely to be found by hand. 
By the same argument as above, we can show that the learning rate of the parameters in the last decoding layer of the baseline-autoencoder should be set equal to $\eta_{\mathrm{GE}}$. 
Finally, all remaining parameters of the baseline-autoencoder are structurally equivalent to those of the GE-autoencoder, 
and hence we set all of their learning rates equal to $\eta_{\mathrm{GE}}$ as well.

\section{Functional comparison to magnetization and Onsager's solution}
\label{appendix:onsager}

Here we regard the order observables (magnetization, baseline-encoder, and GE-encoder) as elements of a Hilbert space and measure the distance from each of them to magnetization. 
Given any Euclidean vector space $\RR^n$ with inner product $\langle\cdot, \cdot\rangle$, we define a measure of error $\nu:\RR^n\times\RR^n\mapsto [0, 1]$ as follows:
\[ \nu(\mathbf{v}, \mathbf{w}) = \sin^2\theta = 1-\left(\frac{\langle \mathbf{v},\mathbf{w}\rangle}{\Vert \mathbf{v}\Vert \Vert \mathbf{w}\Vert}\right)^2, \]
where $\theta$ is the angle between vectors $\mathbf{v}$ and $\mathbf{w}$. 
As $\nu$ takes values in $[0, 1]$, we will express it as a percentage. 
Note that $\nu(\mathbf{v}, \mathbf{w})$ is independent of the signs and norms of $\mathbf{v}$ and $\mathbf{w}$-- a desirable property for our purpose. 
We evaluate each order observable $\calO$ on a test set of $n=204,800$ lattice configurations ($2,048$ configurations sampled at each of $100$ temperatures) to obtain a vector of measurements $\hat{\calO}\in\RR^n$. 
We then evaluate $\nu(\hat{\calO}, \hat{\calO_M})$ for each order observable $\calO$, where $\calO_M$ is magnetization. 
We trivially have $\nu(\hat{\calO_M}, \hat{\calO_M})=0$, but we also obtain an error around $0.0 - 0.1\%$ or $0.001$ for both the single-scale and multiscale GE-encoders across all lattice sizes and training-validation sample sizes, 
indicating that the GE-encoders learn a function very similar to magnetization. 
In contrast, the baseline-encoder results in a larger error and is thus more distinct from magnetization~(Fig.~\ref{fig:cor_magnetization}); 
in the ferromagnetic case, this error even increases with training-validation sample size, meaning that the baseline-encoder moves away from magnetization as it sees more data. 
This is also consistent with the divergence between the baseline-encoder and magnetization in terms of their critical temperature estimates~(Fig.~\ref{fig:tc_ferro}).

\begin{figure}
\centering
\includegraphics[width=1.0\columnwidth]{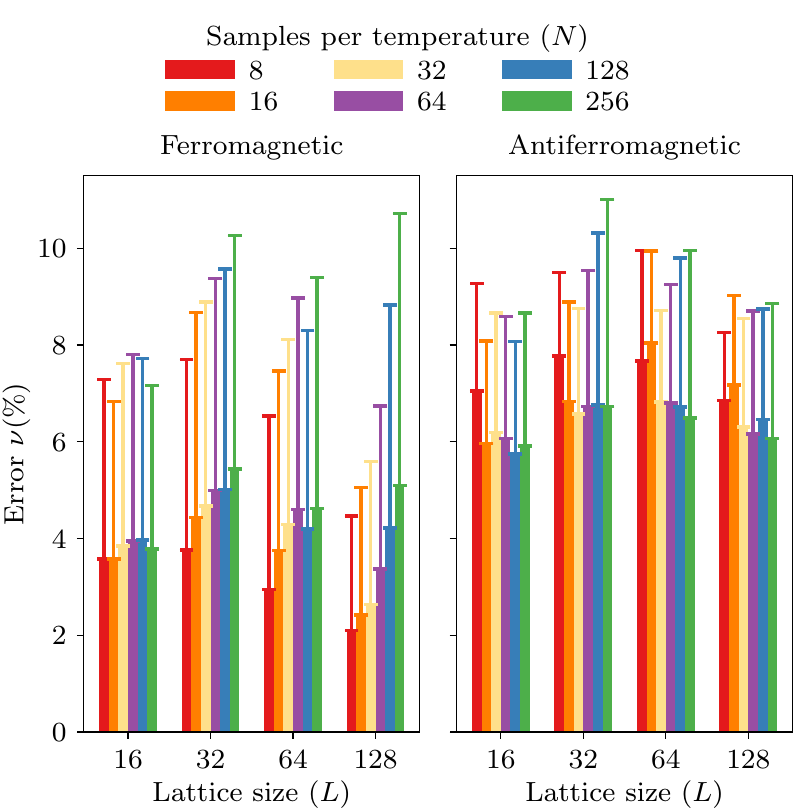}
\caption{\label{fig:cor_magnetization} %
Measure of error between the baseline-encoder and magnetization. 
Error bars indicate the standard deviation across $24$ trials (eight training data folds and three initialization seeds). 
The baseline-encoder is distinct from magnetization (positive error $\nu$) and becomes increasingly distinct with increasing number of training-validation samples in the ferromagnetic case.
}
\end{figure}

We now compare the order parameter $\langle |\calO|\rangle$ derived from each order observable (Fig.~\ref{fig:orders}) to Onsager's solution $M_{\mathrm{ONS}}$ [Eq.~\eqref{eq:onsager}]. 
We evaluate the order parameters and Onsager's solution on the $100$ temperatures in our dataset to obtain vectors in $\RR^{100}$. 
We then evaluate $\nu(\langle|\calO|\rangle, M_{\mathrm{ONS}})$ independently for each jackknife-sample of the order parameter $\langle|\calO|\rangle$. 
We perform least-squares regression on the calculated error $\nu$ against inverse lattice size, finding an approximately linear relationship with the baseline-encoder having a slightly weaker linear dependence (in terms of $r^2$ value) than magnetization and the GE-encoders; 
we visualize this for the maximal case of $N=256$ training-validation samples per temperature~(Fig.~\ref{fig:cor_onsager}). 
Extrapolating to infinite lattice size (i.e., the thermodynamic limit), we find that the order parameters derived from magnetization and the GE-encoders almost converge to Onsager's solution, while the baseline-encoder converges further away from Onsager's solution.

\begin{figure}
\centering
\includegraphics[width=1.0\columnwidth]{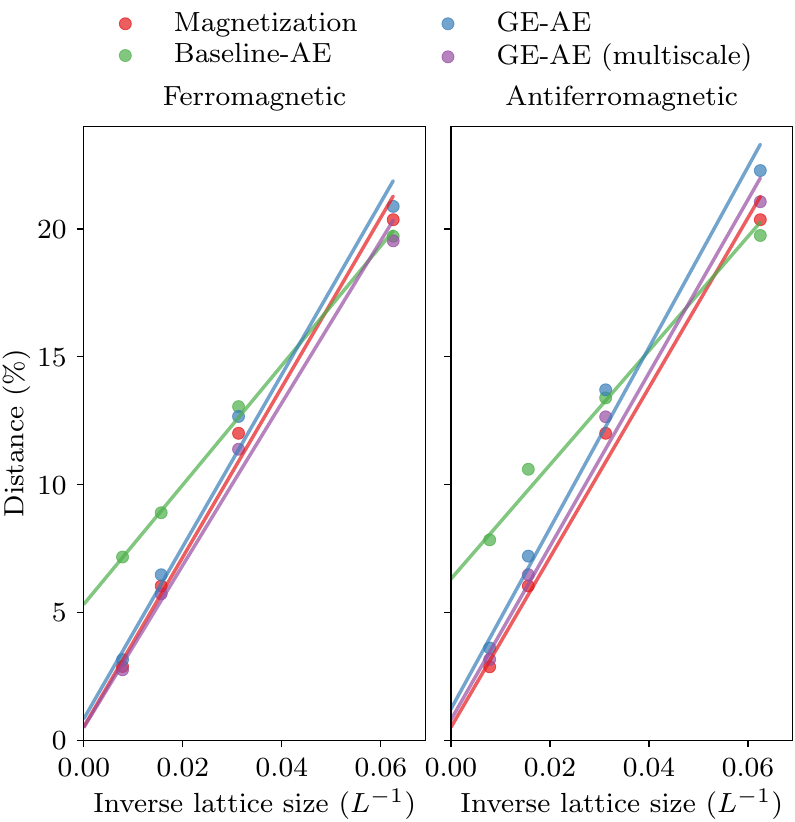}
\caption{\label{fig:cor_onsager} %
Measure of error between the order parameters derived from each order observable and Onsager's solution using $N=256$ training-validation samples per temperature and averaged over $24$ trials (eight training data folds and three initialization seeds). 
The error exhibits an approximately linear dependence on inverse lattice size, with stronger linearity for magnetization and the GE-encoders. 
At infinite lattice size ($L^{-1}=0$), magnetization and the GE-encoders converge more closely to Onsager's solution than does the baseline-encoder.
}
\end{figure}

\section{Obtaining point estimates of the critical temperature}
\label{appendix:optimization}

Here we provide the details of obtaining a point estimate of the critical temperature given interval estimates as described in Sec.~\ref{sec:results:temperature}. 
Let $T_1 < T_2 < \ldots < T_M$ be the temperatures at which we ran MC simulations to generate our dataset, 
and suppose we have $N$ MC-sampled lattice configurations from each of these temperatures. 
Let $t_0\in [T_{m_0}, T_{m_0+1}]$ be the interval estimate of the critical temperature obtained from the Binder cumulant curve based on all available data, in the way described in Sec.~\ref{sec:results:temperature}, 
and let $t_i\in [T_{m_i}, T_{m_i+1}]$ be the interval estimate of the critical temperature obtained from the $i$th jackknife-sample Binder cumulant curve, for $i=1,\ldots,N$. 
We seek the critical temperature point estimate that is optimally stable; 
we do so by minimizing the jackknife variance subject to the interval estimates:
\[ \min_{t_0,\ldots,t_N} \sum_{i=1}^N(t_i-t_0)^2 \mbox{ s.t.} \]
\[ T_{m_i} \leq t_i \leq T_{m_i+1}, \mbox{ for } i=0,\ldots,N. \]
This problem can be expressed more elegantly with vector notation. 
Let $\mathbf{t},\mathbf{a},\mathbf{b}\in\RR^{N+1}$ with elements $t_i$, $T_{m_i}$, and $T_{m_i+1}$ respectively, 
and define the matrix $\mathbf{A}\in\RR^{(N+1)\times (N+1)}$ with elements
\begin{align*}
A_{00} &= N \\
A_{i0} &= A_{0i} = -1 \\
A_{ii} &= 1, \mbox{ for } i=1,\ldots,N,
\end{align*}
and zero for all remaining elements. 
Then the above optimization problem can be expressed as
\[ \min_{\mathbf{t}} \mathbf{t}^\top\mathbf{A}\mathbf{t} \mbox{ s.t. } \mathbf{a} \leq \mathbf{t} \leq \mathbf{b}. \]

The matrix $\mathbf{A}$ admits the factorization $\mathbf{A} = \mathbf{B}^\top \mathbf{B}$, 
where $\mathbf{B}\in\RR^{N\times (N+1)}$ with block structure $\mathbf{B} = [\mathbf{1}\mid \mathbf{I}]$, 
where $\mathbf{1}$ is an $N$-dimensional vector of $1$'s and $\mathbf{I}$ is the $N\times N$ identity matrix. 
The matrix $\mathbf{B}$ clearly has rank $N$, 
and hence $\mathbf{A}$ is a symmetric positive semidefinite matrix of corank $1$. 
The above constrained optimization problem is thus a convex quadratic program, which we efficiently solve numerically using the coneqp solver available in cvxopt~\cite{vandenberghe2010cvxopt}. 

Let $\mathbf{t_{\mathrm{sol}}}$ be the numerical solution obtained. 
Since $\mathbf{A}$ has corank $1$, then the solution set of the optimization problem is at most a line segment (one degree of freedom). 
It is easy to see that if $\mathbf{t_{\mathrm{sol}}}$ is an interior solution, then $\mathbf{t_{\mathrm{sol}}}\pm \eps\mathbf{1}$ is a solution as well for sufficiently small $\eps > 0$, 
where $\mathbf{1}$ is an $(N+1)$-dimensional vector of $1$'s. 
If we set
\[ \eps_1 = \min(\mathbf{t_{\mathrm{sol}}}-\mathbf{a}) \mbox{ and } 
\eps_2 = \min(\mathbf{b}-\mathbf{t_{\mathrm{sol}}}), \]
then the complete solution set is
\[ \{(1-\alpha)(\mathbf{t_{\mathrm{sol}}}-\eps_1\mathbf{1}) + \alpha(\mathbf{t_{\mathrm{sol}}}+\eps_2\mathbf{1}): 0 \leq \alpha \leq 1\}. \]
We select the midpoint $\mathbf{t_*}$ as the vector of jackknife critical temperature point estimates.

The jackknife mean and variance of the critical temperature estimate are then
\begin{align*}
\mathrm{mean} &= \overline{t_*} + B \\
\mathrm{variance} &= \sum_{i=1}^N (t_{*i}-\overline{t_*})^2 + \frac{B^2}{N},
\end{align*}
where $\overline{t_*} = \frac{1}{N}\sum_{i=1}^N t_{*i} $ 
is the mean of the jackknife samples (not including the overall estimate based on all $N$ samples), and $B = N (t_0 - \overline{t_*})$ 
is a term added to reduce the bias in the jackknife mean, at the cost of incurring additional variance. 
Note that the above expression for the variance is equivalent to the objective function of the optimization problem we solved. 
As already discussed in Sec.~\ref{sec:results:temperature}, our estimates of $B$ were unstable due to the non-differentiability of our underlying critical temperature estimator; 
we therefore set $B=0$ in the above equations to obtain our final expressions for the mean critical temperature point estimate and variance.

\section{On vector order observables}

\subsection{A 2D order observable}
\label{appendix:2D}

Before stating Prop.~\ref{prop:never_broken_2D}, we need the following lemma, which characterizes the real orthogonal square and fourth roots of the $2\times 2$ identity matrix.

\begin{lemma} \label{lemma:A4}
Let $A$ be a $2\times 2$ real orthogonal matrix such that $A^4 = I$, 
where $I$ denotes the $2\times 2$ identity matrix. 
\begin{enumerate}[label={(\alph*)}]
\item \label{lemma:A4:a}
If $A^2 = I$, then
\[ A = \pm I \mbox{ or } A = \lmat \cos\theta & \sin\theta\\ \sin\theta & -\cos\theta\rmat, \quad\theta\in [0, 2\pi). \]
\item \label{lemma:A4:b}
If $A^2\neq I$, then
\[ A = \pm \lmat 0 & -1\\ 1 & 0\rmat. \]
\item \label{lemma:A4:c} 
$A^2 = \pm I$.
\end{enumerate}
\end{lemma}
\begin{proof}
Since $A$ is orthogonal, then it has one of the two following forms for some real $\theta$:
\begin{align*}
A &= \lmat\cos\theta & -\sin\theta\\ \cos\theta & \sin\theta\rmat, \mbox{ or } \tag{*} \\
A &= \lmat\cos\theta & \sin\theta\\ \sin\theta & -\cos\theta\rmat. \tag{**}
\end{align*}

To prove (a), suppose $A^2=I$. 
If $A$ has the form (**), then it is easy to check that $A^2=I$ holds trivially. 
Suppose instead $A$ has the form (*). 
Then $A^2=I$ necessitates
\begin{align*}
\cos^2\theta -\sin^2\theta &= 1 \\
\cos^2\theta - (1 - \cos^2\theta) &= 1 \\
\cos^2\theta &= 1,
\end{align*}
and hence $\sin^2\theta=0$. 
By (*), this implies $A = \pm I$ as claimed in (a).

To prove (b), suppose $A^2\neq I$. 
Since $A^4 = (A^2)^2 = I$, then by (a), we must have either $A^2=-I$ or $A^2$ has the following form for some real $\phi$:
\[ A^2 = \lmat\cos\phi & \sin\phi\\ \sin\phi & -\cos\phi\rmat. \tag{***} \]
First suppose $A^2$ has the form (***) and $A$ the form (*). 
Then equating the diagonal terms of (*) squared and (***), 
and doing similar with the off-diagonal terms, 
we obtain the equations
\begin{align*}
\cos^2\theta-\sin^2\theta &= \pm \cos\phi \\
2\cos\theta\sin\theta &= \pm\sin\phi.
\end{align*}
These equations imply
\begin{align*}
\cos^2\theta-\sin^2\theta &= \pm 0 \\
2\cos\theta\sin\theta &= \pm0,
\end{align*}
and hence $\cos\theta=\sin\theta=0$, which is impossible.

Suppose on the other hand $A^2$ has the form (***) and $A$ the form (**). 
Equating the diagonal terms of (***) squared and (**), we obtain
\begin{align*}
(\pm \cos\theta)^2 + \sin^2\theta &= \pm\cos\phi \\
1 &= \pm\cos\phi,
\end{align*}
which is also impossible. 
Thus, $A^2$ cannot have the form (***).

We now turn to the case $A^2=-I$. 
If $A^2=-I$ and $A$ has the form (**), then $\cos^2\theta +\sin^2\theta = -1$, which is impossible. 
On the other hand, if $A^2=-I$ and $A$ has the form (*), then
\begin{align*}
\cos^2\theta - \sin^2\theta &= -1 \\
\cos^2\theta - (1-\cos^2\theta) &= -1 \\
2\cos^2\theta - 1 &= -1 \\
\cos^2\theta &= 0,
\end{align*}
and hence $\sin^2\theta=1$. 
The form (*) thus implies the expression for $A$ claimed in (b).

Finally, to prove (c), simply observe if $A^2=I$, then we are done; 
otherwise, $A$ has the form given in (b), whose square is $-I$.
\end{proof}

We now state and prove Prop.~\ref{prop:never_broken_2D}, which gives the subgroup of never-broken symmetries for a 2D order observable. 
Let $\beta = \rho\alpha\rho^{-1}$; i.e., the generator of horizontal translations. 
Let $\mathrm{O}(2, \RR)$ be the group of $2\times 2$ real orthogonal matrices.

\begin{prop} \label{prop:never_broken_2D}
Let $\Psi$ be the set of all real 2D representations $\psi:G\mapsto\mathrm{O}(2, \RR)$ of the Ising symmetry group $G$. 
Then $\bigcap_{\psi\in\Psi} \operatorname{ker}(\psi) = \langle\alpha^2,\beta^2\rangle$.
\end{prop}
\begin{proof}
Let $I$ be the $2\times 2$ identity matrix, 
and let $\psi\in\Psi$. 
Since $\rho^4 = 1$ by definition, then $\psi_{\rho}^4 = I$. 
By Lemma~\ref{lemma:A4}~(c), $\psi_{\rho}^2 = \pm I$. 
The defining relation $\alpha\rho^2 = \rho^2\alpha^{-1}$ thus implies
\begin{align*}
\psi_{\alpha}\psi_{\rho}^2 &= \psi_{\rho}^2\psi_{\alpha}^{-1} \\
\psi_{\alpha} &= \psi_{\alpha}^{-1} \\
\psi_{\alpha}^2 &= I.
\end{align*}
Ergo, $\alpha^2\in\operatorname{ker}(\psi)$. 
We can similarly show $\beta^2\in\operatorname{ker}(\psi)$, 
and hence
\[ \langle\alpha^2, \beta^2\rangle \leq \bigcap_{\psi\in\Psi} \operatorname{ker}(\psi). \]
All that remains is to prove the reverse inclusion.

Let $\Phi$ be the set of all real scalar representations $\phi:G\mapsto\{-1, 1\}$, 
and define the set
\[ \Psi_1 = \{g\mapsto \phi_g I: \phi\in\Phi\}. \]
Then by Prop.~\ref{prop:never_broken}, we have
\[ \bigcap_{\psi\in\Psi}\operatorname{ker}(\psi) \leq \bigcap_{\psi\in \Psi_1}\operatorname{ker}(\psi) \leq \langle\alpha^2, \rho^2, (\alpha\rho)^2\rangle. \]
If we can show
\[ \rho^2, \alpha\beta \notin \bigcap_{\psi\in\Psi}\operatorname{ker}(\psi), \tag{*} \]
then this will imply
\[ \bigcap_{\psi\in\Psi}\operatorname{ker}(\psi) \leq \langle\alpha^2, \beta^2\rangle, \]
which will then establish the claim. 
Consider $\psi\in\Psi$ defined such that
\begin{align*}
\psi_{\alpha} &= \lmat -1 & 0\\ 0 & 1\rmat \\
\psi_{\rho} &= \lmat 0 & -1\\ 1 & 0\rmat.
\end{align*}
Note these are valid representations as $\psi_{\alpha}\psi_{\rho}^2 = \psi_{\rho}^2\psi_{\alpha}^{-1}$. 
Then we have $\psi_{\rho}^2 = -I$ as well as
\begin{align*}
\psi_{\alpha}\psi_{\beta} 
&= \psi_{\alpha}\psi_{\rho}\psi_{\alpha}\psi_{\rho}^{-1} \\
&= -I,
\end{align*}
so that $\rho^2, \alpha\beta\notin \operatorname{ker}(\psi)$. 
This establishes (*) and hence the proposition.
\end{proof}

Recall the description of the GE-autoencoder architecture in Sec.~\ref{sec:2D}, 
where the input into the network is a 4D block-average $\check{x}$ of a lattice configuration $\mathbf{x}$. 
The following proposition states that our particular choice of architecture places an additional constraint on the representation $\psi$.

\begin{prop} \label{prop:art_2D}
The GE-encoder described in Sec.~\ref{sec:2D} satisfies $\psi_{\alpha}\check{\calO} = \psi_{\rho^2\tau}\check{\calO}$.
\end{prop}
\begin{proof}
Let $\check{x} = (\check{x}_1,\check{x}_2,\check{x}_3,\check{x}_4)$ be the block-average of a lattice configuration $\mathbf{x}$ over non-overlapping $2\times 2$ blocks. 
Then
\begin{align*}
\alpha\check{x} &= (\check{x}_3,\check{x}_4,\check{x}_1,\check{x}_2) \\
\rho\check{x} &= (\check{x}_2,\check{x}_4,\check{x}_1,\check{x}_3) \\
\tau\check{x} &= (\check{x}_2,\check{x}_1,\check{x}_4,\check{x}_3).
\end{align*}
Based on these permutations, it is easy to verify $\alpha\check{x} = \rho^2\tau\check{x}$. 
Thus,
\begin{align*}
\check{\calO}(\alpha\check{x}) &= \check{\calO}(\rho^2\tau\check{x}) \\
\psi_{\alpha}\check{\calO} &= \psi_{\rho^2\tau}\check{\calO},
\end{align*}
completing the proof.
\end{proof}

\subsection{Never-broken symmetries in terms of characters}
\label{appendix:chi}

The following proposition gives a way to compute the subgroup of never-broken symmetries of an arbitrary finite group $G$ in terms of its character table.

\begin{prop} \label{prop:chi}
Let $G$ be a finite group. 
Let $\Psi_d$ be the set of all real orthogonal representations of $G$ with degree $d$. 
Let $\hat{G}$ be the set of all irreducible characters of $G$. 
Then 
\[ \bigcap_{\psi\in\Psi_d}\operatorname{ker}(\psi) = \bigcap_{\chi\in\hat{G}\mid \degR(\chi)\leq d} \operatorname{ker}(\chi). \]
\end{prop}
\begin{proof}
Let $\chi\in\hat{G}$ such that $\degR(\chi) \leq d$. 
Let $\psi$ be the irreducible representation with character $\chi$. 
Let $\psi_{\RR}$ be the smallest real representation built out of $\psi$:
\[ \psi_{\RR} = 
\begin{cases}
\psi, & \mbox{ if } \IFS(\chi) = 1 \\
\psi\oplus\overline{\psi}_{\chi}, & \mbox{ if } \IFS(\chi) = 0 \\
\psi\oplus\psi, & \mbox{ if } \IFS(\chi) = -1.
\end{cases} \]
Note $\degC(\psi_{\RR}) = \degR(\chi)\leq d$. 
Let $n = d-\degC(\psi_{\RR})\geq 0$ and $\psi_1$ the trivial representation. 
Then construct the representation
\[ \psi_{\chi} = \psi_{\RR}\oplus\bigoplus_{i=1}^n \psi_1. \]
The kernel of this representation is clearly $\operatorname{ker}(\psi_{\chi}) = \operatorname{ker}(\psi_{\RR}) = \operatorname{ker}(\chi)$. 
Moreover, $\degC(\psi_{\chi}) = d$ so that $\psi_{\chi}\in\Psi_d$. 
Since a real representation $\psi_{\chi}\in\Psi_d$ can be constructed for every $\chi\in\hat{G}\mid\degR(\chi)\leq d$, then
\[ \bigcap_{\psi\in\Psi_d}\operatorname{ker}(\psi) \leq \bigcap_{\chi\in\hat{G}\mid\degR(\chi)\leq d}\operatorname{ker}(\chi). \tag{*} \]
All that remains is to prove the reverse inclusion as well.

Let $\psi\in\Psi_d$. 
Then $\psi$ admits the decomposition
\[ \psi = \psi_1\oplus \cdots \oplus \psi_k, \]
where each $\psi_i$ is a real-irreducible representation (i.e., irreducible over $\RR$). 
For each $i\in\{1,\ldots,k\}$, there exists a (complex-)irreducible representation $\psi^{\prime}_i$ such that $\psi_i = \psi^{\prime}_i$, $\psi_i = \psi^{\prime}_i\oplus\overline{\psi^{\prime}}_i$, or $\psi_i = \psi^{\prime}_i\oplus\psi^{\prime}_i$. 
Since $\operatorname{ker}(\psi^{\prime}_i) = \operatorname{ker}(\overline{\psi^{\prime}}_i)$, thne in any of these three cases, we have $\operatorname{ker}(\psi_i) = \operatorname{ker}(\psi^{\prime}_i)$. 
Thus,
\[ \operatorname{ker}(\psi) = \bigcap_{i=1}^k\operatorname{ker}(\psi_i) = \bigcap_{i=1}^k\operatorname{ker}(\psi^{\prime}_i). \]
Now let $\chi_i$ and $\chi^{\prime}_i$ be the characters of $\psi_i$ and $\psi^{\prime}_i$ respectively. 
Clearly, $\degR(\chi^{\prime}_i) = \degC(\chi_i)\leq d$ since $\degC(\psi_i)\leq\degC(\psi)=d$ for each $i\in\{1,\ldots,k\}$. 
Thus,
\[ \operatorname{ker}(\psi) = \bigcap_{i=1}^k\operatorname{ker}(\chi^{\prime}_i), \]
where each $\chi^{\prime}_i$ is irreducible and $\degR(\chi^{\prime}_i)\leq d$. 
This establishes the reverse inclusion of (*).
\end{proof}

\subsection{Symmetry regularization}
\label{appendix:reg_2D}

Here we relate the general symmetry regularization terms $R_g$ and $S_g$ [Eqs.~\eqref{eq:R}-\eqref{eq:S}] to the first and second regularization terms in Eq.~\eqref{eq:loss_reduced} for 1D order observables. 
This will provide insight into how Eq.~\eqref{eq:loss_2D} generalizes Eq.~\eqref{eq:loss_reduced}. 
We use the notation as in Eqs.~\eqref{eq:psihat}-\eqref{eq:S}.

First, however, we simplify the expression for $S_g$ [Eq.~\eqref{eq:S}] for arbitrary order dimension $d$.

\begin{lemma} \label{lemma:S_simple}
The regularization term $S_g$ [Eq.~\eqref{eq:S}] admits the expression
\[ S_g = 1 - \frac{\Vert ZZ^+ Z_g\Vert_F^2}{\Vert Z_g\Vert_F^2}. \]
\end{lemma}
\begin{proof}
Recalling the expression [Eq.~\eqref{eq:psihat}] for the linear least squares estimator $\hat{\psi}_g$, 
the squared residual is
\begin{align*}
\Vert Z_g - Z\hat{\psi}_g\Vert_F^2 
&= \Vert Z_g - Z(Z^+ Z_g + P A)\Vert_F^2 \\
&= \Vert Z_g - ZZ^+ Z_g\Vert_F^2 \\
&= \Vert (I - ZZ^+) Z_g\Vert_F^2.
\end{align*}
Since $I-ZZ^+$ is an orthogonal projection operator, then by the Pythagorean Theorem we obtain
\[ \Vert Z_g-Z\hat{\psi}_g\Vert_F^2 = \Vert Z_g\Vert_F^2 - \Vert ZZ^+ Z_g\Vert_F^2, \]
and thus
\[ S_g = 1 - \frac{\Vert ZZ^+ Z_g\Vert_F^2}{\Vert Z_g\Vert_F^2}, \]
completing the proof.
\end{proof}

We now derive the simplified expressions for the regularization terms $R_g$ and $S_g$ for 1D order observables.

\begin{prop} \label{prop:RS}
Let $d=1$. 
Then the regularization terms $R_g$ and $S_g$ [Eq.~\eqref{eq:R}-\eqref{eq:S}] admit the expressions
\begin{align*}
R_g &= \left(1 - \frac{\Vert Z_g\Vert^2}{\Vert Z\Vert^2} (1-S_g)\right)^2 \\
S_g &= 1 - \left(\frac{Z^{\top}Z_g}{\Vert Z\Vert \Vert Z_g\Vert}\right)^2.
\end{align*}
\end{prop}
\begin{proof}
We first consider $S_g$. 
By Lemma~\ref{lemma:S_simple}, 
\[ S_g = 1 - \frac{\Vert ZZ^+ Z_g\Vert^2}{\Vert Z_g\Vert^2}, \]
where we replaced the Frobenius matrix norm with the usual vector norm since $Z$ and $Z_g$ are now $N$-dimensional column vectors. 
Since $ZZ^+$ is an orthogonal projection operator, then $(ZZ^+)^{\top}(ZZ^+) = ZZ^+$ so that
\[ S_g = 1 - \frac{Z_g^{\top} ZZ^+ Z_g}{\Vert Z_g\Vert^2}. \]
Now unless the GE-encoder is exactly zero on every sampled lattice configuration, $Z$ is a nonzero column vector and thus full-rank. 
Its pseudoinverse is thus the row vector
\[ Z^+ = \frac{Z^{\top}}{\Vert Z\Vert^2}. \]
Substituting this into our expression for $S_g$, we obtain
\begin{align*}
S_g 
&= 1 - \frac{Z_g^{\top} ZZ^{\top} Z_g}{\Vert Z\Vert^2 \Vert Z_g\Vert^2} \\
&= 1 - \left(\frac{Z^{\top}Z_g}{\Vert Z\Vert \Vert Z_g\Vert}\right)^2,
\end{align*}
as claimed.

We now move to $R_g$. 
Since $Z$ is full-rank, then its null space is trivial so that $P=0$. 
The estimator $\hat{\psi}_g$ [Eq.~\eqref{eq:psihat}] thus simplifies to
\begin{align*}
\hat{\psi}_g 
&= Z^+Z_g + 0 \\
&= \frac{Z^{\top}Z_g}{\Vert Z\Vert^2} \\
&= \frac{\Vert Z_g\Vert}{\Vert Z\Vert}\cdot \frac{Z^{\top}Z_g}{\Vert Z\Vert\Vert Z_g\Vert}.
\end{align*}
Substituting this into Eq.~\eqref{eq:R} and noting $\hat{\psi}_g$ is a scalar, we have
\begin{align*}
R_g 
&= (1 - \hat{\psi}_g^2)^2 \\
&= \left(1 - \frac{\Vert Z_g\Vert^2}{\Vert Z\Vert^2}\left(\frac{Z^{\top}Z_g}{\Vert Z\Vert\Vert Z_g\Vert}\right)^2\right)^2 \\
&= \left(1 - \frac{\Vert Z_g\Vert^2}{\Vert Z\Vert^2} (1-S_g)\right)^2,
\end{align*}
completing the proof.
\end{proof}

Observe that the expressions for $R_g$ and $S_g$ in Prop.~\ref{prop:RS} match the first and second regularization terms in Eq.~\eqref{eq:loss_reduced}, 
except for the factor $(1-S_g)$ in $R_g$. 
Even still, in the optimal case $S_g=0$, the expression for $R_g$ matches the first regularization term in Eq.~\eqref{eq:loss_reduced}.

\bibliography{references} 

\begin{thebibliography}{60}%
\makeatletter
\providecommand \@ifxundefined [1]{%
 \@ifx{#1\undefined}
}%
\providecommand \@ifnum [1]{%
 \ifnum #1\expandafter \@firstoftwo
 \else \expandafter \@secondoftwo
 \fi
}%
\providecommand \@ifx [1]{%
 \ifx #1\expandafter \@firstoftwo
 \else \expandafter \@secondoftwo
 \fi
}%
\providecommand \natexlab [1]{#1}%
\providecommand \enquote  [1]{``#1''}%
\providecommand \bibnamefont  [1]{#1}%
\providecommand \bibfnamefont [1]{#1}%
\providecommand \citenamefont [1]{#1}%
\providecommand \href@noop [0]{\@secondoftwo}%
\providecommand \href [0]{\begingroup \@sanitize@url \@href}%
\providecommand \@href[1]{\@@startlink{#1}\@@href}%
\providecommand \@@href[1]{\endgroup#1\@@endlink}%
\providecommand \@sanitize@url [0]{\catcode `\\12\catcode `\$12\catcode
  `\&12\catcode `\#12\catcode `\^12\catcode `\_12\catcode `\%12\relax}%
\providecommand \@@startlink[1]{}%
\providecommand \@@endlink[0]{}%
\providecommand \url  [0]{\begingroup\@sanitize@url \@url }%
\providecommand \@url [1]{\endgroup\@href {#1}{\urlprefix }}%
\providecommand \urlprefix  [0]{URL }%
\providecommand \Eprint [0]{\href }%
\providecommand \doibase [0]{https://doi.org/}%
\providecommand \selectlanguage [0]{\@gobble}%
\providecommand \bibinfo  [0]{\@secondoftwo}%
\providecommand \bibfield  [0]{\@secondoftwo}%
\providecommand \translation [1]{[#1]}%
\providecommand \BibitemOpen [0]{}%
\providecommand \bibitemStop [0]{}%
\providecommand \bibitemNoStop [0]{.\EOS\space}%
\providecommand \EOS [0]{\spacefactor3000\relax}%
\providecommand \BibitemShut  [1]{\csname bibitem#1\endcsname}%
\let\auto@bib@innerbib\@empty
\bibitem [{\citenamefont {Ashcroft}\ and\ \citenamefont
  {Mermin}(1976)}]{ashcroft1976solid}%
  \BibitemOpen
  \bibfield  {author} {\bibinfo {author} {\bibfnamefont {N.~W.}\ \bibnamefont
  {Ashcroft}}\ and\ \bibinfo {author} {\bibfnamefont {N.~D.}\ \bibnamefont
  {Mermin}},\ }\href@noop {} {\emph {\bibinfo {title} {{S}olid {S}tate
  {P}hysics}}}\ (\bibinfo  {publisher} {Holt-Saunders},\ \bibinfo {year}
  {1976})\BibitemShut {NoStop}%
\bibitem [{\citenamefont {Friedli}\ and\ \citenamefont
  {Velenik}(2017)}]{friedli2017statistical}%
  \BibitemOpen
  \bibfield  {author} {\bibinfo {author} {\bibfnamefont {S.}~\bibnamefont
  {Friedli}}\ and\ \bibinfo {author} {\bibfnamefont {Y.}~\bibnamefont
  {Velenik}},\ }\href@noop {} {\emph {\bibinfo {title} {Statistical mechanics
  of lattice systems: a concrete mathematical introduction}}}\ (\bibinfo
  {publisher} {Cambridge University Pres},\ \bibinfo {year} {2017})\BibitemShut
  {NoStop}%
\bibitem [{\citenamefont {Gomez}\ \emph {et~al.}(2019)\citenamefont {Gomez},
  \citenamefont {Bures},\ and\ \citenamefont {Moure}}]{Gomez2019Review}%
  \BibitemOpen
  \bibfield  {author} {\bibinfo {author} {\bibfnamefont {H.}~\bibnamefont
  {Gomez}}, \bibinfo {author} {\bibfnamefont {M.}~\bibnamefont {Bures}},\ and\
  \bibinfo {author} {\bibfnamefont {A.}~\bibnamefont {Moure}},\ }\bibfield
  {title} {\bibinfo {title} {{A} review on computational modelling of
  phase-transition problems},\ }\href {https://doi.org/10.1098/rsta.2018.0203}
  {\bibfield  {journal} {\bibinfo  {journal} {Philos. Trans. Roy. Soc. A}\
  }\textbf {\bibinfo {volume} {377}},\ \bibinfo {pages} {20180203} (\bibinfo
  {year} {2019})}\BibitemShut {NoStop}%
\bibitem [{\citenamefont {Nagy}\ \emph {et~al.}(2013)\citenamefont {Nagy},
  \citenamefont {Calixto},\ and\ \citenamefont {Romera}}]{DFT1}%
  \BibitemOpen
  \bibfield  {author} {\bibinfo {author} {\bibfnamefont {{\'{A}}.}~\bibnamefont
  {Nagy}}, \bibinfo {author} {\bibfnamefont {M.}~\bibnamefont {Calixto}},\ and\
  \bibinfo {author} {\bibfnamefont {E.}~\bibnamefont {Romera}},\ }\bibfield
  {title} {\bibinfo {title} {{A} density {F}unctional {T}heory {V}iew of
  {Q}uantum {P}hase {T}ransitions},\ }\href {https://doi.org/10.1021/ct301015n}
  {\bibfield  {journal} {\bibinfo  {journal} {J. Chem. Theory Comput.}\
  }\textbf {\bibinfo {volume} {9}},\ \bibinfo {pages} {1068} (\bibinfo {year}
  {2013})}\BibitemShut {NoStop}%
\bibitem [{\citenamefont {Wu}\ \emph {et~al.}(2006)\citenamefont {Wu},
  \citenamefont {Sarandy}, \citenamefont {Lidar},\ and\ \citenamefont
  {Sham}}]{DFT2}%
  \BibitemOpen
  \bibfield  {author} {\bibinfo {author} {\bibfnamefont {L.-A.}\ \bibnamefont
  {Wu}}, \bibinfo {author} {\bibfnamefont {M.~S.}\ \bibnamefont {Sarandy}},
  \bibinfo {author} {\bibfnamefont {D.~A.}\ \bibnamefont {Lidar}},\ and\
  \bibinfo {author} {\bibfnamefont {L.~J.}\ \bibnamefont {Sham}},\ }\bibfield
  {title} {\bibinfo {title} {Linking entanglement and quantum phase transitions
  via density-functional theory},\ }\href
  {https://doi.org/10.1103/PhysRevA.74.052335} {\bibfield  {journal} {\bibinfo
  {journal} {Phys. Rev. A}\ }\textbf {\bibinfo {volume} {74}},\ \bibinfo
  {pages} {052335} (\bibinfo {year} {2006})}\BibitemShut {NoStop}%
\bibitem [{\citenamefont {Nagy}\ and\ \citenamefont {Romera}(2013)}]{DFT3}%
  \BibitemOpen
  \bibfield  {author} {\bibinfo {author} {\bibfnamefont {A.}~\bibnamefont
  {Nagy}}\ and\ \bibinfo {author} {\bibfnamefont {E.}~\bibnamefont {Romera}},\
  }\bibfield  {title} {\bibinfo {title} {Quantum phase transitions via
  density-functional theory: Extension to the degenerate case},\ }\href
  {https://doi.org/10.1103/PhysRevA.88.042515} {\bibfield  {journal} {\bibinfo
  {journal} {Phys. Rev. A}\ }\textbf {\bibinfo {volume} {88}},\ \bibinfo
  {pages} {042515} (\bibinfo {year} {2013})}\BibitemShut {NoStop}%
\bibitem [{\citenamefont {Shahi}\ \emph {et~al.}(2018)\citenamefont {Shahi},
  \citenamefont {Sun},\ and\ \citenamefont {Perdew}}]{Shahi:2018dft}%
  \BibitemOpen
  \bibfield  {author} {\bibinfo {author} {\bibfnamefont {C.}~\bibnamefont
  {Shahi}}, \bibinfo {author} {\bibfnamefont {J.}~\bibnamefont {Sun}},\ and\
  \bibinfo {author} {\bibfnamefont {J.~P.}\ \bibnamefont {Perdew}},\ }\bibfield
   {title} {\bibinfo {title} {{Accurate critical pressures for structural phase
  transitions of group IV, III-V, and II-VI compounds from the SCAN density
  functional}},\ }\href {https://doi.org/10.1103/PhysRevB.97.094111} {\bibfield
   {journal} {\bibinfo  {journal} {Phys. Rev. B}\ }\textbf {\bibinfo {volume}
  {97}},\ \bibinfo {pages} {094111} (\bibinfo {year} {2018})}\BibitemShut
  {NoStop}%
\bibitem [{\citenamefont {Maurer}\ \emph {et~al.}(2019)\citenamefont {Maurer},
  \citenamefont {Freysoldt}, \citenamefont {Reilly}, \citenamefont
  {Brandenburg}, \citenamefont {Hofmann}, \citenamefont {Björkman},
  \citenamefont {Leb{\`{e}}gue},\ and\ \citenamefont
  {Tkatchenko}}]{Maurer:2019dft}%
  \BibitemOpen
  \bibfield  {author} {\bibinfo {author} {\bibfnamefont {R.~J.}\ \bibnamefont
  {Maurer}}, \bibinfo {author} {\bibfnamefont {C.}~\bibnamefont {Freysoldt}},
  \bibinfo {author} {\bibfnamefont {A.~M.}\ \bibnamefont {Reilly}}, \bibinfo
  {author} {\bibfnamefont {J.~G.}\ \bibnamefont {Brandenburg}}, \bibinfo
  {author} {\bibfnamefont {O.~T.}\ \bibnamefont {Hofmann}}, \bibinfo {author}
  {\bibfnamefont {T.}~\bibnamefont {Björkman}}, \bibinfo {author}
  {\bibfnamefont {S.}~\bibnamefont {Leb{\`{e}}gue}},\ and\ \bibinfo {author}
  {\bibfnamefont {A.}~\bibnamefont {Tkatchenko}},\ }\bibfield  {title}
  {\bibinfo {title} {{A}dvances in {D}ensity-{F}unctional {C}alculations for
  {M}aterials {M}odeling},\ }\href
  {https://doi.org/10.1146/annurev-matsci-070218-010143} {\bibfield  {journal}
  {\bibinfo  {journal} {Ann. Rev. Mat. Res.}\ }\textbf {\bibinfo {volume}
  {49}},\ \bibinfo {pages} {1} (\bibinfo {year} {2019})}\BibitemShut {NoStop}%
\bibitem [{\citenamefont {Haile}(1992)}]{haile1992molecular}%
  \BibitemOpen
  \bibfield  {author} {\bibinfo {author} {\bibfnamefont {J.~M.}\ \bibnamefont
  {Haile}},\ }\href@noop {} {\emph {\bibinfo {title} {Molecular dynamics
  simulation: elementary methods}}}\ (\bibinfo  {publisher} {John Wiley \&
  Sons, Inc.},\ \bibinfo {year} {1992})\BibitemShut {NoStop}%
\bibitem [{\citenamefont {Sasaki}\ \emph {et~al.}(2020)\citenamefont {Sasaki},
  \citenamefont {Hayashi},\ and\ \citenamefont {Kawauchi}}]{Sasaki:2020md}%
  \BibitemOpen
  \bibfield  {author} {\bibinfo {author} {\bibfnamefont {R.}~\bibnamefont
  {Sasaki}}, \bibinfo {author} {\bibfnamefont {Y.}~\bibnamefont {Hayashi}},\
  and\ \bibinfo {author} {\bibfnamefont {S.}~\bibnamefont {Kawauchi}},\
  }\bibfield  {title} {\bibinfo {title} {Acceleration of liquid-crystalline
  phase transition simulations using selectively scaled and returned molecular
  dynamics},\ }\bibfield  {booktitle} {\emph {\bibinfo {booktitle} {Journal of
  Chemical Information and Modeling}},\ }\href
  {https://doi.org/10.1021/acs.jcim.0c00239} {\bibfield  {journal} {\bibinfo
  {journal} {J. Chem. Inf. Model.}\ }\textbf {\bibinfo {volume} {60}},\
  \bibinfo {pages} {3499} (\bibinfo {year} {2020})}\BibitemShut {NoStop}%
\bibitem [{\citenamefont {Shanavas}\ and\ \citenamefont
  {Sharma}(2009)}]{Shanavas2009review}%
  \BibitemOpen
  \bibfield  {author} {\bibinfo {author} {\bibfnamefont {K.~V.}\ \bibnamefont
  {Shanavas}}\ and\ \bibinfo {author} {\bibfnamefont {S.~M.}\ \bibnamefont
  {Sharma}},\ }\bibfield  {title} {\bibinfo {title} {Molecular dynamics
  simulations of phase transitions in argon-filled single-walled carbon
  nanotube bundles under high pressure},\ }\href
  {https://doi.org/10.1103/PhysRevB.79.155425} {\bibfield  {journal} {\bibinfo
  {journal} {Phys. Rev. B}\ }\textbf {\bibinfo {volume} {79}},\ \bibinfo
  {pages} {155425} (\bibinfo {year} {2009})}\BibitemShut {NoStop}%
\bibitem [{\citenamefont {Imada}\ \emph {et~al.}(1998)\citenamefont {Imada},
  \citenamefont {Fujimori},\ and\ \citenamefont {Tokura}}]{Imada1998mit}%
  \BibitemOpen
  \bibfield  {author} {\bibinfo {author} {\bibfnamefont {M.}~\bibnamefont
  {Imada}}, \bibinfo {author} {\bibfnamefont {A.}~\bibnamefont {Fujimori}},\
  and\ \bibinfo {author} {\bibfnamefont {Y.}~\bibnamefont {Tokura}},\
  }\bibfield  {title} {\bibinfo {title} {Metal-insulator transitions},\ }\href
  {https://doi.org/10.1103/RevModPhys.70.1039} {\bibfield  {journal} {\bibinfo
  {journal} {Rev. Mod. Phys.}\ }\textbf {\bibinfo {volume} {70}},\ \bibinfo
  {pages} {1039} (\bibinfo {year} {1998})}\BibitemShut {NoStop}%
\bibitem [{\citenamefont {Johnston}(2010)}]{Johnston2010Review}%
  \BibitemOpen
  \bibfield  {author} {\bibinfo {author} {\bibfnamefont {D.~C.}\ \bibnamefont
  {Johnston}},\ }\bibfield  {title} {\bibinfo {title} {The puzzle of high
  temperature superconductivity in layered iron pnictides and chalcogenides},\
  }\href {https://doi.org/10.1080/00018732.2010.513480} {\bibfield  {journal}
  {\bibinfo  {journal} {Adv. in Phys.}\ }\textbf {\bibinfo {volume} {59}},\
  \bibinfo {pages} {803} (\bibinfo {year} {2010})}\BibitemShut {NoStop}%
\bibitem [{\citenamefont {Stewart}(2017)}]{Stewart2017Review}%
  \BibitemOpen
  \bibfield  {author} {\bibinfo {author} {\bibfnamefont {G.~R.}\ \bibnamefont
  {Stewart}},\ }\bibfield  {title} {\bibinfo {title} {Unconventional
  superconductivity},\ }\href {https://doi.org/10.1080/00018732.2017.1331615}
  {\bibfield  {journal} {\bibinfo  {journal} {Advances in Physics}\ }\textbf
  {\bibinfo {volume} {66}},\ \bibinfo {pages} {75} (\bibinfo {year}
  {2017})}\BibitemShut {NoStop}%
\bibitem [{\citenamefont {Keimer}\ \emph {et~al.}(2015)\citenamefont {Keimer},
  \citenamefont {Kivelson}, \citenamefont {Norman}, \citenamefont {Uchida},\
  and\ \citenamefont {Zaanen}}]{Keimer2015Review}%
  \BibitemOpen
  \bibfield  {author} {\bibinfo {author} {\bibfnamefont {B.}~\bibnamefont
  {Keimer}}, \bibinfo {author} {\bibfnamefont {S.~A.}\ \bibnamefont
  {Kivelson}}, \bibinfo {author} {\bibfnamefont {M.~R.}\ \bibnamefont
  {Norman}}, \bibinfo {author} {\bibfnamefont {S.}~\bibnamefont {Uchida}},\
  and\ \bibinfo {author} {\bibfnamefont {J.}~\bibnamefont {Zaanen}},\
  }\bibfield  {title} {\bibinfo {title} {From quantum matter to
  high-temperature superconductivity in copper oxides},\ }\href
  {https://doi.org/10.1038/nature14165} {\bibfield  {journal} {\bibinfo
  {journal} {Nature}\ }\textbf {\bibinfo {volume} {518}},\ \bibinfo {pages}
  {179} (\bibinfo {year} {2015})}\BibitemShut {NoStop}%
\bibitem [{\citenamefont {Zhou}\ \emph {et~al.}(2017)\citenamefont {Zhou},
  \citenamefont {Kanoda},\ and\ \citenamefont {Ng}}]{Zhou2017Review}%
  \BibitemOpen
  \bibfield  {author} {\bibinfo {author} {\bibfnamefont {Y.}~\bibnamefont
  {Zhou}}, \bibinfo {author} {\bibfnamefont {K.}~\bibnamefont {Kanoda}},\ and\
  \bibinfo {author} {\bibfnamefont {T.-K.}\ \bibnamefont {Ng}},\ }\bibfield
  {title} {\bibinfo {title} {Quantum spin liquid states},\ }\href
  {https://doi.org/10.1103/RevModPhys.89.025003} {\bibfield  {journal}
  {\bibinfo  {journal} {Rev. Mod. Phys.}\ }\textbf {\bibinfo {volume} {89}},\
  \bibinfo {pages} {025003} (\bibinfo {year} {2017})}\BibitemShut {NoStop}%
\bibitem [{\citenamefont {Savary}\ and\ \citenamefont
  {Balents}(2016)}]{Savary2016spinliquid}%
  \BibitemOpen
  \bibfield  {author} {\bibinfo {author} {\bibfnamefont {L.}~\bibnamefont
  {Savary}}\ and\ \bibinfo {author} {\bibfnamefont {L.}~\bibnamefont
  {Balents}},\ }\bibfield  {title} {\bibinfo {title} {{Q}uantum spin liquids: a
  review},\ }\href {https://doi.org/10.1088/0034-4885/80/1/016502} {\bibfield
  {journal} {\bibinfo  {journal} {Rep. Prog. Phys.}\ }\textbf {\bibinfo
  {volume} {80}},\ \bibinfo {pages} {016502} (\bibinfo {year}
  {2016})}\BibitemShut {NoStop}%
\bibitem [{\citenamefont {Fradkin}\ \emph {et~al.}(2015)\citenamefont
  {Fradkin}, \citenamefont {Kivelson},\ and\ \citenamefont
  {Tranquada}}]{Fradkin2015Review}%
  \BibitemOpen
  \bibfield  {author} {\bibinfo {author} {\bibfnamefont {E.}~\bibnamefont
  {Fradkin}}, \bibinfo {author} {\bibfnamefont {S.~A.}\ \bibnamefont
  {Kivelson}},\ and\ \bibinfo {author} {\bibfnamefont {J.~M.}\ \bibnamefont
  {Tranquada}},\ }\bibfield  {title} {\bibinfo {title} {Colloquium: Theory of
  intertwined orders in high temperature superconductors},\ }\href
  {https://doi.org/10.1103/RevModPhys.87.457} {\bibfield  {journal} {\bibinfo
  {journal} {Rev. Mod. Phys.}\ }\textbf {\bibinfo {volume} {87}},\ \bibinfo
  {pages} {457} (\bibinfo {year} {2015})}\BibitemShut {NoStop}%
\bibitem [{\citenamefont {Maier}\ \emph {et~al.}(2005)\citenamefont {Maier},
  \citenamefont {Jarrell}, \citenamefont {Schulthess}, \citenamefont {Kent},\
  and\ \citenamefont {White}}]{Maier2005superconductivity}%
  \BibitemOpen
  \bibfield  {author} {\bibinfo {author} {\bibfnamefont {T.~A.}\ \bibnamefont
  {Maier}}, \bibinfo {author} {\bibfnamefont {M.}~\bibnamefont {Jarrell}},
  \bibinfo {author} {\bibfnamefont {T.~C.}\ \bibnamefont {Schulthess}},
  \bibinfo {author} {\bibfnamefont {P.~R.~C.}\ \bibnamefont {Kent}},\ and\
  \bibinfo {author} {\bibfnamefont {J.~B.}\ \bibnamefont {White}},\ }\bibfield
  {title} {\bibinfo {title} {Systematic study of $d$-wave superconductivity in
  the {2D} repulsive {Hubbard} model},\ }\href
  {https://doi.org/10.1103/PhysRevLett.95.237001} {\bibfield  {journal}
  {\bibinfo  {journal} {Phys. Rev. Lett.}\ }\textbf {\bibinfo {volume} {95}},\
  \bibinfo {pages} {237001} (\bibinfo {year} {2005})}\BibitemShut {NoStop}%
\bibitem [{\citenamefont {Zheng}\ \emph {et~al.}(2017)\citenamefont {Zheng},
  \citenamefont {Chung}, \citenamefont {Corboz}, \citenamefont {Ehlers},
  \citenamefont {Qin}, \citenamefont {Noack}, \citenamefont {Shi},
  \citenamefont {White}, \citenamefont {Zhang},\ and\ \citenamefont
  {Chan}}]{Zheng2017absence}%
  \BibitemOpen
  \bibfield  {author} {\bibinfo {author} {\bibfnamefont {B.-X.}\ \bibnamefont
  {Zheng}}, \bibinfo {author} {\bibfnamefont {C.-M.}\ \bibnamefont {Chung}},
  \bibinfo {author} {\bibfnamefont {P.}~\bibnamefont {Corboz}}, \bibinfo
  {author} {\bibfnamefont {G.}~\bibnamefont {Ehlers}}, \bibinfo {author}
  {\bibfnamefont {M.-P.}\ \bibnamefont {Qin}}, \bibinfo {author} {\bibfnamefont
  {R.~M.}\ \bibnamefont {Noack}}, \bibinfo {author} {\bibfnamefont
  {H.}~\bibnamefont {Shi}}, \bibinfo {author} {\bibfnamefont {S.~R.}\
  \bibnamefont {White}}, \bibinfo {author} {\bibfnamefont {S.}~\bibnamefont
  {Zhang}},\ and\ \bibinfo {author} {\bibfnamefont {G.~K.-L.}\ \bibnamefont
  {Chan}},\ }\bibfield  {title} {\bibinfo {title} {Stripe order in the
  underdoped region of the two-dimensional {Hubbard} model},\ }\href
  {https://doi.org/10.1126/science.aam7127} {\bibfield  {journal} {\bibinfo
  {journal} {Science}\ }\textbf {\bibinfo {volume} {358}},\ \bibinfo {pages}
  {1155} (\bibinfo {year} {2017})}\BibitemShut {NoStop}%
\bibitem [{\citenamefont {Jiang}\ and\ \citenamefont
  {Devereaux}(2019)}]{Jiang2019stripes}%
  \BibitemOpen
  \bibfield  {author} {\bibinfo {author} {\bibfnamefont {H.-C.}\ \bibnamefont
  {Jiang}}\ and\ \bibinfo {author} {\bibfnamefont {T.~P.}\ \bibnamefont
  {Devereaux}},\ }\bibfield  {title} {\bibinfo {title} {Superconductivity in
  the doped {Hubbard} model and its interplay with next-nearest hopping
  $t^\prime$},\ }\href {https://doi.org/10.1126/science.aal5304} {\bibfield
  {journal} {\bibinfo  {journal} {Science}\ }\textbf {\bibinfo {volume}
  {365}},\ \bibinfo {pages} {1424} (\bibinfo {year} {2019})}\BibitemShut
  {NoStop}%
\bibitem [{\citenamefont {Haldane}(1983)}]{Haldane1983}%
  \BibitemOpen
  \bibfield  {author} {\bibinfo {author} {\bibfnamefont {F.}~\bibnamefont
  {Haldane}},\ }\bibfield  {title} {\bibinfo {title} {Continuum dynamics of the
  1-{D} {H}eisenberg antiferromagnet: {I}dentification with the {O(3)}
  nonlinear sigma model},\ }\href
  {https://doi.org/https://doi.org/10.1016/0375-9601(83)90631-X} {\bibfield
  {journal} {\bibinfo  {journal} {Phys. Lett. A}\ }\textbf {\bibinfo {volume}
  {93}},\ \bibinfo {pages} {464} (\bibinfo {year} {1983})}\BibitemShut
  {NoStop}%
\bibitem [{\citenamefont {Kennedy}\ and\ \citenamefont
  {Tasaki}(1992)}]{Kennedy1992}%
  \BibitemOpen
  \bibfield  {author} {\bibinfo {author} {\bibfnamefont {T.}~\bibnamefont
  {Kennedy}}\ and\ \bibinfo {author} {\bibfnamefont {H.}~\bibnamefont
  {Tasaki}},\ }\bibfield  {title} {\bibinfo {title} {Hidden symmetry breaking
  and the {Haldane} phase in {$S=1$} quantum spin chains},\ }\href
  {https://doi.org/10.1007/BF02097239} {\bibfield  {journal} {\bibinfo
  {journal} {Comm. Math. Phys.}\ }\textbf {\bibinfo {volume} {147}},\ \bibinfo
  {pages} {431} (\bibinfo {year} {1992})}\BibitemShut {NoStop}%
\bibitem [{\citenamefont {Anderson}(1963)}]{Anderson1963}%
  \BibitemOpen
  \bibfield  {author} {\bibinfo {author} {\bibfnamefont {P.~W.}\ \bibnamefont
  {Anderson}},\ }\bibfield  {title} {\bibinfo {title} {Plasmons, gauge
  invariance, and mass},\ }\href {https://doi.org/10.1103/PhysRev.130.439}
  {\bibfield  {journal} {\bibinfo  {journal} {Phys. Rev.}\ }\textbf {\bibinfo
  {volume} {130}},\ \bibinfo {pages} {439} (\bibinfo {year}
  {1963})}\BibitemShut {NoStop}%
\bibitem [{\citenamefont {Wen}(1990)}]{Wen:1990ws}%
  \BibitemOpen
  \bibfield  {author} {\bibinfo {author} {\bibfnamefont {X.~G.}\ \bibnamefont
  {Wen}},\ }\bibfield  {title} {\bibinfo {title} {{Topological} {Orders} {in}
  {Rigid} {Stages}},\ }\href {https://doi.org/10.1142/s0217979290000139}
  {\bibfield  {journal} {\bibinfo  {journal} {Int. J. Mod. Phys. B}\ }\textbf
  {\bibinfo {volume} {04}},\ \bibinfo {pages} {239} (\bibinfo {year}
  {1990})}\BibitemShut {NoStop}%
\bibitem [{\citenamefont {Kivelson}\ and\ \citenamefont
  {Lederer}(2019)}]{Kivelson2019}%
  \BibitemOpen
  \bibfield  {author} {\bibinfo {author} {\bibfnamefont {S.~A.}\ \bibnamefont
  {Kivelson}}\ and\ \bibinfo {author} {\bibfnamefont {S.}~\bibnamefont
  {Lederer}},\ }\bibfield  {title} {\bibinfo {title} {Linking the pseudogap in
  the cuprates with local symmetry breaking: A commentary},\ }\href
  {https://doi.org/10.1073/pnas.1908786116} {\bibfield  {journal} {\bibinfo
  {journal} {Proc. Nat. Acad. Sci.}\ }\textbf {\bibinfo {volume} {116}},\
  \bibinfo {pages} {14395} (\bibinfo {year} {2019})}\BibitemShut {NoStop}%
\bibitem [{\citenamefont {Goodfellow}\ \emph {et~al.}(2016)\citenamefont
  {Goodfellow}, \citenamefont {Bengio},\ and\ \citenamefont
  {Courville}}]{goodfellow2016deep}%
  \BibitemOpen
  \bibfield  {author} {\bibinfo {author} {\bibfnamefont {I.}~\bibnamefont
  {Goodfellow}}, \bibinfo {author} {\bibfnamefont {Y.}~\bibnamefont {Bengio}},\
  and\ \bibinfo {author} {\bibfnamefont {A.}~\bibnamefont {Courville}},\
  }\href@noop {} {\emph {\bibinfo {title} {Deep learning}}}\ (\bibinfo
  {publisher} {MIT pres},\ \bibinfo {year} {2016})\BibitemShut {NoStop}%
\bibitem [{\citenamefont {Carrasquilla}\ and\ \citenamefont
  {Melko}(2017)}]{carrasquilla2017machine}%
  \BibitemOpen
  \bibfield  {author} {\bibinfo {author} {\bibfnamefont {J.}~\bibnamefont
  {Carrasquilla}}\ and\ \bibinfo {author} {\bibfnamefont {R.~G.}\ \bibnamefont
  {Melko}},\ }\bibfield  {title} {\bibinfo {title} {Machine learning phases of
  matter},\ }\href {https://doi.org/10.1038/nphys4035} {\bibfield  {journal}
  {\bibinfo  {journal} {Nature Phys.}\ }\textbf {\bibinfo {volume} {13}},\
  \bibinfo {pages} {431} (\bibinfo {year} {2017})}\BibitemShut {NoStop}%
\bibitem [{\citenamefont {Broecker}\ \emph {et~al.}(2017)\citenamefont
  {Broecker}, \citenamefont {Carrasquilla}, \citenamefont {Melko},\ and\
  \citenamefont {Trebst}}]{broecker2017machine}%
  \BibitemOpen
  \bibfield  {author} {\bibinfo {author} {\bibfnamefont {P.}~\bibnamefont
  {Broecker}}, \bibinfo {author} {\bibfnamefont {J.}~\bibnamefont
  {Carrasquilla}}, \bibinfo {author} {\bibfnamefont {R.~G.}\ \bibnamefont
  {Melko}},\ and\ \bibinfo {author} {\bibfnamefont {S.}~\bibnamefont
  {Trebst}},\ }\bibfield  {title} {\bibinfo {title} {Machine learning quantum
  phases of matter beyond the fermion sign problem},\ }\href
  {https://www.nature.com/articles/s41598-017-09098-0} {\bibfield  {journal}
  {\bibinfo  {journal} {Scientific Rep.}\ }\textbf {\bibinfo {volume} {7}},\
  \bibinfo {pages} {1} (\bibinfo {year} {2017})}\BibitemShut {NoStop}%
\bibitem [{\citenamefont {Ch'ng}\ \emph {et~al.}(2017)\citenamefont {Ch'ng},
  \citenamefont {Carrasquilla}, \citenamefont {Melko},\ and\ \citenamefont
  {Khatami}}]{ch2017machine}%
  \BibitemOpen
  \bibfield  {author} {\bibinfo {author} {\bibfnamefont {K.}~\bibnamefont
  {Ch'ng}}, \bibinfo {author} {\bibfnamefont {J.}~\bibnamefont {Carrasquilla}},
  \bibinfo {author} {\bibfnamefont {R.~G.}\ \bibnamefont {Melko}},\ and\
  \bibinfo {author} {\bibfnamefont {E.}~\bibnamefont {Khatami}},\ }\bibfield
  {title} {\bibinfo {title} {Machine learning phases of strongly correlated
  fermions},\ }\href {https://doi.org/10.1103/PhysRevX.7.031038} {\bibfield
  {journal} {\bibinfo  {journal} {Phys. Rev. X}\ }\textbf {\bibinfo {volume}
  {7}},\ \bibinfo {pages} {031038} (\bibinfo {year} {2017})}\BibitemShut
  {NoStop}%
\bibitem [{\citenamefont {Wetzel}\ and\ \citenamefont
  {Scherzer}(2017)}]{wetzel2017machine}%
  \BibitemOpen
  \bibfield  {author} {\bibinfo {author} {\bibfnamefont {S.~J.}\ \bibnamefont
  {Wetzel}}\ and\ \bibinfo {author} {\bibfnamefont {M.}~\bibnamefont
  {Scherzer}},\ }\bibfield  {title} {\bibinfo {title} {Machine learning of
  explicit order parameters: From the {Ising} model to {SU(2)} lattice gauge
  theory},\ }\href {https://doi.org/10.1103/PhysRevB.96.184410} {\bibfield
  {journal} {\bibinfo  {journal} {Phys. Rev. B}\ }\textbf {\bibinfo {volume}
  {96}},\ \bibinfo {pages} {184410} (\bibinfo {year} {2017})}\BibitemShut
  {NoStop}%
\bibitem [{\citenamefont {Morningstar}\ and\ \citenamefont
  {Melko}(2018)}]{Morningstar:2018dl}%
  \BibitemOpen
  \bibfield  {author} {\bibinfo {author} {\bibfnamefont {A.}~\bibnamefont
  {Morningstar}}\ and\ \bibinfo {author} {\bibfnamefont {R.~G.}\ \bibnamefont
  {Melko}},\ }\bibfield  {title} {\bibinfo {title} {Deep learning the ising
  model near criticality},\ }\href {http://jmlr.org/papers/v18/17-527.html}
  {\bibfield  {journal} {\bibinfo  {journal} {J. Mach. Learn. Res.}\ }\textbf
  {\bibinfo {volume} {18}},\ \bibinfo {pages} {1} (\bibinfo {year}
  {2018})}\BibitemShut {NoStop}%
\bibitem [{\citenamefont {Efthymiou}\ \emph {et~al.}(2019)\citenamefont
  {Efthymiou}, \citenamefont {Beach},\ and\ \citenamefont
  {Melko}}]{Efthymiou:2019uy}%
  \BibitemOpen
  \bibfield  {author} {\bibinfo {author} {\bibfnamefont {S.}~\bibnamefont
  {Efthymiou}}, \bibinfo {author} {\bibfnamefont {M.~J.~S.}\ \bibnamefont
  {Beach}},\ and\ \bibinfo {author} {\bibfnamefont {R.~G.}\ \bibnamefont
  {Melko}},\ }\bibfield  {title} {\bibinfo {title} {{S}uper-resolving the
  {I}sing model with convolutional neural networks},\ }\href
  {https://doi.org/10.1103/physrevb.99.075113} {\bibfield  {journal} {\bibinfo
  {journal} {Phys. Rev. B}\ }\textbf {\bibinfo {volume} {99}},\ \bibinfo
  {pages} {075113} (\bibinfo {year} {2019})}\BibitemShut {NoStop}%
\bibitem [{\citenamefont {Walker}\ \emph {et~al.}(2020)\citenamefont {Walker},
  \citenamefont {Tam},\ and\ \citenamefont {Jarrell}}]{Walker:2020hi}%
  \BibitemOpen
  \bibfield  {author} {\bibinfo {author} {\bibfnamefont {N.}~\bibnamefont
  {Walker}}, \bibinfo {author} {\bibfnamefont {K.-M.}\ \bibnamefont {Tam}},\
  and\ \bibinfo {author} {\bibfnamefont {M.}~\bibnamefont {Jarrell}},\
  }\bibfield  {title} {\bibinfo {title} {{D}eep learning on the 2-dimensional
  {I}sing model to extract the crossover region with a variational
  autoencoder},\ }\href {https://doi.org/10.1038/s41598-020-69848-5} {\bibfield
   {journal} {\bibinfo  {journal} {Scientific Rep.}\ }\textbf {\bibinfo
  {volume} {10}},\ \bibinfo {pages} {1038} (\bibinfo {year}
  {2020})}\BibitemShut {NoStop}%
\bibitem [{\citenamefont {Johnston}\ \emph {et~al.}(2022)\citenamefont
  {Johnston}, \citenamefont {Khatami},\ and\ \citenamefont
  {Scalettar}}]{JohnstonReview}%
  \BibitemOpen
  \bibfield  {author} {\bibinfo {author} {\bibfnamefont {S.}~\bibnamefont
  {Johnston}}, \bibinfo {author} {\bibfnamefont {E.}~\bibnamefont {Khatami}},\
  and\ \bibinfo {author} {\bibfnamefont {R.}~\bibnamefont {Scalettar}},\
  }\bibfield  {title} {\bibinfo {title} {A perspective on machine learning and
  data science for strongly correlated electron problems},\ }\href
  {https://doi.org/https://doi.org/10.1016/j.cartre.2022.100231} {\bibfield
  {journal} {\bibinfo  {journal} {Carbon Trends}\ }\textbf {\bibinfo {volume}
  {9}},\ \bibinfo {pages} {100231} (\bibinfo {year} {2022})}\BibitemShut
  {NoStop}%
\bibitem [{\citenamefont {Krizhevsky}\ \emph {et~al.}(2012)\citenamefont
  {Krizhevsky}, \citenamefont {Sutskever},\ and\ \citenamefont
  {Hinton}}]{krizhevsky2012imagenet}%
  \BibitemOpen
  \bibfield  {author} {\bibinfo {author} {\bibfnamefont {A.}~\bibnamefont
  {Krizhevsky}}, \bibinfo {author} {\bibfnamefont {I.}~\bibnamefont
  {Sutskever}},\ and\ \bibinfo {author} {\bibfnamefont {G.~E.}\ \bibnamefont
  {Hinton}},\ }\bibfield  {title} {\bibinfo {title} {Imagenet classification
  with deep convolutional neural networks},\ }in\ \href@noop {} {\emph
  {\bibinfo {booktitle} {Advances in neural information processing systems}}}\
  (\bibinfo {year} {2012})\ pp.\ \bibinfo {pages} {1097--1105}\BibitemShut
  {NoStop}%
\bibitem [{\citenamefont {Antipov}\ \emph {et~al.}(2015)\citenamefont
  {Antipov}, \citenamefont {Berrani}, \citenamefont {Ruchaud},\ and\
  \citenamefont {Dugelay}}]{antipov2015learned}%
  \BibitemOpen
  \bibfield  {author} {\bibinfo {author} {\bibfnamefont {G.}~\bibnamefont
  {Antipov}}, \bibinfo {author} {\bibfnamefont {S.-A.}\ \bibnamefont
  {Berrani}}, \bibinfo {author} {\bibfnamefont {N.}~\bibnamefont {Ruchaud}},\
  and\ \bibinfo {author} {\bibfnamefont {J.-L.}\ \bibnamefont {Dugelay}},\
  }\bibfield  {title} {\bibinfo {title} {Learned vs. hand-crafted features for
  pedestrian gender recognition},\ }in\ \href@noop {} {\emph {\bibinfo
  {booktitle} {Proceedings of the 23rd ACM international conference on
  Multimedia}}}\ (\bibinfo {organization} {ACM},\ \bibinfo {year} {2015})\ pp.\
  \bibinfo {pages} {1263--1266}\BibitemShut {NoStop}%
\bibitem [{\citenamefont {Liang}\ \emph {et~al.}(2017)\citenamefont {Liang},
  \citenamefont {Sun}, \citenamefont {Sun},\ and\ \citenamefont
  {Gao}}]{liang2017text}%
  \BibitemOpen
  \bibfield  {author} {\bibinfo {author} {\bibfnamefont {H.}~\bibnamefont
  {Liang}}, \bibinfo {author} {\bibfnamefont {X.}~\bibnamefont {Sun}}, \bibinfo
  {author} {\bibfnamefont {Y.}~\bibnamefont {Sun}},\ and\ \bibinfo {author}
  {\bibfnamefont {Y.}~\bibnamefont {Gao}},\ }\bibfield  {title} {\bibinfo
  {title} {{T}ext feature extraction based on deep learning: a review},\ }\href
  {https://doi.org/10.1186/s13638-017-0993-1} {\bibfield  {journal} {\bibinfo
  {journal} {{EURASIP} Journal on Wireless Communications and Networking}\
  }\textbf {\bibinfo {volume} {2017}},\ \bibinfo {pages} {1186} (\bibinfo
  {year} {2017})}\BibitemShut {NoStop}%
\bibitem [{\citenamefont {Wang}(2016)}]{Wang2016unsupervised}%
  \BibitemOpen
  \bibfield  {author} {\bibinfo {author} {\bibfnamefont {L.}~\bibnamefont
  {Wang}},\ }\bibfield  {title} {\bibinfo {title} {Discovering phase
  transitions with unsupervised learning},\ }\href
  {https://doi.org/10.1103/PhysRevB.94.195105} {\bibfield  {journal} {\bibinfo
  {journal} {Phys. Rev. B}\ }\textbf {\bibinfo {volume} {94}},\ \bibinfo
  {pages} {195105} (\bibinfo {year} {2016})}\BibitemShut {NoStop}%
\bibitem [{\citenamefont {Wetzel}(2017)}]{wetzel2017unsupervised}%
  \BibitemOpen
  \bibfield  {author} {\bibinfo {author} {\bibfnamefont {S.~J.}\ \bibnamefont
  {Wetzel}},\ }\bibfield  {title} {\bibinfo {title} {Unsupervised learning of
  phase transitions: From principal component analysis to variational
  autoencoders},\ }\href {https://doi.org/10.1103/PhysRevE.96.022140}
  {\bibfield  {journal} {\bibinfo  {journal} {Phys. Rev. E}\ }\textbf {\bibinfo
  {volume} {96}},\ \bibinfo {pages} {022140} (\bibinfo {year}
  {2017})}\BibitemShut {NoStop}%
\bibitem [{\citenamefont {Ch'ng}\ \emph {et~al.}(2018)\citenamefont {Ch'ng},
  \citenamefont {Vazquez},\ and\ \citenamefont {Khatami}}]{ch2018unsupervised}%
  \BibitemOpen
  \bibfield  {author} {\bibinfo {author} {\bibfnamefont {K.}~\bibnamefont
  {Ch'ng}}, \bibinfo {author} {\bibfnamefont {N.}~\bibnamefont {Vazquez}},\
  and\ \bibinfo {author} {\bibfnamefont {E.}~\bibnamefont {Khatami}},\
  }\bibfield  {title} {\bibinfo {title} {Unsupervised machine learning account
  of magnetic transitions in the {Hubbard} model},\ }\href
  {https://doi.org/10.1103/PhysRevE.97.013306} {\bibfield  {journal} {\bibinfo
  {journal} {Phys. Rev. E}\ }\textbf {\bibinfo {volume} {97}},\ \bibinfo
  {pages} {013306} (\bibinfo {year} {2018})}\BibitemShut {NoStop}%
\bibitem [{\citenamefont {Alexandrou}\ \emph {et~al.}(2020)\citenamefont
  {Alexandrou}, \citenamefont {Athenodorou}, \citenamefont {Chrysostomou},\
  and\ \citenamefont {Paul}}]{alexandrou2020critical}%
  \BibitemOpen
  \bibfield  {author} {\bibinfo {author} {\bibfnamefont {C.}~\bibnamefont
  {Alexandrou}}, \bibinfo {author} {\bibfnamefont {A.}~\bibnamefont
  {Athenodorou}}, \bibinfo {author} {\bibfnamefont {C.}~\bibnamefont
  {Chrysostomou}},\ and\ \bibinfo {author} {\bibfnamefont {S.}~\bibnamefont
  {Paul}},\ }\bibfield  {title} {\bibinfo {title} {{T}he critical temperature
  of the 2{D}-ising model through deep learning autoencoders},\ }\href
  {https://doi.org/10.1140/epjb/e2020-100506-5} {\bibfield  {journal} {\bibinfo
   {journal} {The European Physical Journal B}\ }\textbf {\bibinfo {volume}
  {93}},\ \bibinfo {pages} {1140} (\bibinfo {year} {2020})}\BibitemShut
  {NoStop}%
\bibitem [{\citenamefont {Yevick}(2021)}]{yevick2021variational}%
  \BibitemOpen
  \bibfield  {author} {\bibinfo {author} {\bibfnamefont {D.}~\bibnamefont
  {Yevick}},\ }\bibfield  {title} {\bibinfo {title} {Variational autoencoder
  analysis of ising model statistical distributions and phase transitions},\
  }\href {https://arxiv.org/abs/2104.06368} {\bibfield  {journal} {\bibinfo
  {journal} {arXiv:2104.06368}\ } (\bibinfo {year} {2021})}\BibitemShut
  {NoStop}%
\bibitem [{\citenamefont {Hinton}\ and\ \citenamefont
  {Salakhutdinov}(2006)}]{hinton2006reducing}%
  \BibitemOpen
  \bibfield  {author} {\bibinfo {author} {\bibfnamefont {G.~E.}\ \bibnamefont
  {Hinton}}\ and\ \bibinfo {author} {\bibfnamefont {R.~R.}\ \bibnamefont
  {Salakhutdinov}},\ }\bibfield  {title} {\bibinfo {title} {Reducing the
  dimensionality of data with neural networks},\ }\href
  {http://dx.doi.org/10.1126/science.1127647} {\bibfield  {journal} {\bibinfo
  {journal} {Science}\ }\textbf {\bibinfo {volume} {313}},\ \bibinfo {pages}
  {504} (\bibinfo {year} {2006})}\BibitemShut {NoStop}%
\bibitem [{\citenamefont {Kingma}\ and\ \citenamefont
  {Welling}(2014)}]{kingma2014auto}%
  \BibitemOpen
  \bibfield  {author} {\bibinfo {author} {\bibfnamefont {D.}~\bibnamefont
  {Kingma}}\ and\ \bibinfo {author} {\bibfnamefont {M.}~\bibnamefont
  {Welling}},\ }\bibfield  {title} {\bibinfo {title} {Auto-encoding variational
  {B}ayes},\ }in\ \href@noop {} {\emph {\bibinfo {booktitle} {International
  Conference on Learning Representations}}}\ (\bibinfo {year}
  {2014})\BibitemShut {NoStop}%
\bibitem [{\citenamefont {Georgii}(2011)}]{georgii2011gibbs}%
  \BibitemOpen
  \bibfield  {author} {\bibinfo {author} {\bibfnamefont {H.-O.}\ \bibnamefont
  {Georgii}},\ }\href@noop {} {\emph {\bibinfo {title} {Gibbs measures and
  phase transitions}}}\ (\bibinfo  {publisher} {de Gruyter},\ \bibinfo {year}
  {2011})\BibitemShut {NoStop}%
\bibitem [{\citenamefont {Onsager}(1944)}]{onsager1944crystal}%
  \BibitemOpen
  \bibfield  {author} {\bibinfo {author} {\bibfnamefont {L.}~\bibnamefont
  {Onsager}},\ }\bibfield  {title} {\bibinfo {title} {{Crystal Statistics. I. A
  Two-Dimensional Model with an Order-Disorder Transition}},\ }\href
  {https://doi.org/10.1103/PhysRev.65.117} {\bibfield  {journal} {\bibinfo
  {journal} {Phys. Rev.}\ }\textbf {\bibinfo {volume} {65}},\ \bibinfo {pages}
  {117} (\bibinfo {year} {1944})}\BibitemShut {NoStop}%
\bibitem [{\citenamefont {Smith}\ \emph {et~al.}(2018)\citenamefont {Smith},
  \citenamefont {Kindermans}, \citenamefont {Ying},\ and\ \citenamefont
  {Le}}]{smith2018dont}%
  \BibitemOpen
  \bibfield  {author} {\bibinfo {author} {\bibfnamefont {S.~L.}\ \bibnamefont
  {Smith}}, \bibinfo {author} {\bibfnamefont {P.}~\bibnamefont {Kindermans}},
  \bibinfo {author} {\bibfnamefont {C.}~\bibnamefont {Ying}},\ and\ \bibinfo
  {author} {\bibfnamefont {Q.~V.}\ \bibnamefont {Le}},\ }\bibfield  {title}
  {\bibinfo {title} {Don't decay the learning rate, increase the batch size},\
  }in\ \href@noop {} {\emph {\bibinfo {booktitle} {6th International Conference
  on Learning Representations, {ICLR} 2018, Vancouver, BC, Canada, April 30 -
  May 3, 2018, Conference Track Proceedings}}}\ (\bibinfo  {publisher}
  {OpenReview.net},\ \bibinfo {year} {2018})\BibitemShut {NoStop}%
\bibitem [{\citenamefont {Binder}\ \emph {et~al.}(1993)\citenamefont {Binder},
  \citenamefont {Heermann}, \citenamefont {Roelofs}, \citenamefont
  {Mallinckrodt},\ and\ \citenamefont {McKay}}]{binder1993monte}%
  \BibitemOpen
  \bibfield  {author} {\bibinfo {author} {\bibfnamefont {K.}~\bibnamefont
  {Binder}}, \bibinfo {author} {\bibfnamefont {D.}~\bibnamefont {Heermann}},
  \bibinfo {author} {\bibfnamefont {L.}~\bibnamefont {Roelofs}}, \bibinfo
  {author} {\bibfnamefont {A.~J.}\ \bibnamefont {Mallinckrodt}},\ and\ \bibinfo
  {author} {\bibfnamefont {S.}~\bibnamefont {McKay}},\ }\bibfield  {title}
  {\bibinfo {title} {{Monte Carlo} simulation in statistical physics},\
  }\href@noop {} {\bibfield  {journal} {\bibinfo  {journal} {Computers in
  Physics}\ }\textbf {\bibinfo {volume} {7}},\ \bibinfo {pages} {156} (\bibinfo
  {year} {1993})}\BibitemShut {NoStop}%
\bibitem [{\citenamefont {Young}(2015)}]{young2015everything}%
  \BibitemOpen
  \bibfield  {author} {\bibinfo {author} {\bibfnamefont {P.}~\bibnamefont
  {Young}},\ }\href@noop {} {\emph {\bibinfo {title} {Everything you wanted to
  know about data analysis and fitting but were afraid to ask}}}\ (\bibinfo
  {publisher} {Springer},\ \bibinfo {year} {2015})\BibitemShut {NoStop}%
\bibitem [{\citenamefont {Coniglio}\ \emph {et~al.}(1989)\citenamefont
  {Coniglio}, \citenamefont {de~Liberto}, \citenamefont {Monroy},\ and\
  \citenamefont {Peruggi}}]{coniglio1989exact}%
  \BibitemOpen
  \bibfield  {author} {\bibinfo {author} {\bibfnamefont {A.}~\bibnamefont
  {Coniglio}}, \bibinfo {author} {\bibfnamefont {F.}~\bibnamefont
  {de~Liberto}}, \bibinfo {author} {\bibfnamefont {G.}~\bibnamefont {Monroy}},\
  and\ \bibinfo {author} {\bibfnamefont {F.}~\bibnamefont {Peruggi}},\
  }\bibfield  {title} {\bibinfo {title} {Exact relations between droplets and
  thermal fluctuations in external field},\ }\href@noop {} {\bibfield
  {journal} {\bibinfo  {journal} {Journal of Physics A: Mathematical and
  General}\ }\textbf {\bibinfo {volume} {22}},\ \bibinfo {pages} {L837}
  (\bibinfo {year} {1989})}\BibitemShut {NoStop}%
\bibitem [{GAP()}]{anon2021gap}%
  \BibitemOpen
  GAP,\ \href@noop {} {\bibinfo {title} {{GAP} {\textendash} {G}roups,
  {A}lgorithms, and {P}rogramming, {V}ersion 4.11.1}},\ \bibinfo {howpublished}
  {\href{https://www.gap-system.org}{\texttt{https://www.gap-system.org}}}
  (\bibinfo {year} {2021})\BibitemShut {NoStop}%
\bibitem [{\citenamefont {Agrawal}\ and\ \citenamefont
  {Ostrowski}(2022)}]{agrawal2022classification}%
  \BibitemOpen
  \bibfield  {author} {\bibinfo {author} {\bibfnamefont {D.}~\bibnamefont
  {Agrawal}}\ and\ \bibinfo {author} {\bibfnamefont {J.}~\bibnamefont
  {Ostrowski}},\ }\bibfield  {title} {\bibinfo {title} {A classification of
  {$G$}-invariant shallow neural networks},\ }\href@noop {} {\bibfield
  {journal} {\bibinfo  {journal} {Advances in Neural Information Processing
  Systems}\ }\textbf {\bibinfo {volume} {35}} (\bibinfo {year}
  {2022})}\BibitemShut {NoStop}%
\bibitem [{\citenamefont {Liu}\ \emph {et~al.}(2017)\citenamefont {Liu},
  \citenamefont {Qi}, \citenamefont {Meng},\ and\ \citenamefont
  {Fu}}]{PhysRevB.95.041101}%
  \BibitemOpen
  \bibfield  {author} {\bibinfo {author} {\bibfnamefont {J.}~\bibnamefont
  {Liu}}, \bibinfo {author} {\bibfnamefont {Y.}~\bibnamefont {Qi}}, \bibinfo
  {author} {\bibfnamefont {Z.~Y.}\ \bibnamefont {Meng}},\ and\ \bibinfo
  {author} {\bibfnamefont {L.}~\bibnamefont {Fu}},\ }\bibfield  {title}
  {\bibinfo {title} {Self-learning {Monte Carlo} method},\ }\href
  {https://doi.org/10.1103/PhysRevB.95.041101} {\bibfield  {journal} {\bibinfo
  {journal} {Phys. Rev. B}\ }\textbf {\bibinfo {volume} {95}},\ \bibinfo
  {pages} {041101} (\bibinfo {year} {2017})}\BibitemShut {NoStop}%
\bibitem [{\citenamefont {Shen}\ \emph {et~al.}(2018)\citenamefont {Shen},
  \citenamefont {Liu},\ and\ \citenamefont {Fu}}]{ShenPRB2018}%
  \BibitemOpen
  \bibfield  {author} {\bibinfo {author} {\bibfnamefont {H.}~\bibnamefont
  {Shen}}, \bibinfo {author} {\bibfnamefont {J.}~\bibnamefont {Liu}},\ and\
  \bibinfo {author} {\bibfnamefont {L.}~\bibnamefont {Fu}},\ }\bibfield
  {title} {\bibinfo {title} {Self-learning {Monte Carlo} with deep neural
  networks},\ }\href {https://doi.org/10.1103/PhysRevB.97.205140} {\bibfield
  {journal} {\bibinfo  {journal} {Phys. Rev. B}\ }\textbf {\bibinfo {volume}
  {97}},\ \bibinfo {pages} {205140} (\bibinfo {year} {2018})}\BibitemShut
  {NoStop}%
\bibitem [{\citenamefont {Li}\ \emph {et~al.}(2019)\citenamefont {Li},
  \citenamefont {Dee}, \citenamefont {Khatami},\ and\ \citenamefont
  {Johnston}}]{Li2019ann}%
  \BibitemOpen
  \bibfield  {author} {\bibinfo {author} {\bibfnamefont {S.}~\bibnamefont
  {Li}}, \bibinfo {author} {\bibfnamefont {P.~M.}\ \bibnamefont {Dee}},
  \bibinfo {author} {\bibfnamefont {E.}~\bibnamefont {Khatami}},\ and\ \bibinfo
  {author} {\bibfnamefont {S.}~\bibnamefont {Johnston}},\ }\bibfield  {title}
  {\bibinfo {title} {Accelerating lattice quantum {Monte Carlo} simulations
  using artificial neural networks: Application to the {Holstein} model},\
  }\href {https://doi.org/10.1103/PhysRevB.100.020302} {\bibfield  {journal}
  {\bibinfo  {journal} {Phys. Rev. B}\ }\textbf {\bibinfo {volume} {100}},\
  \bibinfo {pages} {020302} (\bibinfo {year} {2019})}\BibitemShut {NoStop}%
\bibitem [{\citenamefont {Albergo}\ \emph {et~al.}(2019)\citenamefont
  {Albergo}, \citenamefont {Kanwar},\ and\ \citenamefont
  {Shanahan}}]{AlbergoPRD2019}%
  \BibitemOpen
  \bibfield  {author} {\bibinfo {author} {\bibfnamefont {M.~S.}\ \bibnamefont
  {Albergo}}, \bibinfo {author} {\bibfnamefont {G.}~\bibnamefont {Kanwar}},\
  and\ \bibinfo {author} {\bibfnamefont {P.~E.}\ \bibnamefont {Shanahan}},\
  }\bibfield  {title} {\bibinfo {title} {Flow-based generative models for
  {Markov} chain {Monte Carlo} in lattice field theory},\ }\href
  {https://doi.org/10.1103/PhysRevD.100.034515} {\bibfield  {journal} {\bibinfo
   {journal} {Phys. Rev. D}\ }\textbf {\bibinfo {volume} {100}},\ \bibinfo
  {pages} {034515} (\bibinfo {year} {2019})}\BibitemShut {NoStop}%
\bibitem [{\citenamefont {Nagai}\ \emph {et~al.}(2020)\citenamefont {Nagai},
  \citenamefont {Okumura},\ and\ \citenamefont {Tanaka}}]{PhysRevB.101.115111}%
  \BibitemOpen
  \bibfield  {author} {\bibinfo {author} {\bibfnamefont {Y.}~\bibnamefont
  {Nagai}}, \bibinfo {author} {\bibfnamefont {M.}~\bibnamefont {Okumura}},\
  and\ \bibinfo {author} {\bibfnamefont {A.}~\bibnamefont {Tanaka}},\
  }\bibfield  {title} {\bibinfo {title} {Self-learning {Monte Carlo} method
  with behler-parrinello neural networks},\ }\href
  {https://doi.org/10.1103/PhysRevB.101.115111} {\bibfield  {journal} {\bibinfo
   {journal} {Phys. Rev. B}\ }\textbf {\bibinfo {volume} {101}},\ \bibinfo
  {pages} {115111} (\bibinfo {year} {2020})}\BibitemShut {NoStop}%
\bibitem [{\citenamefont {Chen}\ \emph {et~al.}(2018)\citenamefont {Chen},
  \citenamefont {Xu}, \citenamefont {Liu}, \citenamefont {Batrouni},
  \citenamefont {Scalettar},\ and\ \citenamefont {Meng}}]{PhysRevB.98.041102}%
  \BibitemOpen
  \bibfield  {author} {\bibinfo {author} {\bibfnamefont {C.}~\bibnamefont
  {Chen}}, \bibinfo {author} {\bibfnamefont {X.~Y.}\ \bibnamefont {Xu}},
  \bibinfo {author} {\bibfnamefont {J.}~\bibnamefont {Liu}}, \bibinfo {author}
  {\bibfnamefont {G.}~\bibnamefont {Batrouni}}, \bibinfo {author}
  {\bibfnamefont {R.}~\bibnamefont {Scalettar}},\ and\ \bibinfo {author}
  {\bibfnamefont {Z.~Y.}\ \bibnamefont {Meng}},\ }\bibfield  {title} {\bibinfo
  {title} {Symmetry-enforced self-learning {Monte Carlo} method applied to the
  {Holstein} model},\ }\href {https://doi.org/10.1103/PhysRevB.98.041102}
  {\bibfield  {journal} {\bibinfo  {journal} {Phys. Rev. B}\ }\textbf {\bibinfo
  {volume} {98}},\ \bibinfo {pages} {041102} (\bibinfo {year}
  {2018})}\BibitemShut {NoStop}%
\bibitem [{\citenamefont {Vandenberghe}(2010)}]{vandenberghe2010cvxopt}%
  \BibitemOpen
  \bibfield  {author} {\bibinfo {author} {\bibfnamefont {L.}~\bibnamefont
  {Vandenberghe}},\ }\bibfield  {title} {\bibinfo {title} {The cvxopt linear
  and quadratic cone program solvers},\ }\href@noop {} {\bibfield  {journal}
  {\bibinfo  {journal} {Online:
  \url{http://cvxopt.org/documentation/coneprog.pdf}}\ } (\bibinfo {year}
  {2010})}\BibitemShut {NoStop}%
\end{thebibliography}%

\end{document}